\documentclass[english, thm-restate]{lipics-v2021}

\hideLIPIcs  

\newtheorem{rst-lemma}[theorem]{Lemma}

\newcommand*\ca{\includegraphics[height=3mm]{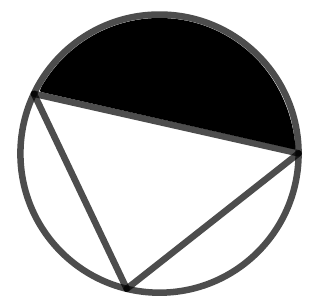}}
\newcommand{\R}{\mathbb{R}}

\newcommand{\Z}{\mathbb{Z}}

\newcommand{\kk}{\texorpdfstring{$k$ }{k}}

\newcommand{\mS}{\mathcal{S}}
\newcommand{\mR}{\mathcal{R}}
\newcommand{\RR}{\widetilde{\mathcal{R}}}

\DeclareMathOperator{\conv}{conv}
\DeclareMathOperator{\Err}{Err}

\title{Minimum-Error Triangulations for Sea Surface Reconstruction} 

\author{Anna Arutyunova}
{Institute for Computer Science, University of Bonn, Germany}
{}{}{}

\author{Anne Driemel}
{Hausdorff Center for Mathematics, University of Bonn, Germany}
{}{}{}

\author{{Jan-Henrik} Haunert}
{Institute of Geodesy and Geoinformation, University of Bonn, Germany}
{}{}{}

\author{Herman Haverkort}
{Institute for Computer Science, University of Bonn, Germany}
{}{}{}

\author{J\"urgen Kusche}
{Institute of Geodesy and Geoinformation, University of Bonn, Germany}
{}{}{}

\author{Elmar Langetepe}
{Institute for Computer Science, University of Bonn, Germany}
{}{}{}

\author{Philip Mayer}
{Institute for Computer Science, University of Bonn, Germany}
{}{}{}

\author{Petra Mutzel}
{Institute for Computer Science, University of Bonn, Germany}
{}{}{}

\author{Heiko R\"oglin}
{Institute for Computer Science, University of Bonn, Germany}
{}{}{}

\authorrunning{A. Arutyunova et al.}

\Copyright{Anna Arutyunova, Anne Driemel, {Jan-Henrik} Haunert, Herman Haverkort, J\"urgen Kusche, Elmar Langetepe, Philip Mayer, Petra Mutzel and Heiko R\"oglin} 



\category{} 

\relatedversion{} 

\supplement{The code and information about the data acquisition is available at: \url{https://github.com/PhilipMayer94/dynamic-programming-for-min-error-triangulations}.
}

\acknowledgements{We thank the anonymous reviewers for their insightful comments and suggestions.}

\funding{This work is funded by the Deutsche Forschungsgemeinschaft (DFG, German Research Foundation) under Germany's Excellence Strategy -- EXC-2047/1 -- 390685813 and DFG grant RO 5439/1-1.}

\nolinenumbers 

\EventEditors{Xavier Goaoc and Michael Kerber}
\EventNoEds{2}
\EventLongTitle{38th International Symposium on Computational Geometry (SoCG 2022)}
\EventShortTitle{SoCG 2022}
\EventAcronym{SoCG}
\EventYear{2022}
\EventDate{June 7--10, 2022}
\EventLocation{Berlin, Germany}
\EventLogo{}
\SeriesVolume{224}
\ArticleNo{XX}  

\begin{document}

\maketitle

\begin{abstract}
We apply state-of-the-art computational geometry methods to the problem of reconstructing a time-varying sea surface from tide gauge records. 
Our work builds on a recent article by Nitzke et al.~(Computers \& Geosciences, 157:104920, 2021) who have suggested to learn a triangulation $D$ of a given set of tide gauge stations. The objective is to minimize the misfit of the piecewise linear surface induced by $D$ to a reference surface that has been acquired with satellite altimetry.
The authors restricted their search to k-order Delaunay ($k$-OD) triangulations and used an integer linear program in order to solve the resulting optimization problem. 

In geometric terms, the input to our problem consists of two sets of points in $\R^2$ with elevations: a set $\mathcal{S}$ that is to be triangulated, and a set $\mathcal{R}$ of reference points. Intuitively, we define the error of a triangulation as the average vertical distance of a point in $\mathcal{R}$ to the triangulated surface that is obtained by interpolating elevations of $\mathcal{S}$ linearly in each triangle. Our goal is to find the triangulation of $\mathcal{S}$ that has minimum error with respect to $\mathcal{R}$.

In our work, we prove that the minimum-error triangulation problem is NP-hard and cannot be approximated within any multiplicative factor in polynomial time unless $P=NP$. 
At the same time we show that the problem instances that occur in our application (considering sea level data from several hundreds of tide gauge stations worldwide) can be solved relatively fast using dynamic programming when restricted to $k$-OD triangulations for $k\le 7$.  
In particular, instances for which the number of connected components of the so-called $k$-OD fixed-edge graph is small can be solved within few seconds.
\end{abstract}

\section{Introduction}\label{sec:Introduction}
Reconstructing the sea level for the past is of paramount importance for understanding the influences of climate change.
Two types of observational data are often used for this task: (1) data from tide gauge stations, which are usually located at the sea shore, and (2) gridded altimeter data acquired from satellites.
The tide gauge data is available from the 18th century from stations that are sparsely distributed globally
(e.g., the RLR database given by the PSMSL contains 1\,548 stations).
The gridded altimeter data, which has been acquired since 1993, admits much more accurate reconstructions of the sea surface for the last 29 years.
We build on the work by Nitzke et al.~\cite{NitzkeNFKH2021}, who suggested an approach 
for combining these two types of data using integer linear programming techniques.
The approach is to learn a plausible triangulation of the tide gauge stations for an epoch $E$ for which the altimeter data is available, 
and then use that triangulation to reconstruct the sea surface in another epoch, where gauge data is available, but no altimeter data. 
Given the gauge and altimeter data for $E$, the task is to compute a minimum-error triangulation of the gauge stations, 
that is, a triangulation that minimizes the sum of squared differences between the reference (altimeter) data and the piecewise linear surface defined with the triangulation.

For piecewise linear surfaces, Delaunay triangulations are often chosen, since they have many desirable properties.
However, they are unique and so they do not have potential for optimization. 
On the other hand, computing a minimum-error triangulation among the set of all triangulations can lead to badly shaped triangles, 
which can cause large interpolation errors for epochs other than the training epoch.
Therefore, Nitzke et al.~\cite{NitzkeNFKH2021} suggested computing a triangulation of minimum error among all \emph{$k$-order} Delaunay ($k$-OD) triangulations~\cite{GudmundssonHK2002}. 
A $k$-OD triangulation consists of triangles with up to $k$ points inside each triangle's circumcircle ($k=0$ corresponds to Delaunay triangles). 
This creates room for optimization while ensuring (reasonably) well-shaped triangles.
Moreover, restricting the solution to the set of $k$-order Delaunay triangulations has computational advantages. 
\begin{figure}[t]
\begin{center}
	\includegraphics[height=40mm]{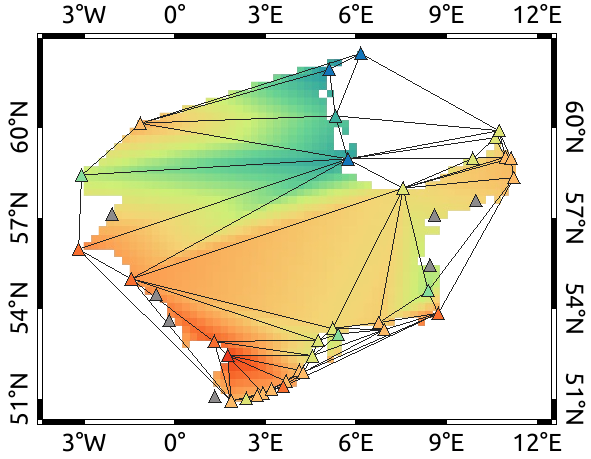} \hspace{10mm}
	\includegraphics[height=34mm]{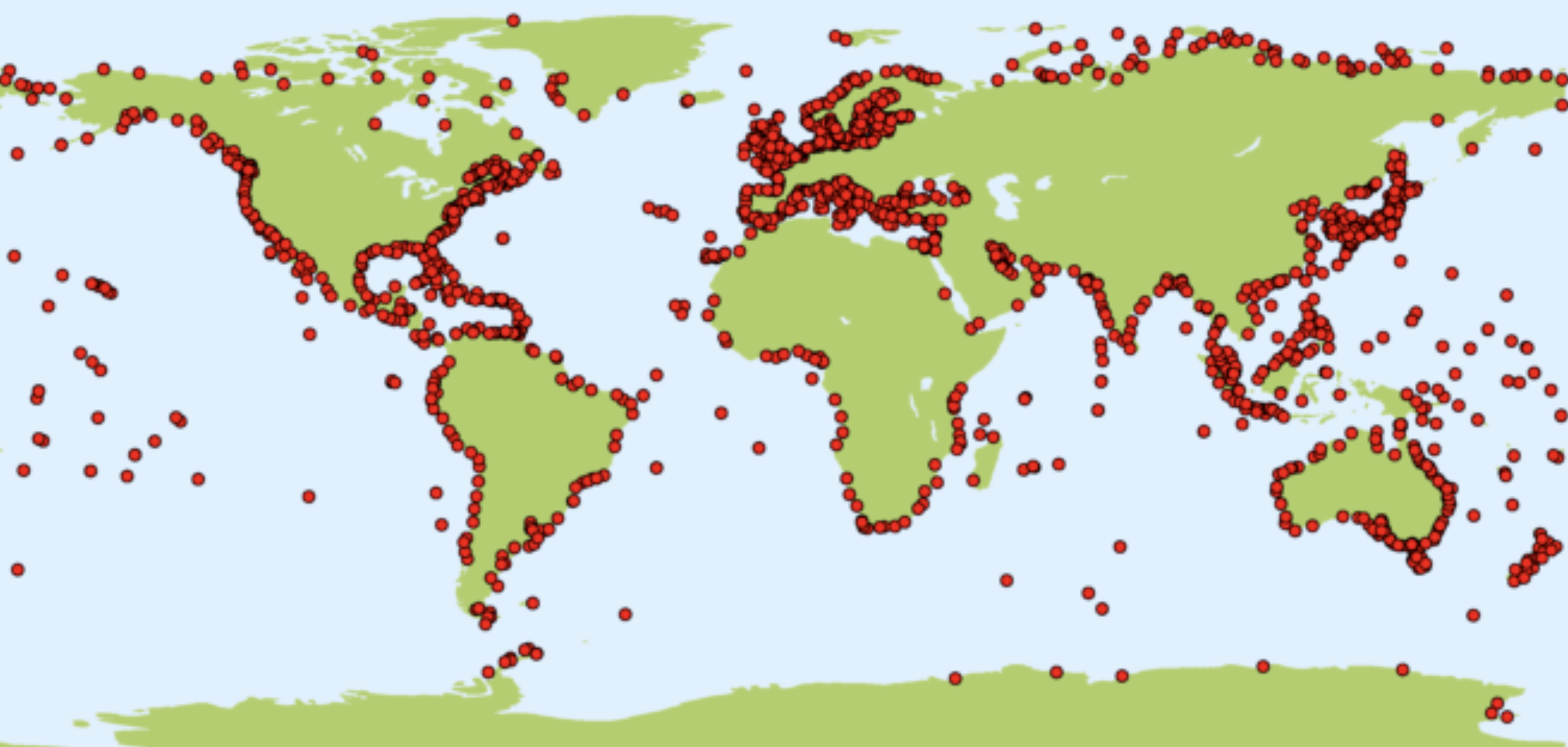} 
	\caption{Left: A minimum-error triangulation of the North Sea data (June 2010) with 34 tide gauge stations computed with the approach in~\cite{NitzkeNFKH2021}. Right: Locations of all 
	tide gauge stations in the PSMSL database (\url{www.psmsl.org/products/data_coverage}) }
		\label{fig:NorthseaILP}
	\end{center}
\end{figure}
Nitzke et al.~\cite{NitzkeNFKH2021} modeled their approach as an integer linear program (ILP) and evaluated it on the North Sea dataset with up to 40 stations and $k\le 3$, whose locations are projected on the plane; see Figure~\ref{fig:NorthseaILP}. 
The evaluation showed that the $k$-OD minimum-error triangulation is substantially more effective than the method based on the Delaunay triangulation suggested in~\cite{OlivieriS2016} for Sea Surface Anomaly reconstructions of up to 19 years back in time. 

The aim of our work is to speed up the above approach using computational geometry in order to apply it to areas of global extent (instances with up to 800 tide gauge stations). 

\subparagraph*{Our contribution:} 
\begin{itemize}
\item We first show that 
the minimum-error triangulation problem is NP-hard and that it is even NP-hard to approximate an optimal solution.
\item We discuss an alternative optimization approach to the ILP-based one by Nitzke et al.~\cite{NitzkeNFKH2021}. Our approach is based on the dynamic programming (DP) algorithm by Silveira and Van Krefeld~\cite{SilveiraK2009}. The runtime of the DP algorithm depends on the Delaunay order $k$; since we are only interested in small orders, we are able to calculate minimum-error order-$k$ Delaunay triangulations for the datasets given by the sea surface reconstruction problem. 
\item The algorithm's runtime depends on a subgraph of the Delaunay triangulation, which we call the order-$k$ fixed-edge graph. 
 It is known that for order $1$ the fixed-edge graph is connected \cite{GudmundssonHK2002}. We investigate the fixed-edge graph for orders $k=2,3$. We show that for $k=2$ no vertex can be isolated and give an example where the fixed-edge graph is not connected. For $k\geq 3$ we give an example where $\lfloor \frac{n}{6}\rfloor$ connected components are inside a face of the fixed-edge graph, which implies exponential runtime for the algorithm. This complements the observations by Silveira et al.\ given in \cite{SilveiraK2009}.
\item We perform experiments with different projections of the tide gauge dataset to analyze the structure of the fixed-edge graphs for a real-world dataset. Our experiments confirm the assumption by Silveira and Van Krefeld~\cite{SilveiraK2009} that the DP algorithm can be used to solve practical problems for medium-sized datasets, if the order is small ($k\leq 7$).
\item Lastly, we perform the reconstruction task that was given in \cite{NitzkeNFKH2021} for the global dataset. Our evaluation shows that on the used global dataset with up to 800 stations the quality improves with growing $k$, which contrasts with the findings in \cite{NitzkeNFKH2021} on the local North Sea dataset with about 40 stations, where $k=2$ consistently delivered the best reconstructions.
\end{itemize}

The paper is organized as follows. First, we outline the formal definitions of the triangulation problem in Section~\ref{sec:TheTriangulationProblem}. After that, we discuss related works in Section~\ref{sec:RelatedWorks}. In Section~\ref{sec:NPHardness} we present our NP-hardness proof for the minimum-error triangulation problem. Section~\ref{sec:HODOptimization} presents the DP algorithm by Silveira  et al.~\cite{SilveiraK2009} and discusses our findings regarding the fixed-edge graphs. In Section~\ref{sec:Experiments} we provide the application of the DP algorithm to the sea surface reconstruction problem. Finally, we give our conclusion in Section~\ref{sec:Conclusion}. 

\section{The triangulation problem}\label{sec:TheTriangulationProblem}
Let $\mathcal{S}\subset\mathbb{R}^2$ be a set of $n$ points and $f\colon\mathcal{S} \rightarrow \mathbb{R}$. We call $\mathcal{S}$ the set of triangulation points and $f(s)$ the measurement value of $s\in\mathcal{S}$. Additionally, we are given a set $\mathcal{R}\subset\text{conv}(\mathcal{S})$ of $m$ points and a function $h \colon\mathcal{R} \rightarrow \mathbb{R}$. We refer to $\mathcal{R}$ as the set of reference points and to $h(r)$ as the reference value of $r\in\mathcal{R}$.

A triangulation $D$ of $\mathcal{S}$ is given by a maximal set of non-crossing straight-line edges between points in $\mathcal{S}$. We can extend the function $f$ on the points in $\text{conv}(\mathcal{S})$ by linearly interpolating $f$ in every triangle. In this way we obtain a piece-wise linear function $s_D\colon \text{conv}(\mathcal S)\rightarrow\mathbb{R}$.\\
The \emph{minimum-error triangulation problem} asks for a triangulation $D$ of $\mathcal{S}$ that minimizes the squared error between the reference values and the interpolation, i.e.,
\begin{align*}
\text{Err}_D(\mathcal{R})=\sum_{r\in\mathcal{R}}(s_D(r)-h(r))^2.
\end{align*}
For the dynamic programming algorithm used in our approach and discussed in Section~\ref{sec:HODOptimization},  
we transform  the minimum-error triangulation problem to the \emph{minimum triangle-weighted triangulation problem}.
Let $\mathfrak{T}$ be the set of all $O(n^3)$  possible triangles that may be used in any triangulation of $\mathcal{S}$. Then we can assign the weight
\begin{align*}
w_T(\mathcal{R})=\sum_{r\in T}(s_T(r)-h(r))^2
\end{align*}
to every triangle $T\in\mathfrak{T}$, where $s_T$ is the linear interpolation given by the triangle $T$. If we assume that no reference point lies on any triangulation edge, we get
\begin{align*}
\text{Err}_D(\mathcal{R})=\sum_{r\in\mathcal{R}}(s_D(r)-h(r))^2=\sum_{T\in D}\sum_{r\in T}(s_T(r)-h(r))^2=\sum_{T\in D}w_T(\mathcal{R}).
\end{align*}
To get rid of the previous assumption we assign points that lie on an edge $\overline{uv}$ only to the triangles left of $\overrightarrow{uv}$. Points coinciding with triangulation points can be ignored.

Using these weights our cost function becomes a decomposable measure as discussed by Bern and Eppstein in \cite{BernE1995}. Broadly speaking, decomposable measures are all measures that easily allow computation using dynamic programming approaches for triangulations.

\section{Related works}\label{sec:RelatedWorks}

\subparagraph*{Sea level reconstruction:} Conventional methods for 
sea level reconstruction use global base functions (empirical orthogonal base functions) which are \emph{learned} within the altimeter decades \cite{Church+2004}.
Olivieri and Spada suggested the first triangulation-based reconstruction approach \cite{OlivieriS2016}.
However, this approach does not use the altimeter data in any way and generates a Delaunay triangulation of the station data.  
Nevertheless, the resulting reconstruction of the sea surface was quite promising.
The approach suggested by Nitzke et al.~\cite{NitzkeNFKH2021} marries the conventional thinking and the triangulation method.
The authors proposed the use of data-dependent triangulations which were introduced in \cite{DynLR1990} by Dyn, Levin and Rippa. 
The particular focus of Nitzke et al.\ were the minimum-error triangulations. Since they also want to reconstruct the sea level in the pre-altimetry era, they formulate the reconstruction as a learning task and use higher-order Delaunay constraints, which were introduced in \cite{GudmundssonHK2002} by Gudmundsson, Hammar and van Kreveld, as regularizer.

\subparagraph*{Triangulating point sets:} Triangulating point sets in the plane is a fundamental task of computational geometry. It is of high relevance for data interpolation and surface modeling tasks, where for every data point a data value (or height) is given in addition to the point’s two coordinates. The Delaunay triangulation is most often applied as it optimizes several criteria and can be computed efficiently. In particular, it maximizes the minimum angle among all the angles of all the triangles. 
\emph{Data-dependent triangulations} have been defined in \cite{DynLR1990} as triangulations that are computed under consideration of the data values.  As optimization criteria the authors have considered (1) smoothness criteria, (2) criteria based on three-dimensional properties of the triangles, (3) variational criteria, and (4) the minimum-error criterion, which is optimized by the previously defined minimum-error triangulation.

There are many heuristics for computing data-dependent triangulations ~\cite{Alboul2000,Brown91, DynLR1990,Wang2001}, which are usually based on Lawson's edge flip algorithm \cite{Lawson1977}.
For small instances, the problem can be solved to optimality based on integer linear programming~\cite{NitzkeNFKH2021}.
There are multiple fixed-parameter-tractable algorithms using dynamic programming for the minimum-weight triangulation (MWT) problem \cite{Klincsek1980,ChengT1995,BorgeldBC2008,AnagnostouC1993,Gilbert1979} that can be adapted for decomposable measures~\cite{BernE1995}. Using problem specific structural properties the MWT problem has been solved for instances with up to 30 million points \cite{Haas2018,FeketeHL+2020}.

In \cite{DEKOK2007,RodriguezS2017} heuristics and higher-order Delaunay constraints were used for terrain approximation.
Using established techniques, exact polynomial-time algorithms can be obtained for restricted cases with higher-order Delaunay constraints~\cite{ GudmundssonHK2002,SilveiraK2009}. However, prior to our work, little was known about the  complexity of computing or approximating minimum-error triangulations in the general case. For related problems some hardness results exist \cite{Agarwal1994,MulzerR2008}.

\section{Minimum-error triangulation is NP-hard}\label{sec:NPHardness}
\begin{figure}[t]
	\centering
	\includegraphics[scale=0.5]{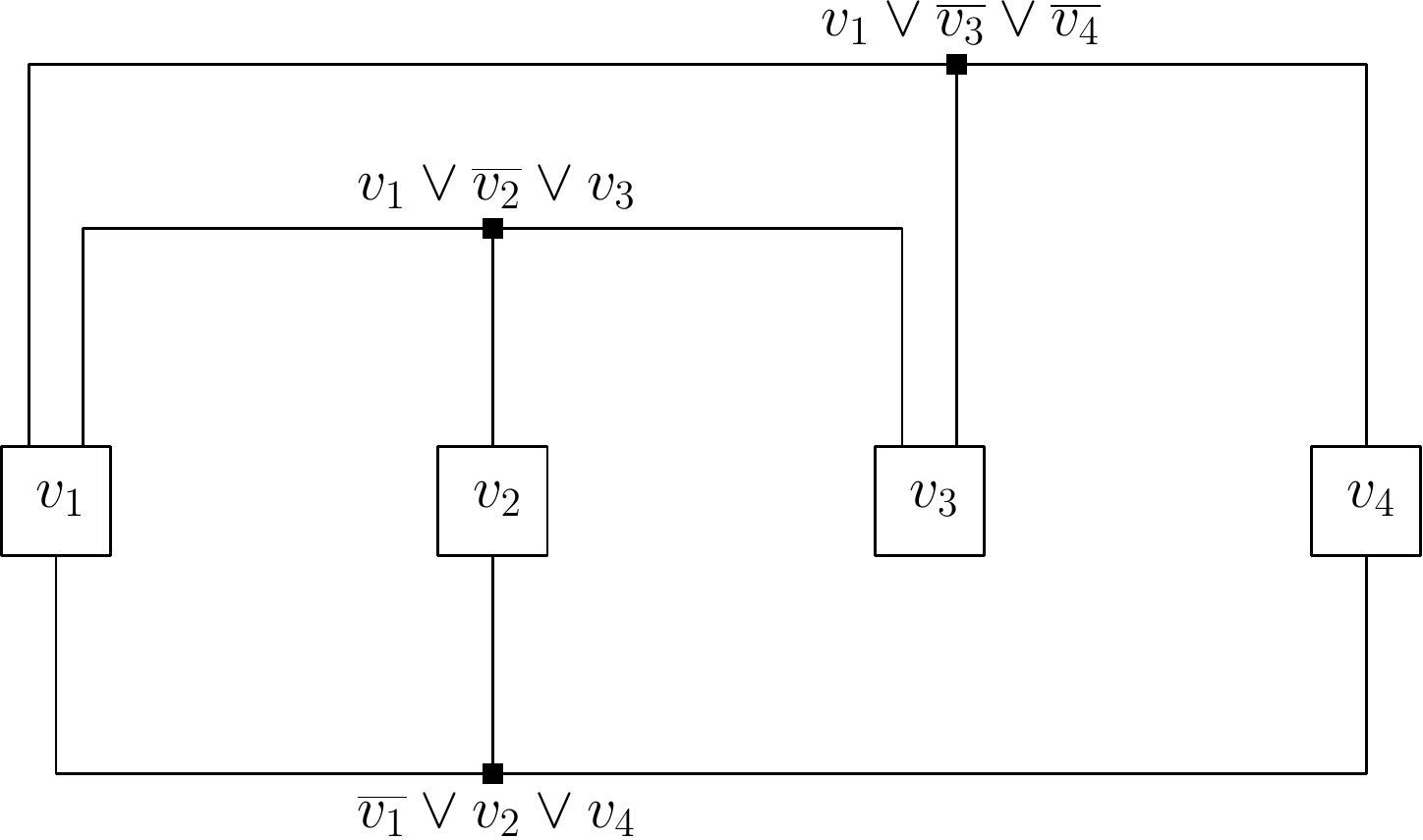}
	\caption{Embedding of the 3SAT formula $(\overline{v_1}\vee v_2\vee v_4)\wedge (v_1\vee\overline{v_2}\vee v_3)\wedge(v_1\vee\overline{v_3}\vee\overline{v_4}) $.}
	\label{fig:planar_3SAT}
\end{figure}
The \emph{zero-error triangulation problem} asks for a triangulation $D$ of $\mS$ with $s_D(r)=h(r)$ for all $r\in\mR$, or equivalently $\Err_D(\mR)=0$. We prove that this problem is NP-hard.
 \begin{restatable}{theorem}{mainTheorem}
	\label{thm}
	The zero-error triangulation problem is \textup{NP}-hard. Thus the minimum-error triangulation problem cannot be approximated within any multiplicative factor in polynomial time unless \textup{P=NP}.
\end{restatable}
We prove this by a reduction from the planar 3SAT problem, which is NP-complete~\cite{Lichtenstein}. 
An instance of this problem can be embedded into the plane, where every clause is represented by a vertex and every variable by a box placed on the horizontal axis. A box is connected to a vertex via a rectilinear edge if the respective variable is contained in the clause.  
For an example, see Figure~\ref{fig:planar_3SAT}.
Such an embedding is also used, for example, in \cite{Knuth}. 

For every instance of the planar 3SAT problem we construct an instance for the zero-error triangulation problem by replacing the boxes, vertices and edges of its rectilinear embedding in the plane by a set of triangulation points and reference points. For this purpose we handle each component of the 3SAT embedding individually. We construct the \emph{variable gadgets} which replace the boxes, the \emph{wire gadgets}, which replace the rectilinear edges and finally the \emph{clause gadgets} and the \emph{negation gadgets}, where the first replace the vertices and the second can be attached to variable gadgets to handle negated variables in a clause. The combination of these gadgets then constitutes an instance to the zero-error triangulation problem. 

We ensure that there are two possible zero-error triangulations on the points belonging to a variable gadget and the attached negation gadgets and wire gadgets as follows.
Points from $\mS$ together with their measurement value can be seen as points in $\R^3$. We ensure that they lie on a paraboloid in $\R^3$ and exploit the properties of the paraboloid (its convexity and the correspondence of planes in $\R^3$ to circles in $\R^2$) to limit possible zero-error triangulations. Any such triangulation then corresponds to the assignment of value 0 (negative) or 1 (positive) to any variable. We claim that the instance can be triangulated with zero error if and only if the 3SAT instance is solvable. 

\begin{figure}[t]
    \hspace*{4cm}
	\includegraphics[scale=0.75, trim= 5cm 12cm 4cm 10cm, clip]{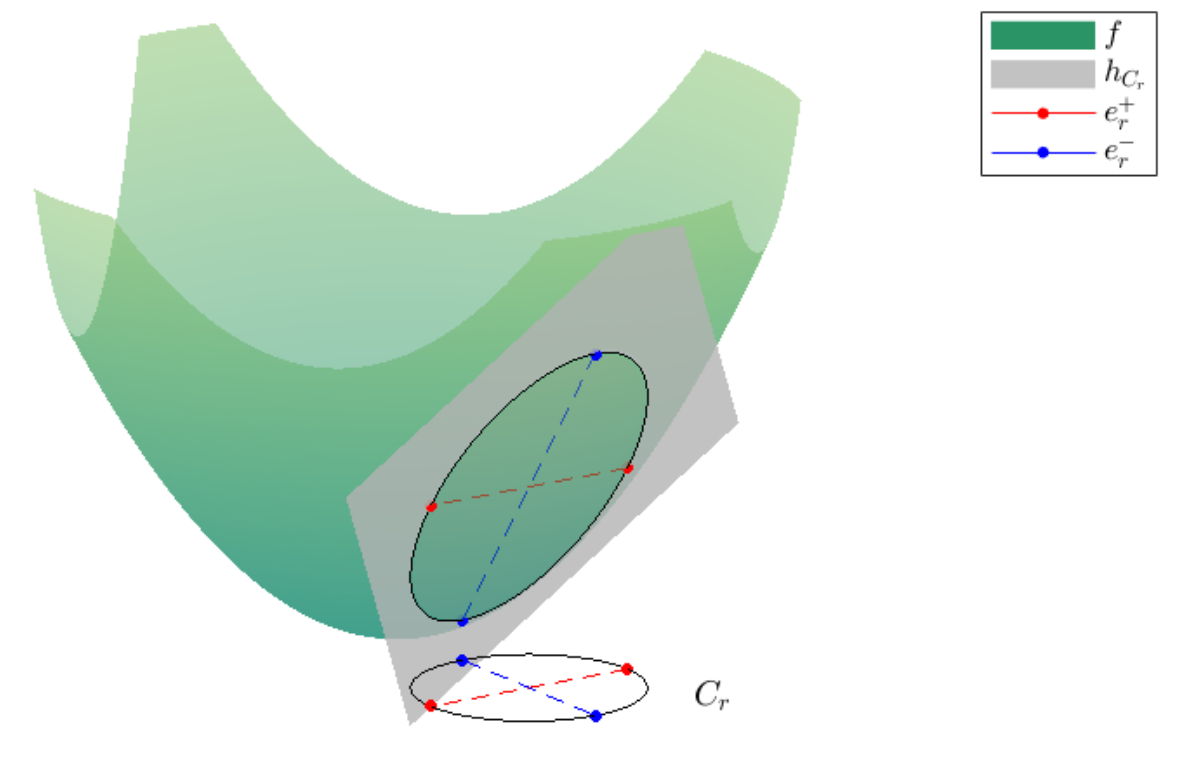}
	\caption{Example of a reference point $r$ with coupled circle $C_r$ and its positive/negative edges crossing at $r$. Lifting the red and blue points to $\R^3$, with their measurement values as third coordinate, we see that these points lie on both the paraboloid and the plane containing $(r,h_{C_r}(r))$.}
	\label{fig:paraboloid}
\end{figure}

\subsection{Notation and local properties}
Our triangulation instance consists of a set of triangulation points with integral coordinates $\mS\subset \Z^2$ and a set $\mR\subset\conv(\mS)$ of reference points. 
The measurement value of a point $p=(p_1,p_2)\in \mS$ is given by $f(p)=p_1^2+p_2^2$. In contrast, reference values are not determined by one single function. Instead we define a set of functions, one for every circle in $\R^2$, and choose for every reference point one of these functions which determines the reference value of this point. Concretely, let $C$ be a circle around a point $x=(x_1,x_2)$ with radius $\rho$. We denote with $I_C=\{y\in \R^2\mid \lVert x-y\rVert_2<\rho\}$ the \emph{interior} of $C$ and with $O_C=\R^2\backslash (C\cup I_C)$ the \emph{exterior} of $C$. Here $\lVert\cdot\rVert_2$ denotes the Euclidean norm. For a reference point $r=(r_1,r_2)\in \mR$ we define the function 
\[h_C(r)=2x_1r_1+2x_2r_2-x_1^2-x_2^2+\rho^2.\] 
The function graph of $f$ is the unit paraboloid $\{(p_1,p_2,p_1^2+p_2^2)\mid (p_1,p_2)\in \R^2\}$ and the 
function graph of $h_C$ is the plane containing the lifting of $C$ onto the paraboloid (Figure~\ref{fig:paraboloid}).

Every point $r\in \mR$ is then \emph{coupled} to a circle, which we denote by $C_r$. It will be defined during the construction of the gadgets and determines the reference value $h(r)=h_{C_r}(r)$.  
Let an edge $e=\overline{st}$ denote the convex hull of two points (its vertices) $s,t\in \R^2$. For each $r\in\mR$ we define a \emph{positive edge} $e_r^+$ and a \emph{negative edge} $e_r^-$ both having triangulation points lying on $C_r$ as endpoints and intersecting each other at $r$ (i.e.,  $e_r^+\cap e_r^-=\{r\}$). Figure~\ref{fig:paraboloid} shows the whole construction.
We say for a triangulation $D$ that the \emph{signal} at $r\in \mR$ is \emph{positive} if $D$ contains edge $e_r^+$ and \emph{negative} if it contains $e_r^-$, otherwise we call it \emph{ambiguous}. 
Similarly for every set $M\subset \mR$ we call $D$ \emph{positive} on $M$ if the signal at all $r\in M$ is positive and \emph{negative} on $M$ if the signal at all $r\in M$ is negative.  The error incurred by $D$ on $M$ is given by
\[\Err_D(M)=\sum_{r\in M}(s_D(r)-h(r))^2.\]

A triangle $T$ is the convex hull of three points $s,t,u\in \R^2$, which we call the vertices of $T$. We say that a triangle $T$ is in $D$ if all of its edges $\overline{st},\overline{tu},\overline{us}$ are in $D$ and $T$ does not contain further points from $\mS$, i.e., $T\cap \mS=\{s,t,u\}$.
We say that $r\in \mR$ is \emph{represented with zero error} by $T$ if $r\in T$ and the value at $r$ of the linear interpolation of $f$ on $T$ equals $h(r)$.
\begin{restatable}{rst-lemma}{lemmaCircleRep}  \label{lemma:circle_rep}
	Let $r$ be a point of $\mR$ and let $T\subset \R^2$ be a triangle with vertices $s,t,u$ and $r\in \conv(\{s,t,u\}\cap C_r)$. Then $r$ is represented with zero error by $T$.
\end{restatable}
If the 3SAT instance is satisfiable, we argue that there is a triangulation containing one of $e_r^{\pm}$ for every reference point $r$. Lemma~\ref{lemma:circle_rep} states that such a triangulation has in fact zero error (see also Figure~\ref{fig:paraboloid}). To represent $r$ with zero error in any other way, we need at least one triangulation point inside and one outside $C_r$. This follows from the convexity of $f$.
\begin{restatable}{rst-lemma}{lemmaInnOut} 
	\label{lemma:lift_innout}
	Let $T\subset\R^2$ be a triangle with vertices $s,t,u$ representing $r\in \mR$ with zero error.  If $r\notin \conv(\{s,t,u\}\cap C_r)$, then $\{s,t,u\}$ has a non-empty intersection with $I_{C_r}$ and $O_{C_r}.$
\end{restatable}
We guarantee during the construction that only few triangulation points lie in $I_{C_r}$ for each reference point $r$. With a concise case analysis we rule out that any of them can be used together with a point in $O_{C_r}$ to form a triangle that represents $r$ with zero error, which limits the choice to triangles containing one of $e_r^{\pm}$. This ensures that every zero-error triangulation yields a solution to the 3SAT instance. 

Our triangulation instance contains a set of \emph{mandatory edges} that we require to be part of any feasible triangulation of $\mS$. Mandatory edges are not part of the zero-error triangulation problem as defined in Section~\ref{sec:TheTriangulationProblem}, but they can be eliminated by an additional construction.

\subsection{The gadgets}
At the core of our reduction lies the design of the gadgets that constitute the triangulation instance. Before we dedicate ourselves to the more complicated gadgets we construct smaller elements called \emph{bits} and \emph{segments} which then are combined into the larger gadgets. 

A \emph{bit} at $r\in\Z^2$ occupies a small construction around the central point $r$, which is also the only reference point of this bit, and can be oriented either horizontally or vertically. We describe the horizontal bit. 
Point $r$ is coupled to a circle $C_r$ which is centered on $r$ and has radius $\sqrt 2$. The integer grid points on this circle, that is, the points $r+(\pm1,\pm1)$, are triangulation points. Moreover $r+(0,1)$ and $r+(0,-1)$ are triangulation points, whereas $r+(-2,0), r+(-1,0), r+(1,0)$ and $r+(2,0)$ are \emph{not}. Therefore, we call the latter points \emph{forbidden}. Furthermore we define the positive and negative edge as 
\[
e_r^+=\conv(r+(-1,-1),r+(1,1)), \quad e_r^-=\conv(r+(-1,1),r+(1,-1)).
\]	
As $r+(\pm1,\pm1)\in C_r$, any triangle containing either $e_r^+$ or $e_r^-$ represents $r$ with zero error by Lemma~\ref{lemma:circle_rep}. For the vertical bit we switch the definition of the positive and negative edge and rotate the whole construction by $\frac{\pi}{2}$. Figure~\ref{fig:bit} illustrates both constructions.
\begin{restatable}{rst-lemma}{lemmaBit}  \label{lemma:bit}
	Suppose the instance contains a bit at $r$. If $\mS\subset \Z^2$ and $\mS$ does not contain forbidden points of the bit, any triangulation $D$ of $\mS$ with $\Err_D(r)=0$ contains one of $e_r^\pm$. 
\end{restatable}

\begin{figure}
	\centering
	\includegraphics[scale=0.5]{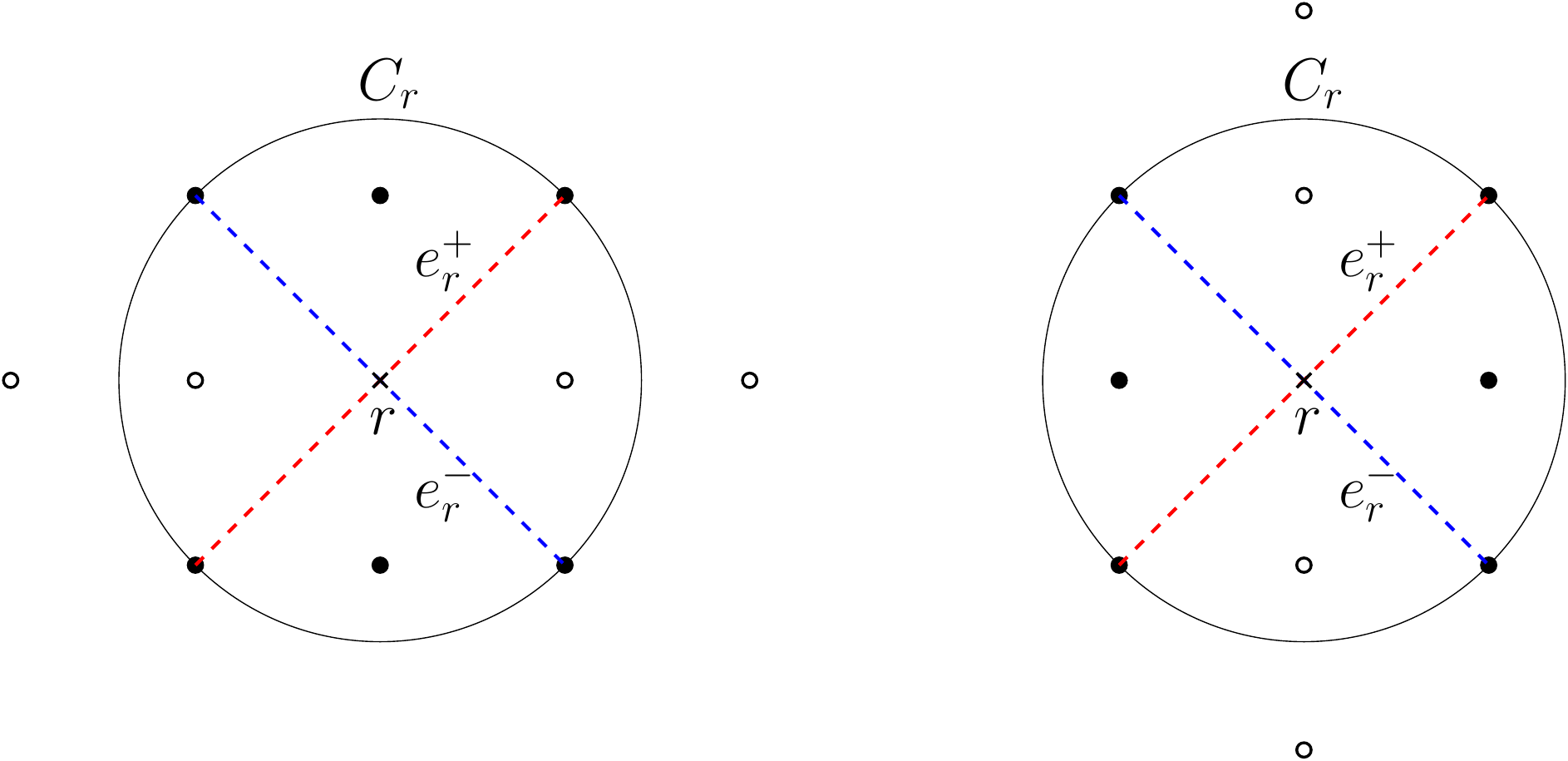}
	\caption{The (horizontal/vertical) bit at $r$ with the positive edge in red and the negative edge in blue. The black points are triangulation points and the white  points are forbidden.}
	\label{fig:bit}
\end{figure} 
The next larger components are the \emph{wire segment} and the \emph{multiplier segment}, which we build from bits. They can be combined at specified reference points, which we call \emph{anchor points}. These points are always reference points of bits.

A \emph{wire segment} connects two points $x,y\in \Z^2$ lying on the same horizontal or vertical line. We place a horizontal or vertical bit on $x,y$ and all integral points lying between these on the line connecting $x$ and $y$. The anchor points of this segment are $x,y$. 

A \emph{multiplier segment} at a point $x\in \Z^2$ consist of two horizontal bits at $x\pm(2,0)$ and two vertical bits at $x\pm(0,2)$. These four points are simultaneously anchor points. Furthermore we add four inner reference points $x\pm(0,1),x\pm(1,0)$ whose coupled circle is of radius $\sqrt{5}$ and centered around $x$. So the circle contains the points $x+(\pm2,\pm1),x+(\pm1,\pm2)$. Figure~\ref{fig:wire_mult_seg} shows the wire segment and the multiplier segment including mandatory edges and the positive/negative edges of the inner reference points.   
\begin{figure}
	\centering
	\includegraphics[scale=0.7]{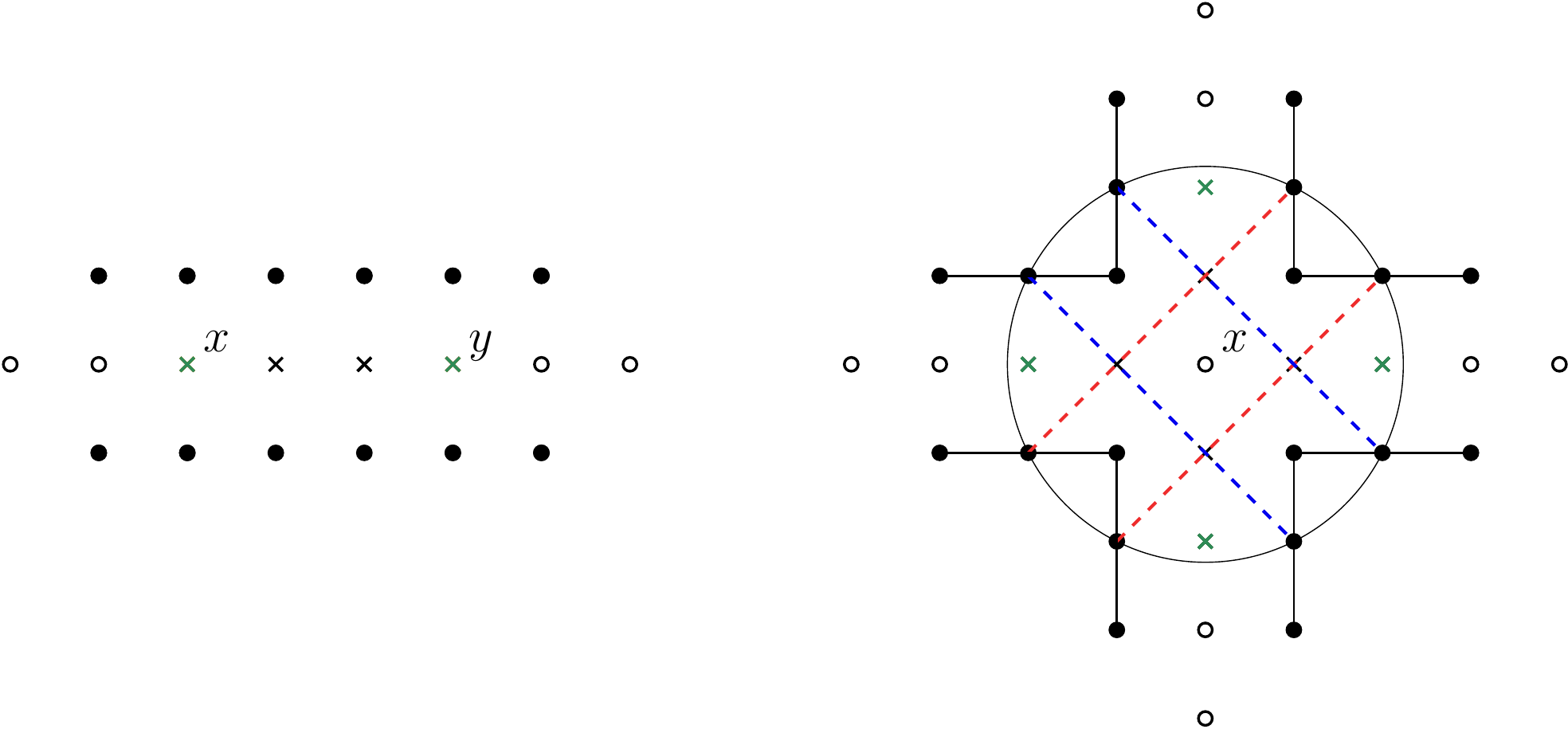}
	\caption{Example of a horizontal wire segment on the left and a multiplier segment with mandatory edges on the right. The red or blue edges indicate the positive or negative edges of the crossing points, respectively. All white points and all reference points are forbidden. The green points are anchor points.}
	\label{fig:wire_mult_seg}
\end{figure}
To obtain the larger variable gadget and wire gadget we combine wire segments with multiplier segments. Two segments can be combined if they share a common anchor point. By the combination of two segments we mean the union of their reference points and triangulation points. A point is forbidden in the combination if it is forbidden in at least one of the segments. Thus it is not allowed to combine two segments if a triangulation point of one is forbidden in the other. The set of anchor points of the combination is defined as the symmetric difference of anchor point sets of both segments. This way we can combine arbitrarily many segments.

Remember that the \emph{wire gadget} replaces the rectilinear edges of the 3SAT embedding, so it has to connect two points $x=(x_1,x_2),y=(y_1,y_2)\in \Z^2$. It consists of a multiplier segment placed on either $(x_1,y_2)$ or $(y_1,x_2)$ to form a corner, which is connected on two of its anchor points via two wire segments to both $x$ and $y$.
A \emph{variable gadget} at $v\in \Z^2$ consists of $\ell$ multiplier segments at sufficiently large distance $\alpha\in \Z$, which we do not specify further. Here $\ell$ denotes the number of clauses. Concretely, we place a multiplier segment on each of the points $v+(k\alpha,0)$ with $0\leq k\leq \ell-1$ and connect them via horizontal wire segments at their anchor points. The multiplier segments ensure that the gadget can later be connected at its anchor points to multiple clause gadgets. We observe that the described combinations of segments for both gadgets are allowed and that they have the following crucial property.
\begin{restatable}{rst-lemma}{lemmaWireVar}
	\label{lemma:wire}
	Suppose the instance contains a wire/variable gadget and let $\RR$ be the reference points of this gadget. If $\mS\subset \Z^2$ and $\mS$ does not contain forbidden points of the gadget, any triangulation $D$ of $\mS$ with $\Err_D(\RR)=0$ is either positive or negative on $\RR$. 
\end{restatable} 
Now we define the \emph{clause gadget} at a point $c\in \Z^2$, which combines three signals. To this end we add a reference point $r_c=c+(0,11)$. Instead of a positive/negative edge it comes with three triangles $T_1,T_2,T_3$ whose vertices lie on $C_{r_c}$, each triangulating $r_c$ with zero error. The clause gadget can be connected to other gadgets at three anchor points $a_1,a_2,a_3$.
With an additional construction we block the triangle $T_i$ if the signal at $a_i$ is positive for $i=1,2$ and $T_3$ if the signal at $a_3$ is negative. 
\begin{figure}
	\centering
	\includegraphics[scale=0.6]{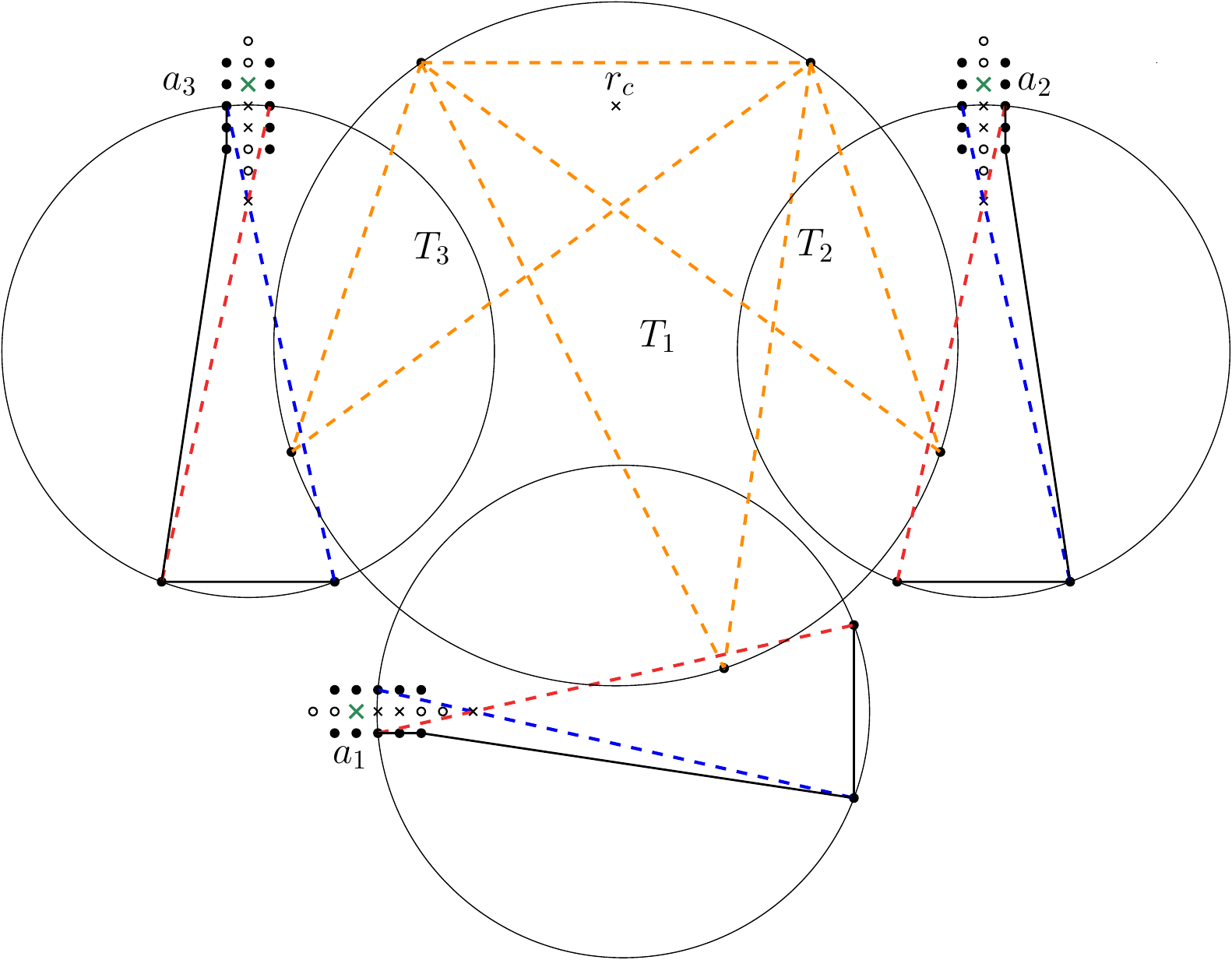}
	\caption{The clause gadget, where the red/blue edges indicate the positive/negative edges of the crossing points. The triangles $T_1,T_2,T_3$ are orange and the anchor points $a_1,a_2,a_3$ green.}
	\label{fig:klausel}
\end{figure}
For the construction we refer to Appendix~\ref{ap:np} and to Figure~\ref{fig:klausel}.
\begin{restatable}{rst-lemma}{lemmaKlausel}
	\label{lemma:clause}
	Suppose the instance contains a clause gadget and let $\RR$ be its reference points. If $\mS\subset \Z^2$ and $\mS$ does not contain forbidden points of the gadget, any triangulation $D$ of $\mS$ with $\Err_D(\RR)=0$ must be negative on one of the anchor points $a_1,a_2$ or positive on $a_3$.
\end{restatable}
The last gadget, the \emph{negation gadget}, is discussed in Appendix~\ref{ap:np}. It is constructed out of wires, multipliers and simplified clause gadgets. Finally, we replace the mandatory edges by an additional construction and argue that all gadgets keep their crucial properties. Using them we construct the zero-error triangulation instance and prove Theorem~\ref{thm} in Appendix~\ref{ap:np}.

\section{Higher-order Delaunay optimization}\label{sec:HODOptimization}
In the previous section we established that finding a minimum-error triangulation is NP-hard. Moreover, the experiments in \cite{NitzkeNFKH2021} by Nitzke et al.\ suggest, that general minimum-error triangulations do not yield the most promising reconstructions of the sea surface. In their paper they used higher-order Delaunay (HOD) triangulations which allow a trade-off between a well shaped triangulation and a good approximation of the training dataset.

In this section we summarize the algorithm given by Silveira et al.\ in \cite{SilveiraK2009}. Additionally, we extend upon their work by investigating the fixed-edge graphs in more detail.

\medskip We only consider point sets $\mathcal{S}$ in general position, i.e.,  no four points lie on a circle and we denote the circle defined by three vertices $u,v,w\in\mathcal{S}$ by $C(u,v,w)$. A triangle $T_{u,v,w}$ is called an \emph{order-$k$ Delaunay ($k$-OD) triangle}, if $C(u,v,w)$ contains at most $k$ points from $\mathcal{S}$ in the interior. A triangulation is called \emph{$k$-OD triangulation}, if all of its triangles have order $k$ and an edge is called \emph{useful $k$-OD edge}, if some $k$-OD triangulation of $\mathcal{S}$ uses it; see Figure \ref{fig:OrderKTriangulation}.

\begin{figure}[!tb]
\begin{minipage}[t]{.58\textwidth}
\centering
\includegraphics[width=.9\textwidth]{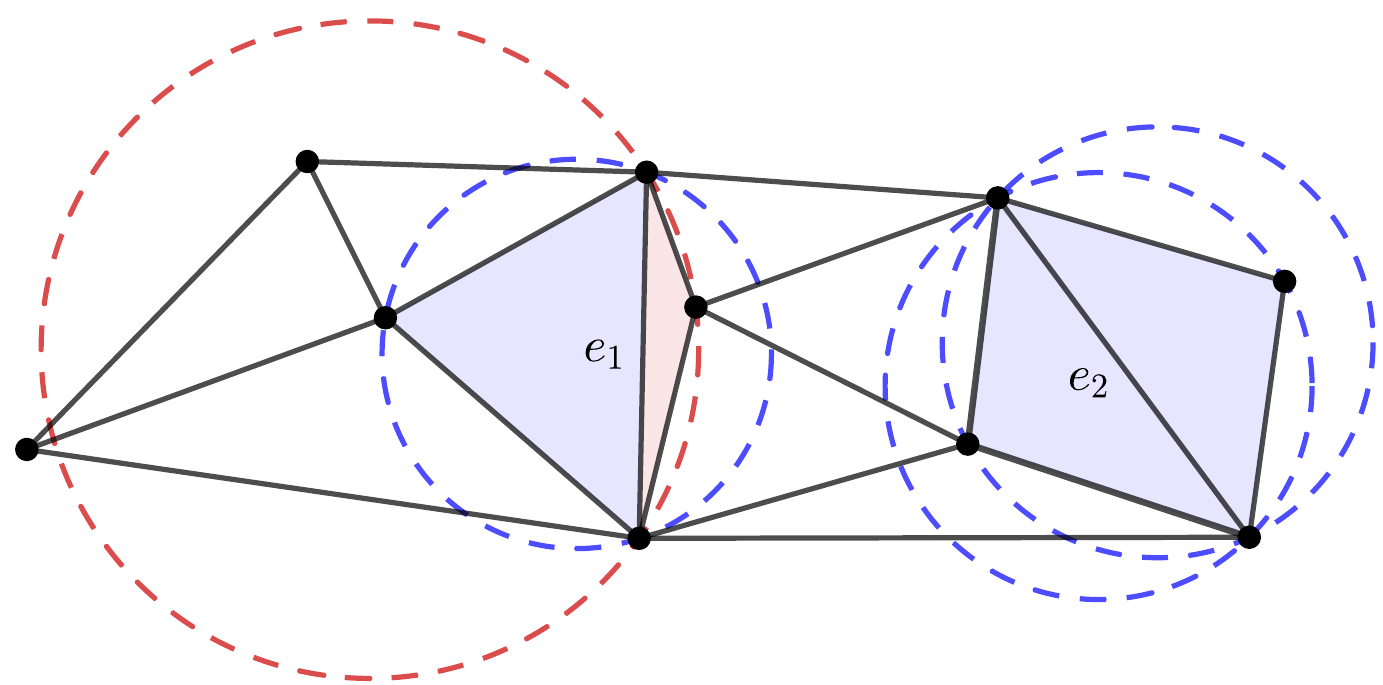}
\caption{A $2$-OD triangulation; in blue the $1$-OD and in red the $2$-OD triangles; $e_1$ is a useful  $2$-OD edge and $e_2$ is a useful $1$-OD edge}\label{fig:OrderKTriangulation}
\end{minipage}
\hfill
\begin{minipage}[t]{.4\textwidth}
\centering
\includegraphics[width=0.7\textwidth]{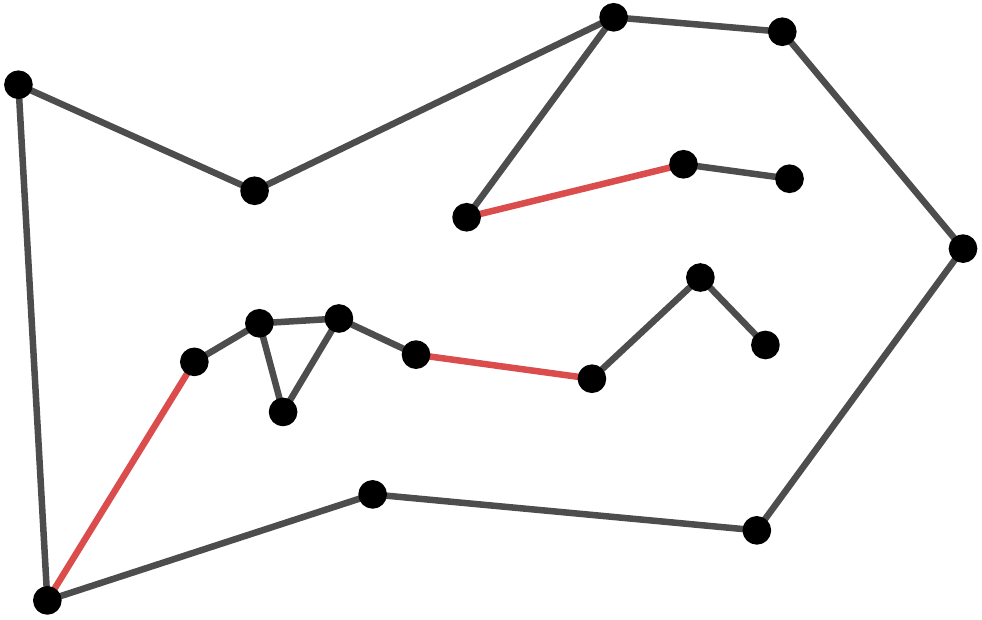}
\captionsetup{textformat=simple}
\caption{In black a (degenerate) polygon with connected components; in red one set $H$ of connections}\label{fig:PolygonWithConnectedComponents}
\end{minipage}
\end{figure}

The minimum-error measure $\text{Err}_D(\mathcal{R})$ can be optimized using dynamic programming, since it is decomposable after pre-processing the triangle weights; see \cite{BernE1995} for a formal definition. 
The well known DP algorithm that was independently proposed by Klincsek in \cite{Klincsek1980} and Gilbert in \cite{Gilbert1979} can be used to optimize polygon triangulations for decomposable measures in $O(n^3)$ time. 
In \cite{SilveiraK2009} the runtime of the DP algorithm is improved to $O(nk^2)$, if the algorithm only considers pre-processed $k$-OD edges and triangles instead of all possible ones.

Furthermore, Silveira et al.\ \cite{SilveiraK2009}  extend the algorithm to the class of polygons $P$ containing $h$ connected components $C_1,\ldots, C_h$; see Figure \ref{fig:PolygonWithConnectedComponents}. The algorithm performs an exhaustive search on a collection $\mathcal{H}$ of sets of edges $H$, such that the planar graph $\bigcup_i C_i\cup P\cup H$ is connected for each $H\in\mathcal{H}$ and at least one $H$ is used in the optimal triangulation. One of the main results in \cite{SilveiraK2009} is the existence of such a collection with size $O(k)^h$.

\begin{theorem}[from \cite{SilveiraK2009}]\label{theorem:CCAlgorithm}
An optimal $k$-OD triangulation with respect to $\text{Err}_D(\mathcal{R})$ of a (degenerate) polygon with $n$ boundary vertices and $h\geq 1$ components inside can be computed in $O(kn\log n)+O(k)^{h+2}n$ expected time.
\end{theorem}
We can apply this algorithm to point sets by finding subgraphs $F$ of the optimal triangulation \cite{ChengT1995,SilveiraK2009} and applying the DP algorithm to the faces of $F$. 

\subsection{The order-\kk fixed-edge graph}\label{sec:FixedEdgeGraph}
\begin{figure}[!tb]
  \centering
    \includegraphics[width=0.9\textwidth]{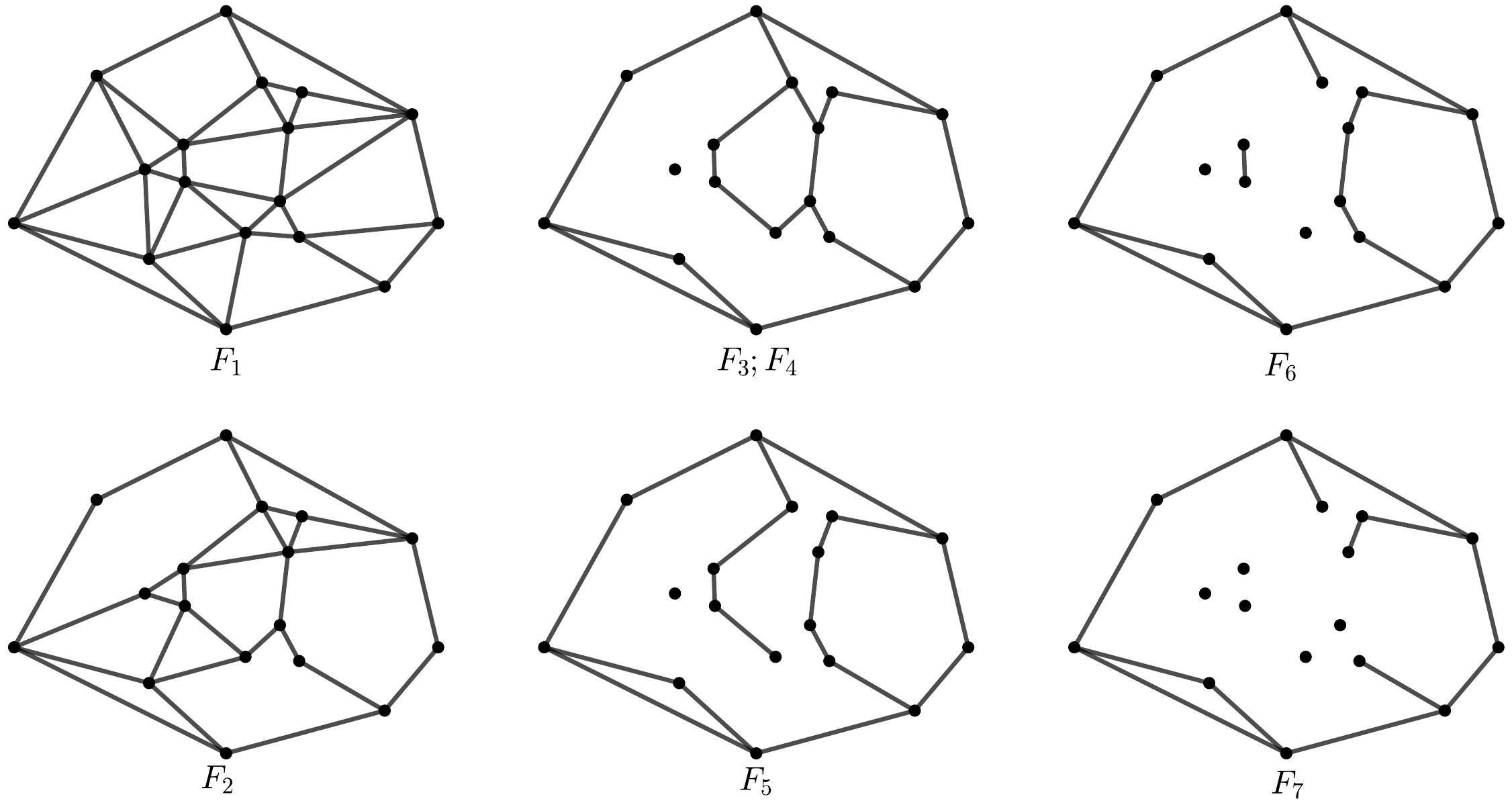}
    \caption{A sequence of fixed-edge graphs $F_1,\ldots,F_7$ for an example point set}\label{fig:fixedEdges}
\end{figure}
A subgraph that is naturally given by HOD constraints is the fixed-edge graph which was first discussed in \cite{SilveiraK2009}. The \emph{order-$k$ Delaunay ($k$-OD) fixed-edge graph} $F_k$ of a pointset $\mathcal{S}$ is given by all useful $k$-OD edges that are not intersected by any other useful $k$-OD edge.

\begin{observation}\label{fixedEdgeGraphProperties}
Let $\mathcal{S}$ be a set of $n$ points. Let $DT$ denote the Delaunay triangulation. We have $DT=F_0\supset F_1 \supset F_2\supset ... \supset F_{m}=...=F_{n}\supset \text{conv}(\mathcal{S})$ for some $m\leq n$.
\end{observation}

In Figure \ref{fig:fixedEdges} a sequence of fixed-edge graphs is illustrated. $F_k$ decomposes the pointset into degenerate polygons $P_1,\ldots,P_m$ that may contain some connected components. An example is given in Figure \ref{fig:decomposition}.
\begin{figure}[!b]
\begin{subfigure}[t]{0.38\textwidth}
\centering
\includegraphics[width =.9\textwidth ]{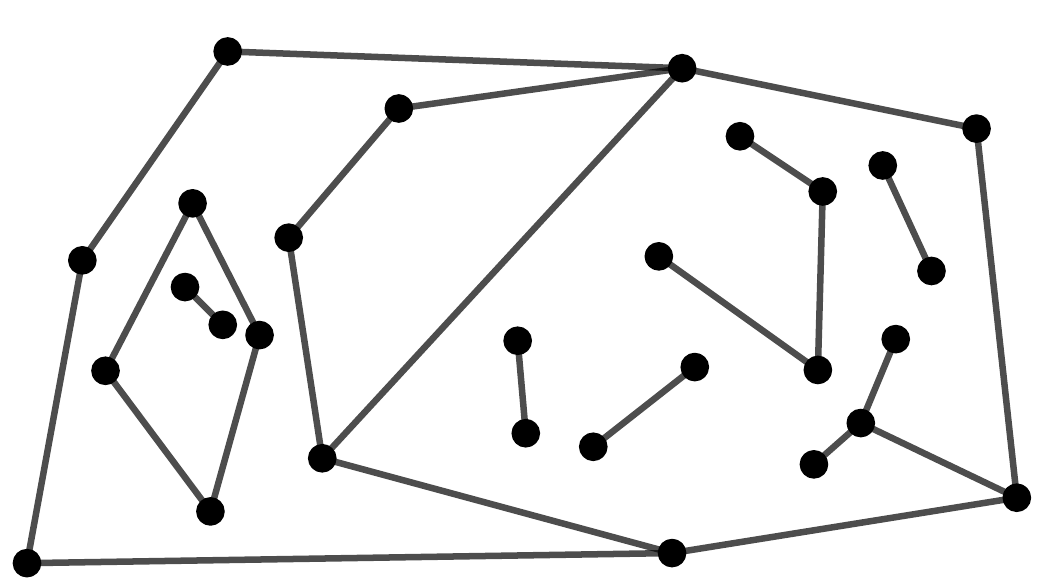}
\end{subfigure}\hfill
\begin{subfigure}[t]{0.38\textwidth}
\centering
\includegraphics[width =.9\textwidth ]{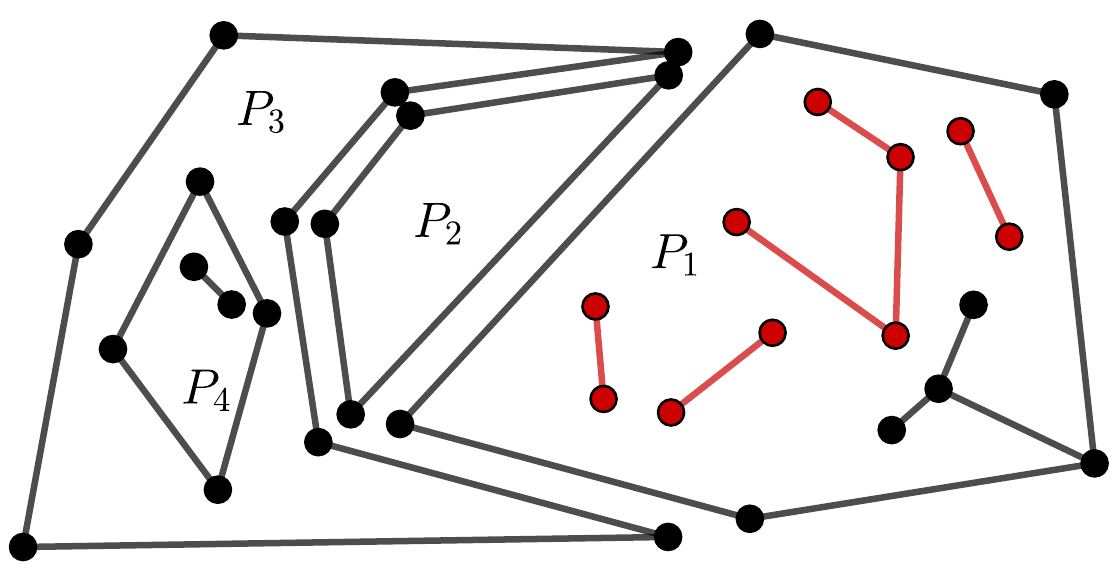}
\end{subfigure}
\caption{The decomposition of a fixed-edge graph into polygons. We have $c_1 = 4$, $c_2 = 0$, $c_3=1$ and $c_4=1$ for the number of components in each polygon. Thus, we have $c_\text{max}=4$. Note that the component inside $P_4$ is not counted towards $c_3$, but to $c_4$.}\label{fig:decomposition}
\end{figure} We can compute optimal solutions $D_i$ for all $P_i$ with the DP algorithm. Since $\text{Err}_D(\mathcal{R})$ is decomposable, the optimal triangulation of $\mathcal{S}$ is given by $\bigcup_i D_i$. Therefore, the runtime of the algorithm is dominated by the polygon with the maximum number of connected components $c_\text{max}$. The application of Theorem \ref{theorem:CCAlgorithm} results in:

\begin{corollary}\label{theorem:DPpointsets}
An optimal $k$-OD triangulation of a point set $\mathcal{S}$ with respect to $\text{Err}_D(\mathcal{R})$ can be computed in $O(kn\log n)+O(k)^{c_\text{max}+2}n$ expected time.
\end{corollary}
Next, we give some theoretical results with respect to the structure of $F_2$ and $F_3$.

Let $v\in \mathcal{S}$ be a triangulation point. We call the graph $N$ given by all edges of its incident Delaunay triangles its \emph{Delaunay neighbourhood}, all of its incident edges in $N$ its \emph{connecting edges} and all other edges of $N$ its \emph{boundary edges}. A useful $2$-OD edge that intersects a connecting edge is called \emph{separation edge}; see Figure \ref{fig:cycle}.
 \begin{restatable}{theorem}{fixedTheorem}
	\label{theorem:fixedEdges}
	Let $\mathcal{S}$ be a set of points. Then every vertex in $F_2$ is adjacent to at least one other vertex of $\mathcal{S}$.
\end{restatable}
\begin{figure}[!tb]
\centering
\includegraphics[width=0.27\textwidth]{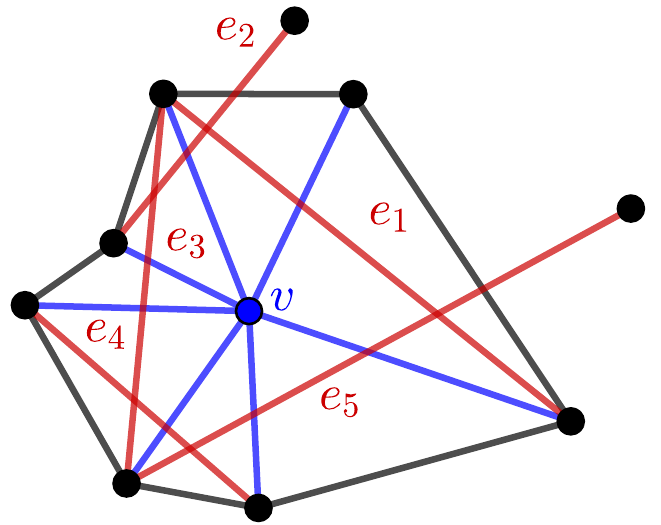}
\caption{The Delaunay Neighbourhood of a point $v$ and a cycle of separation edges given in red}\label{fig:cycle}
\end{figure}

\begin{proof}
(Sketch; complete proof can be found in Appendix \ref{sec:missingproofs}) It is sufficient to prove that for every vertex $v\in\mathcal{S}$ at least one connecting edge cannot be intersected by a separation edge. For the sake of contradiction we assume that there exists a set $E$ of separation edges such that every connecting edge is intersected by at least one $e\in E$.

In a first step we can prove that at least one endpoint of any $e\in E$ must be part of the Delaunay neighbourhood of $v$. Additionally, we can show that no boundary edge $\overline{uw}$ can be intersected by a separation edge for $\overline{vu}$ and a separation edge for $\overline{vw}$. These observations imply that we can order the edges in $E$, such that for all $i$ the separation edge $e_i$ intersects $e_{i-1}$ and $e_{i+1}$, i.e., the separation edges form a cycle as depicted in Figure \ref{fig:cycle}.

Next, we show that every pair of consecutive separation edges $(\overline{u_iv_i},\overline{u_{i+1}v_{i+1}})$ must satisfy a special property, i.e., it must hold that $u_{i+1}\in C(u_i,v_i,v_{i+1})$ and $v_{i+1}\in C(u_i,v_i,u_{i+1})$. Finally, we show that this is not possible which leads to a contradiction.
\end{proof}
It is well known \cite{SilveiraK2009, GudmundssonHK2002} that $F_1$ is connected ($c_\text{max}=0$). Silveira et al.\ stated in \cite{SilveiraK2009} that for $k>1$ the value $c_\text{max}$ can be larger than $0$. But their experiments do not yield any example for which $F_2$ is not connected. We complement the discussion by such an example. Additionally, we show for all $k\geq 3$ there are examples with $c_\text{max}\in\Omega(n)$.
\begin{observation}\label{obs:worstcaseFixededges}
\begin{itemize}
\item There exist point sets with $c_\text{max}>0$ for $F_2$; see Figure \ref{fig:2ODNotConnected}.
\item For every $n$ and $k\geq 3$ there are point sets of size $n$ with $c_\text{max}=\lfloor \frac{n}{6}\rfloor$ for $F_k$; see Figure~\ref{fig:Examplek3}.
\end{itemize}
\end{observation}
\begin{figure}[!b]
\begin{minipage}[t]{.49\textwidth}
\centering
\includegraphics[width=0.5\textwidth]{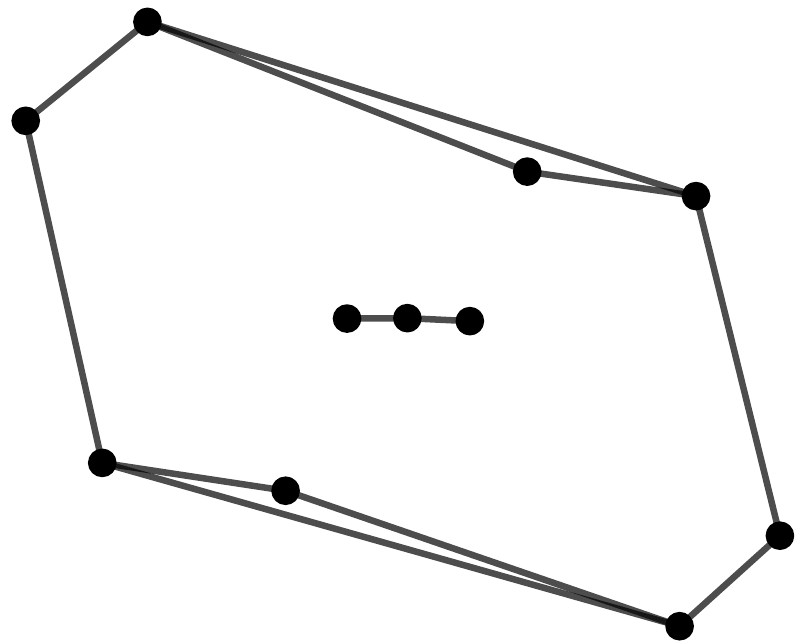}
\caption{An example with disconnected $F_2$}\label{fig:2ODNotConnected}
\end{minipage}
\hfill
\begin{minipage}[t]{.49\textwidth}
\centering
\includegraphics[width=1\textwidth]{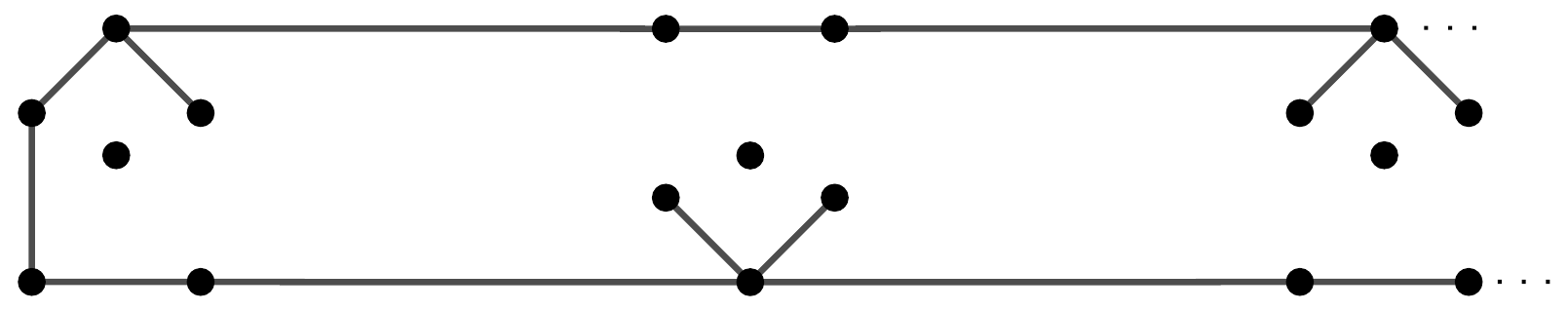}
\captionsetup{textformat=simple}
\caption{An example with $c_\text{max}=\frac{n}{6}$ for $F_3$}\label{fig:Examplek3}
\end{minipage}
\end{figure}
\subparagraph*{Open question:} 
Is there a constant $d$, such that $F_2$ has $c_\text{max}\leq d$ for every point set?
\subparagraph*{Practical implications:} Our results are interesting from a theoretical point of view, but the experiments in \cite{SilveiraK2009} with random point sets by Silveira et al.\ and also our own experiments (Appendix \ref{sec:RandomExperiments}) indicate that for practical datasets $c_\text{max}$ is small for $k\leq 7$. Next, we confirm this assumption for the tide gauge dataset which is used for the sea surface reconstruction.

\section{Experiments}\label{sec:Experiments}
We start this section by discussing the datasets. Next, we discuss the fixed-edge graphs of the tide gauge dataset. Afterwards, we provide the reconstruction process and our experimental setup. Finally, we present our results regarding the runtime and quality.

\subsection{The datasets}
The triangulation points for the minimum-error triangulation problem are given by the monthly tide-gauge time series from the Permanent Service for Mean Sea Level (PSMSL)
\cite{PSMSL2021}, which is further discussed in \cite{Holgate+2012}. 
We use the revised local reference (RLR) datasets. Furthermore, we remove some stations which do not have any values in our time-frame. This results in a dataset with 1502 stations, but not all of them record monthly. Thus, we only use between 513 and 804 different stations at once for a reconstruction.

As reference data $\mathcal{R}$ we use the satellite altimeter datasets provided by the ESA Sea Level Climate Change Initiative (SLCCI), which are given in \cite{ESA2021SSAData} and are further discussed in \cite{Cazenave+2015}. They are given as monthly gridded sea level anomalies with a spatial resolution of 0.25 degrees and are available for the timespan January 1993 to December 2015.

We assume that both datasets are given in radial coordinates. Since we focus on planar triangulations, we need to use a global map projection. We chose the Lambert azimuthal projection (LAP) which  unfolds the sphere onto the plane starting at an anchor point $(\lambda_0,\phi_0)$. For our experiments the LAP has one advantage: The projection results in significantly different distributions of the stations for sufficiently different anchor points $(\lambda_0,\phi_0)$. This allows us to perform the fixed-edge graph experiments for a wide variety of point distributions.

It is important to note that the experiments in this paper focus on the runtime of the DP algorithm for a real world application. Thus, we only de-mean the tide gauge data as discussed in \cite{NitzkeNFKH2021} and do not apply any additional corrections. 

\subsection{The fixed-edge graphs of the tide gauge set}
For our experiments with respect to the fixed-edge graphs we use the complete RLR dataset, i.e., all $1502$ stations.
We use the LAP with anchors $(\lambda_0,\phi_0)$ on an uniform $2$-D $20\times 20$ grid to generate 400 distributions of the dataset. In Table \ref{tab:5realWorldFixed} the experiments are summarized. The values $\text{avg}_{c_\text{max}}$ are given by the average value of $c_\text{max}$ over all samples. Additionally, we have \textit{min} and \textit{max} that depict the minimal and maximal value of $c_\text{max}$ for all samples.
\begin{table}[!b]\begin{center}
\caption{The average of $c_\text{max}$ and the min/max value of $c_\text{max}$ for the projections of the RLR data}\label{tab:5realWorldFixed}
\centering
\begin{tabular}{| c | c |  c | c | c | c | c | c | c | c | c |}
\hline
$ $  & $k=1$  & $k=2$ & $k=3$ & $k=4$ & $k=5$ & $k=6$ & $k=7$ & $k=8$ & $k=9$\\
\hline
$\text{avg}_{c_\text{max}}$ & 0.00 & 0.00 & 0.45 & 1.20  & 2.05 & 3.68 & 7.11 & 15.88 &  33.16 \\
min/max & 0/0 &  0/0 & 0/2 & 0/3 & 1/5 & 2/12 & 3/18 & 6/38 & 11/82\\
\hline
\end{tabular}
\end{center}\end{table} 
The results roughly coincide with the experiments performed on random point sets by Silveira et al.\ in \cite{SilveiraK2009} and our own preliminary experiments. The experiments suggest, that we can expect the DP algorithm to compute optimal solutions for $k\leq 7$ in reasonable time.
Since Nitzke et al.\ suggest very small $k$ for the reconstruction in \cite{NitzkeNFKH2021}, these experiments are promising.

\subsection{Sea surface reconstruction}
The reconstruction process can be summarized as follows: We learn a minimum-error triangulation $D$ in some epoch $i$ and then use it to reconstruct the sea surface at some other point in time $j$, by using the triangulation $D$ with the height values of epoch $j$.
Since not all tide gauge stations provide data for every epoch $i$, we need to consider the set $G^{ij}$ which is given by all stations that have reasonable values for epoch $i$ as well as for $j$. We denote the optimal triangulation using $G^{ij}$ and the reference points $A_i$ by $D^{ij}_M$. For comparison we use the Delaunay triangulation $D^{ij}_D$ of the set $G^{ij}$ which has already been successfully used for the sea surface reconstruction task in \cite{OlivieriS2016}.
If we have altimeter data available for epoch $j$, we can evaluate the quality of our approximation. 
Overall the reconstruction for epoch $j$ using $i$ and order $k$ can be performed as follows:

\begin{enumerate}
\item Compute the set $G^{ij}$ and the $k$-OD triangles $\mathfrak{T}^{ij}$ as described in \cite{SilveiraK2009}.
\item Compute the weights $w_T(A_i)$ of all $T\in \mathfrak{T}^{ij}$ with respect to $A_i$ as discussed in Section \ref{sec:TheTriangulationProblem}.
\item Compute the optimal $k$-OD triangulation $D^{ij}_M$ with the DP algorithm given in Section \ref{sec:HODOptimization} and also compute the Delaunay triangulation $D^{ij}_D$.
\item Evaluate the quality of the triangulations with respect to $A_j$.
\end{enumerate}

\noindent For the evaluation we compute the \emph{empirical variance} of a triangulation
\begin{align*}
\sigma_{ij}^2({D})=\frac{1}{n-1}\sum_{T\in{D}}\sum_{a\in A_j, a\in T } (s_T(a)-h_j(a))^2,
\end{align*}
where $n$ is the number of altimeter points in conv${(D)}$. Note that this is exactly the average minimum error. Additionally, we define the \emph{variance reduction} of a reconstruction by 
\begin{align*}
\Delta \sigma_{ij}^2=\sigma_{ij}^2(D^{ij}_M)-\sigma_{ij}^2(D^{ij}_D).
\end{align*}
Next, we can group together reconstructions for epochs $i,j$ and $i',j'$ where $|i-j|=|i'-j'|$. This allows us to define the \emph{average variance reduction} of a temporal difference $\Delta d$ by
\begin{align*}
q(\Delta d)= \frac{1}{|\mathcal{D}(\Delta d)|} \sum _{(i,j)\in\mathcal{D}(\Delta d)}\Delta\sigma^2_{ij}.
\end{align*}
The set $\mathcal{D}(\Delta d)$ is given by all tuples $(i,j)$ with $|i-j|=\Delta d $. Using the temporal difference, we can
investigate how far back in time our optimized triangulation outperforms the Delaunay triangulation (DT). 
Nitzke et al.\cite{NitzkeNFKH2021} noticed that $q$ has a seasonal behaviour, i.e., $q$ has local maxima every 12 month. Thus, we only use datasets with $j =i\pm 12l$ with $l\in\mathbb{N}$ for the reconstruction.
A more in depth discussion of the evaluation methods can be found in \cite{NitzkeNFKH2021}.

\subparagraph*{Reconstruction quality:}
\begin{figure}[!tb]
  \centering
    \includegraphics[width=0.9\textwidth]{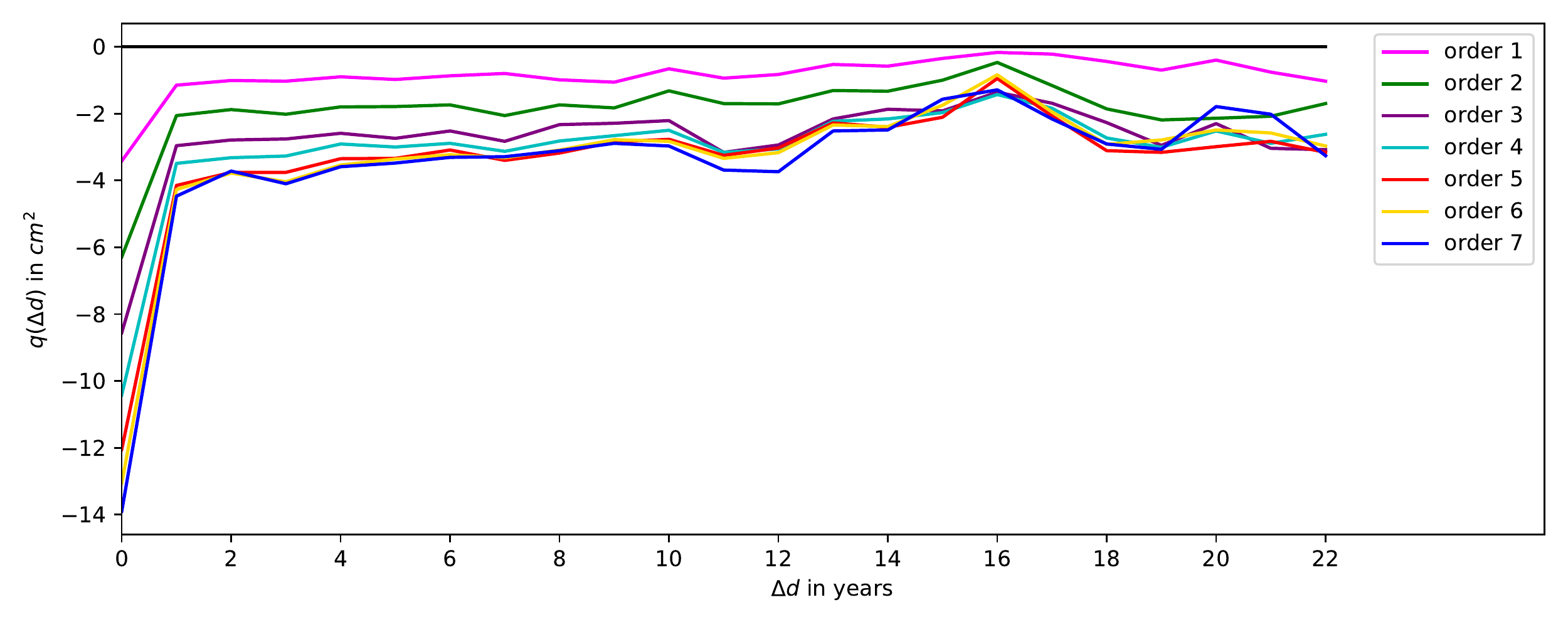}
    \caption{Averaged $q(\Delta d)$ of our approach w.r.t.\ the epoch difference $\Delta d$ for different order $k$}\label{fig:quality_plot}
\end{figure}
For all of the experiments we choose an LAP anchored in the Atlantic Ocean, namely $(-40, 16)$. We compute all possible reconstructions for epochs $i$ and $j$ with $i \geq j$ for the orders $k \leq 7$, i.e., we use every epoch $i$ for training and validate the learned triangulation on all possible epochs $j$ with $j =i-12l$. Next, we group them with respect to $\Delta d$.  In Figure \ref{fig:quality_plot} the $q(\Delta d)$ values are depicted. Recall that our approach performs better than the DT, if $q(\Delta d) < 0$. It should be mentioned, however, that for $\Delta d\geq 18$  the quality of the experiments deteriorates, since only few samples span this epoch difference.

Note that the variance reductions for $\Delta d = 0$ are far better than for larger $\Delta d$, since the reconstruction epoch is the same as the training epoch. The variance reductions for order~$1$ and order $2$ are smoother, but also worse than the ones for higher orders.
For $\Delta d > 10$ the variance reductions for the orders $3$--$6$ are very similar and even order $7$ is comparable. The aforementioned orders also share local extrema at $\Delta d = 10, 11, 18, 20$. For order $7$ the extrema become more pronounced which leads to better minima but also to worse maxima. 
Note that calculating the empirical variances $\sigma_{ij}^2(D^{ij}_D)$ for all epochs yields values between $80\text{cm}^2$ and $120\text{cm}^2$. Hence, for example, an absolute variance reduction of $2\text{cm}^2$ roughly coincides with a relative variance reduction of $2\%$.

The overall variance reduction gets better for higher orders. This is contrary to the results by Nitzke et al.~\cite{NitzkeNFKH2021}, who suggested $k=1,2$ for the reconstruction.
This difference may have geometric reasons, i.e., the points in the North Sea dataset used in \cite{NitzkeNFKH2021} more or less trace a polygon without inner points and our global datasets have a more arbitrary distribution. 
Moreover, the LAP distorts distances as well as angles which may also contribute to the different results for the local and global datasets.

\subparagraph*{Runtime:}
\begin{figure}[!tb]
\begin{minipage}[t]{.48\textwidth}
\centering
\includegraphics[width=1\textwidth]{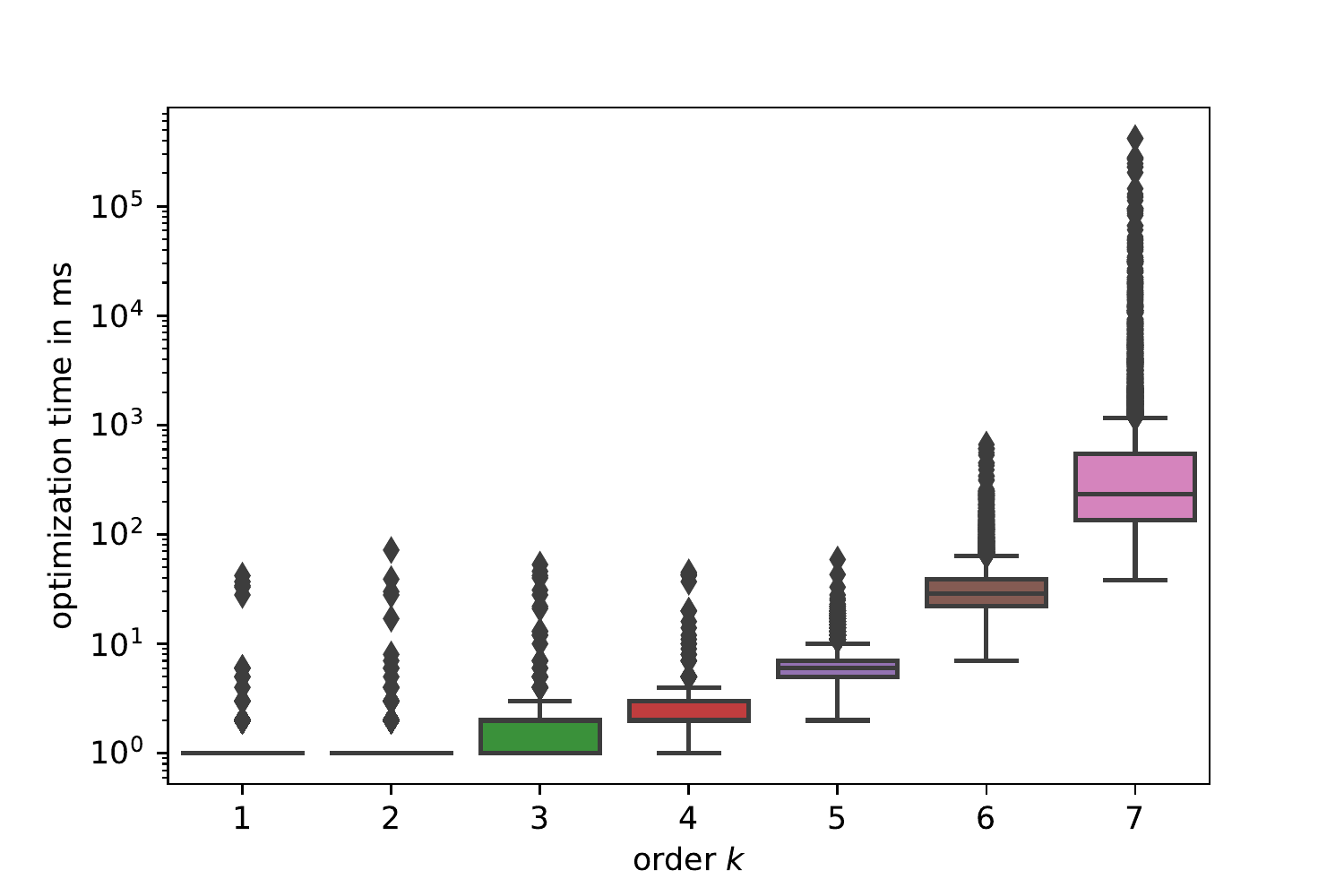}
\caption{Optimization time depending on the order}\label{fig:order_plot}
\end{minipage}
\hfill
\begin{minipage}[t]{.48\textwidth}
\centering
\includegraphics[width=1\textwidth]{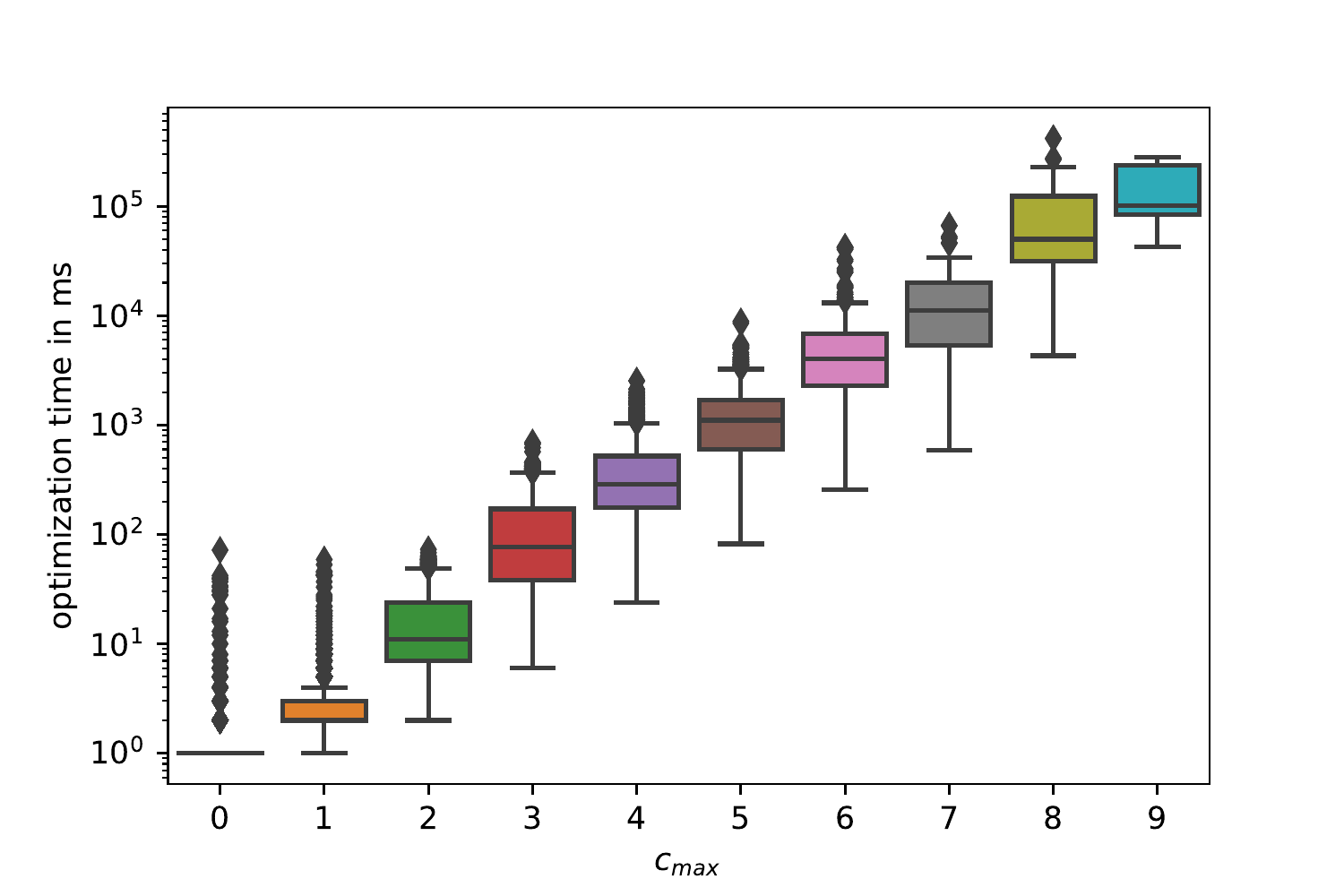}
\captionsetup{textformat=simple}
\caption{Optimization time depending on $c_\text{max}$}\label{fig:component_plot}
\end{minipage}
\end{figure}
For the experiments we used a machine with an \textit{AMD Ryzen 5 3600 6-Core Processor} clocked at \textit{4.4 GHz} and \textit{16 GB RAM}. We did not implement the geometric pre-processing as discussed in \cite{GudmundssonHK2002}. Our pre-processing has roughly cubic runtime ($3$--$4$ seconds per reconstruction). 
For larger orders $k$ we expect the optimization to dominate the runtime. 

The optimization time with respect to the order is given in Figure \ref{fig:order_plot}. Note that the optimization time for $k\leq 5$ is at most $30$ms. 
For $k=6$ the average runtime is still low with roughly $50$ms.  
For $k=7$ most datasets can be optimized in a few seconds, but some need around $20$ minutes for the optimization and five datasets reach a cut-off time of one hour.

The box-plot in Figure \ref{fig:component_plot} depicts the runtime with respect to the number of connected components $c_{\text{max}}$. The logarithmic scaling nicely illustrates the exponential increase.
If we also consider the distribution of $c_\text{max}$ for the different datasets and orders, we can easily connect the two box-plots. 
For $k\leq 4$ all of the datasets have $c_{\text{max}}\le 2$. 
Thus, the maximal runtime for orders $k\leq 4$ matches the worst runtime for $c_{max}\le 2$. 
For orders $k=5,6,7$ the $c_\text{max}$ distributions are illustrated in Figure \ref{fig:CCHistogram}. Note that for $k=5$ and $k=6$ most datasets still have $c_\text{max}\leq 2$  which results in the very low average runtime. For $k=7$ the distribution starts to shift towards higher $c_\text{max}$ which results in the higher average runtime.

\medskip

\noindent In summary, our experiments show that for our datasets we can compute $k$-OD min-error triangulations for $k \leq 6$ and also for $k = 7$ except for a few samples in reasonable time. 

\begin{figure}[tb]
\begin{minipage}[t]{.3\textwidth}
\centering
\includegraphics[width=1\textwidth]{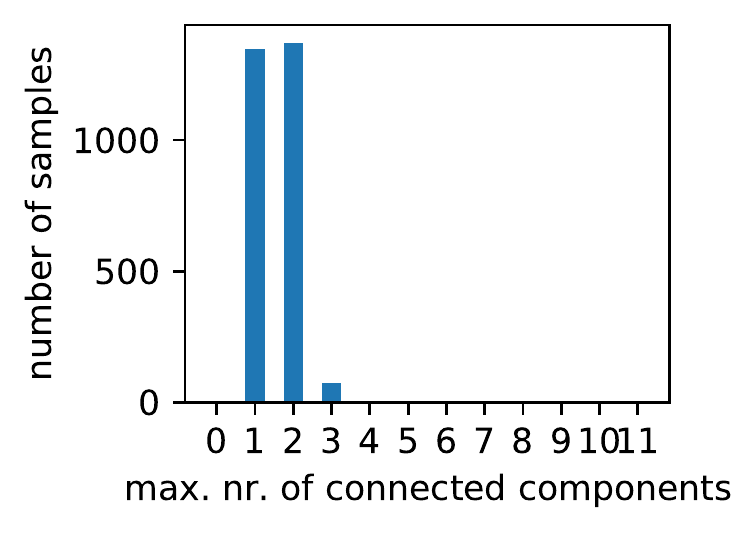}
\end{minipage}
\hfill
\begin{minipage}[t]{.3\textwidth}
\centering
\includegraphics[width=1\textwidth]{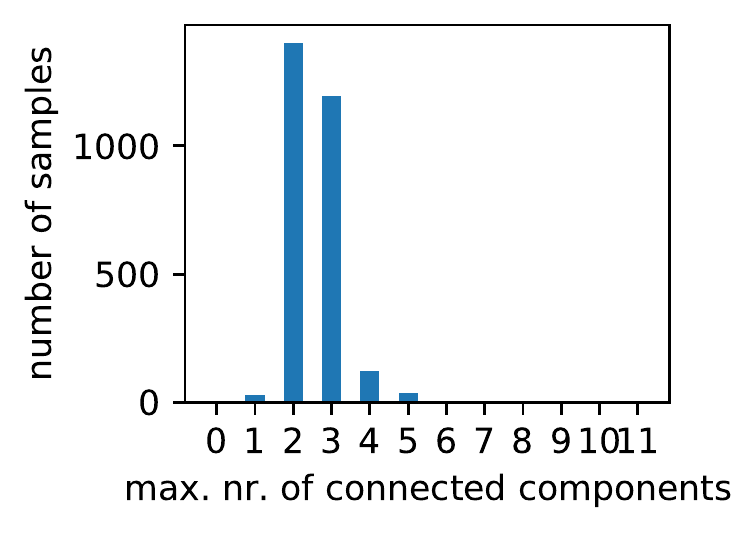}
\end{minipage}
\hfill
\begin{minipage}[t]{.3\textwidth}
\centering
\includegraphics[width=1\textwidth]{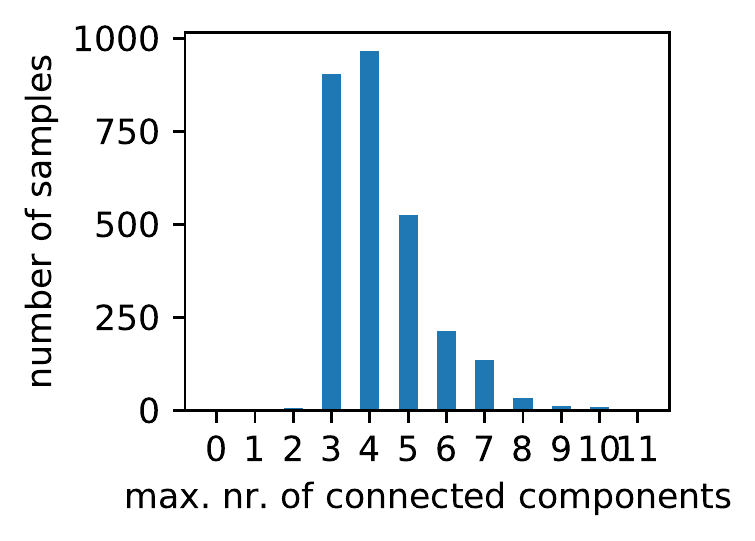}
\end{minipage}
\caption{The $c_\text{max}$ distribution of the reconstruction datasets for orders $k=5,6,7$}\label{fig:CCHistogram}
\end{figure}

\section{Conclusion}\label{sec:Conclusion}
We prove that it is NP-hard to approximate an optimal solution to the minimum-error triangulation problem. Our results also imply the inapproximability of the following generalization: minimizing the distance between $s_D$ and $h$ on $\mR$ for any metric on $\R^m$, especially the $L_p$-metric $\big(\sum_{r\in\mR}|s_D(r)-h(r)|^p\big)^{1/p}$ for $p\in[1,\infty)$ and the $L_\infty$-metric $\max_{r\in\mR}|s_D(r)-h(r)|$.
Additionally, we apply the dynamic programming algorithm by Silveira et al.~\cite{SilveiraK2009} to minimum-error triangulations and extend their experiments, regarding the fixed edges to a real world dataset. We further investigate the fixed-edge graphs for order $k=2$ and give a worst-case example for $k=3$. 
Finally, we perform the dynamic sea surface reconstruction similar to Nitzke et al.\ in \cite{NitzkeNFKH2021} for significantly larger datasets using a new algorithmic approach. 

A future line of research is the extension of the dynamic programming algorithm to datasets on the sphere, i.e., spherical triangulations. 
This would allow a more realistic reconstruction of the global dynamic sea surface. 
A combination with ILP techniques will be a further step~\cite{FeketeHL+2020}.
It would also be interesting to include multiple datasets for the learning of the reconstruction triangulation.
We believe that our work will open the door for the application of optimal triangulation approaches to the problem of multi-decadal global sea level reconstructions from tide gauge data.
In addition, with the growing amount of satellite and in-situ ocean sensors (buoys, Argo floats, ...) we see potential for a more widespread application of triangulation methods in generating gridded ocean data products.


\bibliography{Minimum_Error_Triangulations_for_Sea_Surface_Reconstruction}

\appendix
\section{Missing constructions and proofs of Section \ref{sec:NPHardness}}
\label{ap:np}
\subsection{The paraboloid}
\label{sec:paraboloid}
The graph of the function $f$ is the paraboloid \[\Gamma_f=\{(p_1,p_2,p_1^2+p_2^2)\mid (p_1,p_2)\in\R^2\}.\] Remember that a circle $C$ with radius $\rho$ around $x=(x_1,x_2)\in \R^2$ defines the function $h_C(r)=2x_1r_1+2x_2r_2-x_1^2-x_2^2+\rho^2$ for $r=(r_1,r_2)\in \R^2$. Let \[\Gamma_C=\{(r_1,r_2,h_C(r))\mid r=(r_1,r_2)\in \R^2\}\] denote the graph of $h_C$.

\begin{lemma} \label{lemma:lift_dist}
	For $y\in \R^2$ we have $f(y)-h_C(y)=\lVert y-x\rVert_2^2-\rho^2$.
\end{lemma}
\begin{proof}
	We have 
	\begin{align*}
	f(y)-h_C(y)	&=y_1^2+y_2^2-2x_1y_1-2x_2y_2+x_1^2+x_2^2-\rho^2\\
	&=(y_1-x_1)^2+(y_2-x_2)^2-\rho^2\\
	&=\lVert y-x\rVert_2^2-\rho^2.\qedhere
	\end{align*}
\end{proof}

\begin{lemma} \label{lemma:lift} 
	We have $\Gamma_C\cap\Gamma_f=\{(y_1,y_2,y_1^2+y_2^2)\mid (y_1,y_2)\in C\}$.
\end{lemma}
\begin{proof}
	We have 
	\begin{align*}
	\{(y_1,y_2,&y_1^2+y_2^2)\mid (y_1,y_2)\in C\} = \{(y_1,y_2,y_1^2+y_2^2)\mid (y_1-x_1)^2+(y_2-x_2)^2=\rho^2\}\\
	&=\{(y_1,y_2,y_3)\mid y_3-2x_1y_1-2x_2y_2+x_1^2+x_2^2-\rho^2=0, y_3=y_1^2+y_2^2\}\\
	&=\Gamma_C\cap \Gamma_f.\qedhere
	\end{align*}
\end{proof}

We are now able to prove Lemmas~\ref{lemma:circle_rep} and \ref{lemma:lift_innout}. 
\lemmaCircleRep*
\begin{proof}
	By Lemma~\ref{lemma:lift} we know that  $f(v)=h_{C_r}(v)$ for all $v\in \{s,t,u\}\cap C_r$. As $h_{C_r}$ is affine this means that $r$ is represented with zero error by $T$.
\end{proof}

\lemmaInnOut*
\begin{proof}
	We pick a convex combination $\lambda s+\mu t+\gamma u=r.$ As $T$ represents $r$ with zero error, we have 
	\[h_{C_r}(r)=\lambda f(s)+\mu f(t)+\gamma f(u).\]
	Since $h_{C_r}$ is affine also 
	\[h_{C_r}(r)=h_{C_r}(\lambda s+\mu t+\gamma u)=\lambda h_{C_r}(s)+\mu h_{C_r}(t)+\gamma h_{C_r}(u).\]  
	Combining the two equations we obtain 
	\[0	=\lambda (f(s)-h_{C_r}(s))+\mu (f(t)-h_{C_r}(t))+\gamma (f(u)-h_{C_r}(u)).\]
	Observe that $f(p)-h_{C_r}(p)<0$ for $p\in  I_{C_r}$ and $f(p)-h_{C_r}(p)>0$ for $p\in  O_{C_r}$ by Lemma~\ref{lemma:lift_dist}. We distinguish two cases: If points in $\{s,t,u\}\cap I_{C_r}$ and $\{s,t,u\}\cap  O_{C_r}$ appear with factor zero in the above equation, then $r\in \conv(\{s,t,u\}\cap C_r)$ contradicting our assumption. Otherwise the sets $\{s,t,u\}\cap I_{C_r}$ and $\{s,t,u\}\cap O_{C_r}$ must be non-empty.
\end{proof}	

We need some additional statement about the behavior of the error under orthogonal transformations and translations of the triangle and the reference point it is representing. 

\begin{lemma}
	\label{lemma:trans}
	Let $T\subset\R^2$ be a triangle representing $r\in \mR$ with error $\epsilon$. Applying a fixed orthogonal transformation or a translation on $T, r$ and $C_r$ preserves the error. In particular $r$ is still represented with error $\epsilon$ by $T$ after such operations.
\end{lemma}
\begin{proof}
	Let $s,t,u$ be the vertices of $T$. We pick a convex combination $\lambda s+\mu t+\gamma u=r$ and obtain 
	\begin{align*}
	\epsilon&=|\lambda f(s)+\mu f(t)+\gamma f(u)-h_{C_r}(r)|\\
	&=|\lambda (f(s)-h_{C_r}(s))+\mu (f(t)-h_{C_r}(t))+\gamma (f(u)-h_{C_r}(u))|.
	\end{align*}
	By Lemma~\ref{lemma:lift_dist} the last part depends on the radius of $C_r$ and the distance between its center and $s,t,u$. These values do not change after an orthogonal transformation or translation of $T,r$ and $C_r$. 
\end{proof}

\subsection{Analysis of gadgets}
\label{sec:gadgets}
We prove the property of a bit having an either positive or negative signal under certain conditions. We will show that the bit is well-behaved in the sense that any zero-error triangulation must have a triangulation edge that carries either a positive or a negative signal. We show this by analyzing all possible triangles formed by points on the integer grid, which represent $r$ with zero error.

\lemmaBit*
\begin{proof}
	Clearly, $D$ can contain at most one of $e_r^+$ and $e_r^-$. What needs to be proven, is that $D$ cannot do without, that is, we cannot have a triangle $T$ that represents $r$ with zero error and contains neither $e_r^-$ nor $e_r^+$ as one of its edges. 
	We first observe that we can assume $r=(0,0)$ as translation of $T,r$ and $C_r$ by $-r$ does not change the error by Lemma~\ref{lemma:trans}.
	
	Let $v(T)=\{s,t,u\}\subset \mS$ denote the vertices of $T$. Clearly, if $r\in \conv(v(T)\cap C_r)$, then $e_r^+$ or $e_r^-$ would be an edge of $T$. So, henceforth, we consider the case that $r\notin \conv(v(T)\cap C_r)$. Lemma~\ref{lemma:lift_innout} now tells us that $v(T)\cap I_{C_r}\neq\emptyset$. Since the construction and the error are invariant under rotation (except the labeling of $e_r^+$ and $e_r^-$, which may be switched) by Lemma~\ref{lemma:trans}, we may assume, without loss of generality, that this is the point $t=(0,1)$. 
	
	We lift our construction into the $3$-dimensional space. That is for a point $p=(p_1,p_2)\in \R^2$ we denote with $p'=(p_1,p_2,f(p))$ its lift on the paraboloid $\Gamma_f$ and analogously we define the lift of a set $M\subset \R^2$ as $M'=\{p'\mid p\in M\}$. 
	Let $E$ denote the plane that contains $v(T)'$. As $(0,1)\in v(T)$ we know that $(0,1,1)\in v(T)'\subset E$. Furthermore $(0,0,h_{C_0}(0))=(0,0,2)\in E$ as $T$ represents $(0,0)$ with zero error. Thus a point $(x_1,x_2,x_3)$ on $E$ must satisfy $2a x_1 - x_2 - x_3 + 2 = 0$, for some fixed $a$.
	
	We have $v(T)'\subset E\cap \Gamma_f$ and thus the remaining points of $v(T)$ must lie on the circle described by $x_1^2 + x_2^2 = 2 a x_1 - x_2 + 2$. That is, the circle $C_T$ with center $(a, -1/2)$ and squared radius $a^2 + 9/4$.
	Since the construction and the error are invariant under reflection (except the labeling of $e_r^\pm$), we may now assume, without loss of generality, that $a$ is non-negative. As $T'$ must include $r'$, $T$ must include at least one point $u=(u_1,u_2)$, different from $t$, such that $u_1 \leq 0$. We will now investigate all possible locations of $u$. Remember that $\mS\subset \Z^2$ so $u$ must have integral coordinates.
	
	\textbf{Case 1:} Suppose that $u_1 \leq -2$. The first coordinate of any point of $C_T$ is at least $a - \sqrt{a^2 + 9/4}$. For $a \geq 0$, this expression grows with $a$, starting from $-3/2$ for $a = 0$. Thus such $u$ cannot exist.
	
	\textbf{Case 2:} Suppose that $u_1 = -1$.  Then the circle equation reads $1 + u_2^2 = -2a - u_2 + 2$, so $(u_2 + 1/2)^2 = 5/4 - 2a$. For $a \geq 0$, this implies $|u_2 + 1/2| < \sqrt{5/4}$, and therefore the only candidate for $u$ is $(-1,-1)$ (as $(-1,0)$ is a forbidden point), with $a = 1/2$. Note that $u$ lies on $C_r$. Now the third point of $v(T)$ must lie to the right of vertical axis, on $C_T$, that is, on the circle with center $(1/2,-1/2)$ and radius $\frac12 \sqrt{10}$. Here the only candidates with integer coordinates are $(1,1)$ (but then $e_r^+$ would be an edge of $T$), $(2,0)$ (which is forbidden and thus not in $\mS$), $(2,-1)$ (which is invalid because $T$ would then contain a fourth triangulation point $(0,-1)$) and $(1,-2)$ (which is invalid for the same reason). Therefore we cannot have $u_1 = -1$.
	
	\textbf{Case 3:} Finally, suppose that $u_1 = 0$. Now we must have $u = (0,-1)$, since $u_2 > 0$ would imply that $T$ does not contain $r$, whereas $u_2 < -1$ would imply that $T$ contains $(0,-1)$ as a fourth triangulation point. But if $t' = (0,1,1)$ and $u' = (0,-1,1)$ are both vertices of $T'$, then $r' = (0,0,2)\notin T'$ and thus we obtain a non-zero error at $r$. Therefore we cannot have $u_1 = 0$.
	
	It follows that every triangle $T$ that represents $r$ with zero error contains either $e_r^+$ or $e_r^-$.
\end{proof}
We prove a similar property on the wire segment and multiplier segment.
\begin{lemma} \label{lemma:line_mult_seg}
	Suppose the instance contains a wire/multiplier segment and let $\RR$ be the reference points of this segment. If $\mS\subset \Z^2$ and $\mS$ does not contain forbidden points of the segment, any triangulation $D$ of $\mS$ with $\Err_D(\RR)=0$ is either positive or negative on $\RR$. 
\end{lemma}
\begin{proof}
	The wire segment connecting the points  $(x_1,x_2),(y_1,y_2)\in \Z^2$ is completely built from bits. By Lemma~\ref{lemma:bit} such a bit must have an either positive or negative signal at its reference point. It is left to show that the signal is either positive or negative on the complete segment. Suppose this is not the case and that $x_1=y_1$ (the other case follows analogously). Then there must be two reference points $r,q\in\RR$  with $r=q+(0,1)$ and the signal at $r$ being different from the signal at $q$. This is not possible as $e_r^+$ and $e_q^-$ intersect each other and so do $e_r^-$ and $e_q^+$.
	
	The situation is more sophisticated when considering a multiplier segment at $x\in\Z^2$. By Lemma~\ref{lemma:trans} we can assume that $x=(0,0)$. We use Lemma~\ref{lemma:bit} to see that the signal on the reference points of bits must be either positive or negative. Thus $D$ must contain one of the edges $e_r^{\pm}$ for every reference point $r\in \{\pm(0,2),\pm(2,0)\}$. Let $F$ be any set of edges that consists of the mandatory edges of the multiplier segment and at least one of the edges $e_r^{\pm}$ for each $r\in \{\pm(0,2),\pm(2,0)\}$. These edges isolate the inner reference points from the remaining instance as every triangle that contains one of the inner points and contains a point outside of the segment must intersect at least one of the edges of $F$, regardless of which of the sixteen possibilities for $F$ is chosen.
	
	Let $T$ be a triangle in $D$ representing an inner point $r$ with zero error. The multiplier segment is invariant under rotation (except the labeling of positive and negative). Furthermore rotation does not change the error at $r$ by Lemma~\ref{lemma:trans}. Thus we can fix $r$ to be $(-1,0)$. 
	
	We claim that $T$ contains one of $e_r^{\pm}$ as an edge. We already observed that the vertices $v(T)$ of $T$ consist of triangulation points from the multiplier segment at $(0,0)$. 
	If $r\in \conv(v(T)\cap C_r)$ we see that one of $e_r^\pm$ is an edge of $T$. If this is not the case we apply Lemma \ref{lemma:lift_innout} to see that $v(T)\cap  I_{C_r}\neq\emptyset\neq v(T)\cap  O_{C_r}$. We enumerate all possibilities for such $T$.
	
	\textbf{Case 1:} Assume that $t=(-1,1)\in v(T)$. Then $(-1,-1)\notin v(T)$, as $h_{C_r}(r)=5\neq 2  =\frac{1}{2}(f((-1,1))+f((-1,-1)))$. Figure \ref{fig:mult_seg_pf2} shows all possibilities to choose the second point $u$ of $v(T)$ such that $\conv(r,t,u)\backslash \{r,t,u\}$ does not contain triangulation points or intersect mandatory edges. Among these points there are four points from $ O_{C_r}$. As we know that $v(T)$ must contain at least one of these, we consider all the cases where we choose one of them as the second point $u$. Figure \ref{fig:mult_seg_pf2} shows all possibilities choosing the last point of $v(T)$ depending on the choice of $u$. In two cases it is not possible to build $v(T)$ with $r\in T$. In the remaining two cases there is exactly one possibility to build $v(T)$ with $r\in T$. In one case $v(T)=\{(-1,1),(-2,1),(1,-3)\}$ and $r=\frac{1}{4}(-1,1)+\frac{1}{2}(-2,1)+\frac{1}{4}(1,-3)$ but \[h_{C_r}(r)=5\neq \frac{11}{2}=\frac{2}{4}+\frac{5}{2}+\frac{10}{4}=\frac{1}{4}f((-1,1))+\frac{1}{2}f((-2,1))+\frac{1}{4}f((1,-3)).\]
	In the other case $v(T)=\{(-1,1),(-3,-1),(1,1)\}$ and $r=\frac{1}{2}((-3,-1)+(1,1))$ but 
	\[h_{C_r}(r)=5\neq 6 =\frac{1}{2}(f((-3,-1))+f((1,1))).\] 
	Thus in both cases we get a contradiction to $T$ representing $r$ with zero error.
	
	\textbf{Case 2:} For $t=(1,1)\in v(T)$ we do the same and obtain two possibilities to choose a point from $O_{C_r}$. Both are depicted in Figure \ref{fig:mult_seg_pf2}. In  one case it is not possible to build $v(T)$ with $r\in T$. In the other case we calculate as in case 1 that $T$ does not represent $r$ with zero error.
	
	The remaining cases $t=(\pm1,-1)$ can be shown analogously. The computations do not change as $f$ is invariant under reflection. We conclude that $D$ contains one of $e_r^\pm$ for all $r\in\RR$ and must be either positive or negative on the whole gadget. 
\end{proof}

\begin{figure}
	\centering
	\includegraphics[scale=0.65, angle=90]{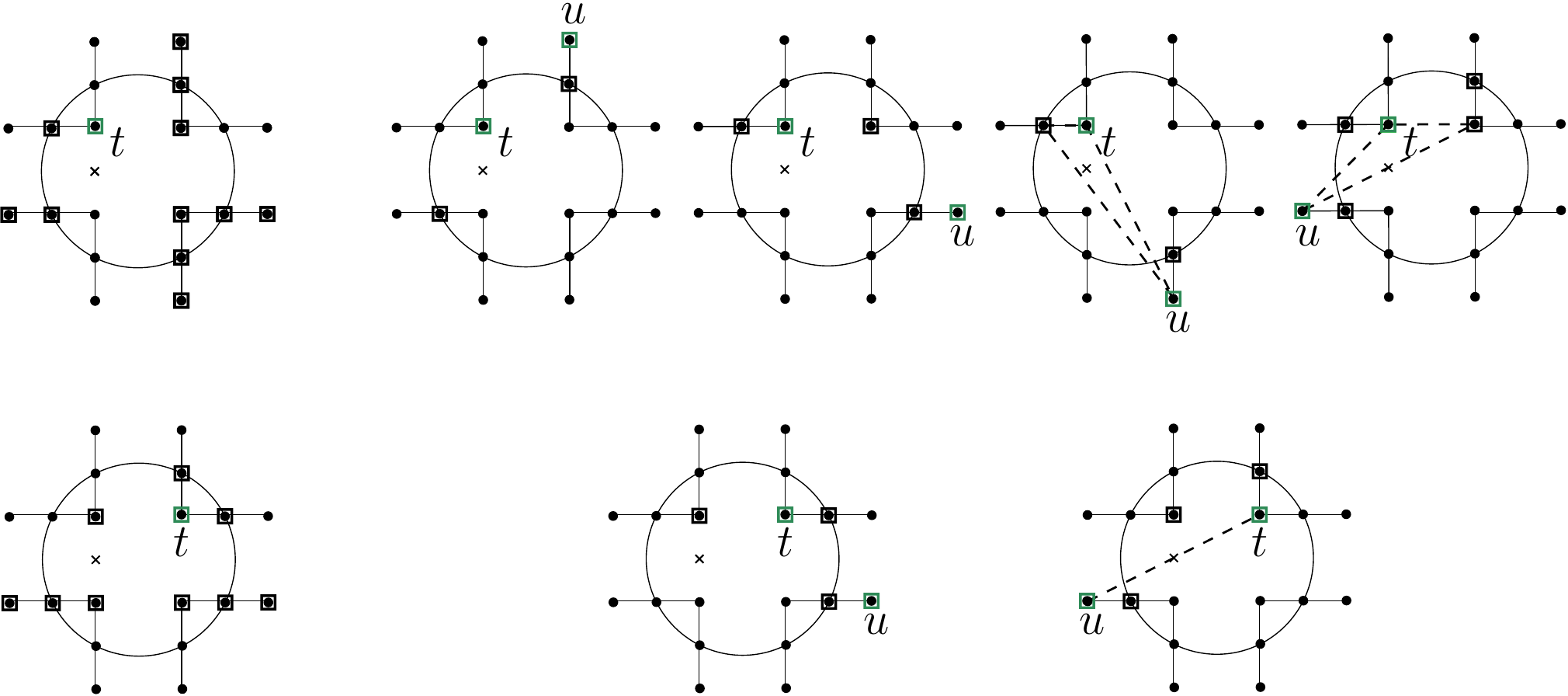}
	\caption{Possibilities to build $v(T)$ starting with $t=(\pm1,1)$. The points outlined in green are currently assumed to be in $v(T)$. All not outlined points cannot be in $v(T)$.}
	\label{fig:mult_seg_pf2}
\end{figure}

\lemmaWireVar*
\begin{proof}
	The signal at a segment that is part of the gadget must be either positive or negative by Lemma~\ref{lemma:line_mult_seg}. If it is connected to another segment at one of its anchor points, this anchor point determines the signal at both segments, which must equal the signal at the anchor point. Proceeding like this we see that the signal must be either positive or negative on the whole gadget. 
\end{proof}
Having the variable gadget and wire gadget in place we need two more constructions, namely the clause gadget and the negation gadget. Both are very similar to each other.

We explain how to build the clause gadget representing a clause of the form $\overline{v_1}\vee\overline{v_2}\vee v_3$ at a point $c\in \Z^2$. For simplicity we assume that $c=(0,0)$.
We declare the points from $\{(5,-15),(\pm15,-5), (\pm9,13)\}$ as triangulation points. Notice  that they all lie on one circle $C_r$ centered at $(0,0)$ with radius $\sqrt{250}$. We declare $r=(0,11)$ as reference point with coupled circle $C_r$. 
This reference point is special as it does not come with a positive and negative edge, instead we observe that it is represented with zero error by the following three triangles
\begin{align*}
T_1&=\conv((5,-15), (\pm9,13))\\
T_2&=\conv((15,-5), (\pm9,13))\\
T_3&=\conv((-15,-5),(\pm9,13)).
\end{align*}
This is true by Lemma \ref{lemma:circle_rep} and $r\in T_i$ for $i=1,2,3$.
A clause combines three values and so does the clause gadget. Every triangle $T_i$ belongs to a reference point $r_i$. A triangle $T$ is blocked by an edge $e$ if both cannot be part of the same triangulation. This is the case if $e$ is not an edge of $T$ and $e\cap T\neq \emptyset$. The triangle $T_i$ is blocked by the positive edge of $r_i$ for $i=1,2$ and by the negative edge of $r_i$ for $i=3$.

Let $a_1=(-12,-17)$ and define $r_1$ to be the intersection of the two edges $e_{r_1}^+=\conv(a_1+(1,-1),a_1+(23,4))$ and $e_{r_1}^-=\conv(a_1+(1,1),a_1+(23,-4))$. Thus we have $r_1=a_1+(\frac{27}{5},0)$.
We declare the vertices of $e_{r_1}^+$ and $e_{r_1}^-$ as triangulation points. Observe that they lie on a common circle $C_{r_1}$, which is the circle coupled to $r_1$. Furthermore we add three horizontal bits, one at each of the points $a_1+(l,0)$ for $l=0,1,2$ and declare $a_1$ as anchor point.

A similar construction is done at the anchor point $a_2=(17,12)$ by reflecting the above construction in the line with slope -1 through $(0,0)$, and at the anchor point $a_3=(-17,12)$ by rotating the construction at $a_1$ clockwise by $\frac{\pi}{2}$. In the construction at $a_3$, we swap the definitions of $e_{r_3}^+$ and $e_{r_3}^-$, so that $e_{r_i}^+$ has positive slope and $e_{r_i}^-$ has negative slope for all $i$.

Let $\widetilde{\mR}$ be the reference points and $\widetilde{\mS}$ be the triangulation points of the clause gadget. A point is forbidden if it is in
\[\bigcup_{r\in\widetilde{\mR}}(C_r\cup I_{C_r})\backslash \widetilde{\mS}\]
or it is already forbidden in one of its bits. Moreover we need some mandatory edges to isolate the clause gadget. They are depicted in Figure \ref{fig:klausel}, where we can find the whole construction. 

\lemmaKlausel*
\begin{proof}
	The signal at a reference point of a bit must be either positive or negative by Lemma~\ref{lemma:bit}. This also includes the anchor points $a_1,a_2,a_3$.
	
	Suppose that $T_1$ is in $D$ and the signal at $a_1$ is positive. We show that the error at $r_1$ is positive contradicting the assumption that $D$ is a zero-error triangulation.
	
	Let $T$ be the triangle in $D$ representing $r_1$ and let $v(T)$ denote its vertices. As $T_1$ belongs to $D$ and the signal at $a_1+(2,0)$ is positive, we know that $e_{r_1}^\pm$ cannot be edges of $T$. Thus by Lemma~\ref{lemma:lift_innout} we know that $v(T)$ has a non-empty intersection with $I_{C_{r_1}}$ and $O_{C_{r_1}}$. Furthermore $v(T)$ must contain a point below the line that supports $e_{r_1}^-$, so $v(T)$ must contain at least one of the points $a_1+(l,-1)$ for $l=1,2,3$. Let $t$ be this point. Choosing the second point $u$ in $v(T)$ from $O_{C_{r_1}}$ already yields a contradiction, because, for any choice of $u$ from $O_{C_{r_1}}$, the hull $\conv(t,u,r_1)$
	contains another triangulation point or intersects a mandatory edge, the triangle $T_1$,
	or the positive edge of $a_1 + (2,0)$.
	
	Analogously one can prove that $T_2$ and $e_{a_2}^+$ or $T_3$ and $e_{a_3}^-$ cannot be simultaneously in $D$. Figure \ref{fig:klausel_pf} illustrates how to exclude all three combinations.
	Since $r$ is triangulated with zero error by $D$ one of the triangles $T_1,T_2,T_3$ must be in $D$. Thus the signal at one of $a_1,a_2$ must be negative or the signal at $a_3$ must be positive.
\end{proof} 
\begin{figure}
	\centering
	\includegraphics[scale=0.7]{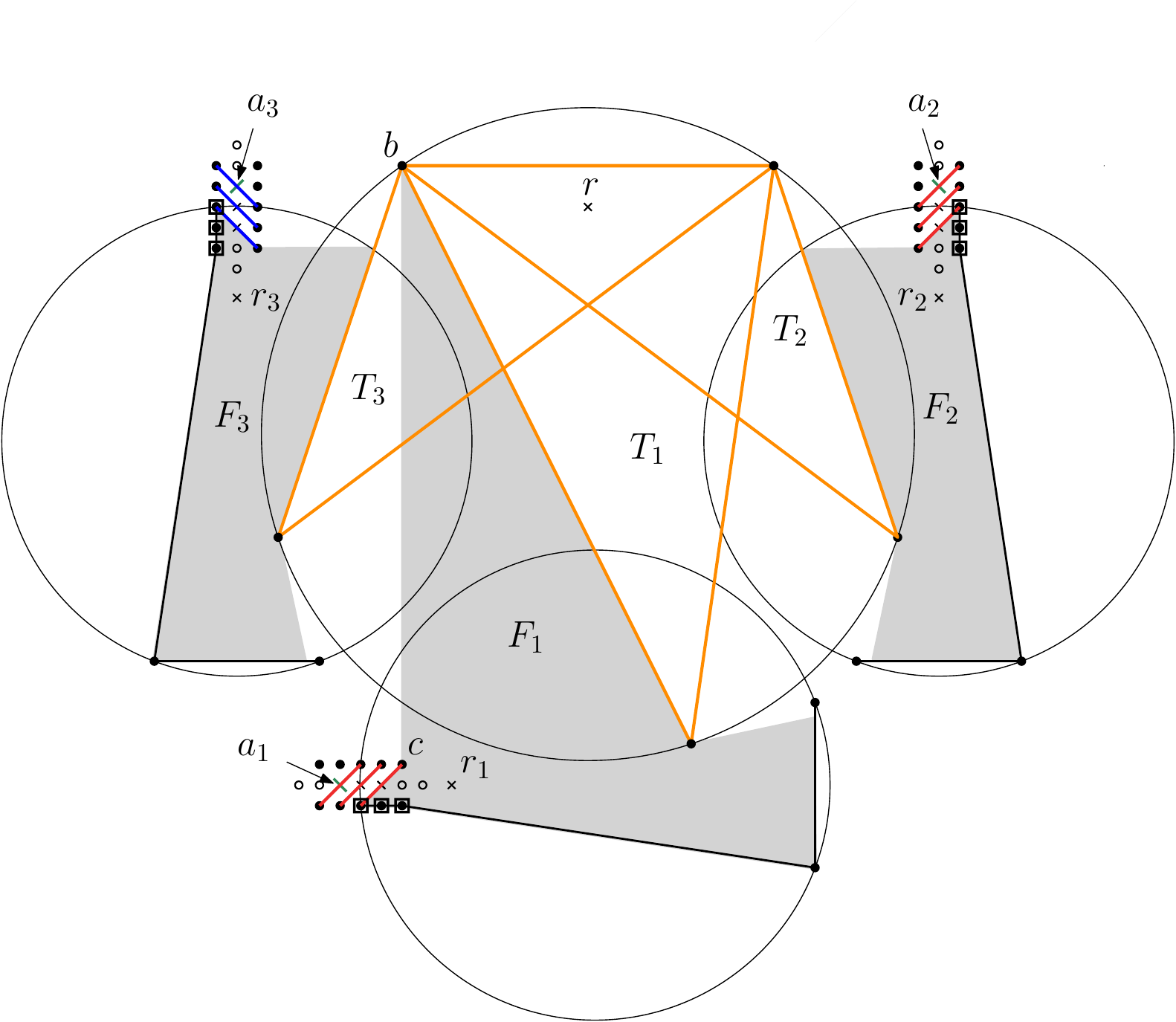}
	\caption{Here we see all three cases together. For $i\in\{1,2,3\}$, if we choose the first point $t$ in $v(T)$ to be one of the three marked points near $a_i$, the remaining points of $v(T)$ must come from $F_i$, otherwise $T$ would intersect other triangulation points or edges. However, $F_i$ does not contain any points outside $C_{r_i}$ (note that $b = (-9,13)$ lies just outside $F_1$, as $\conv(t,b,r_1)$ would include $c = (-9,-16)$ if $t$ is any of the marked points near $a_1$.}
	\label{fig:klausel_pf}
\end{figure} 
\begin{figure}
	\centering
	\includegraphics[scale=0.8]{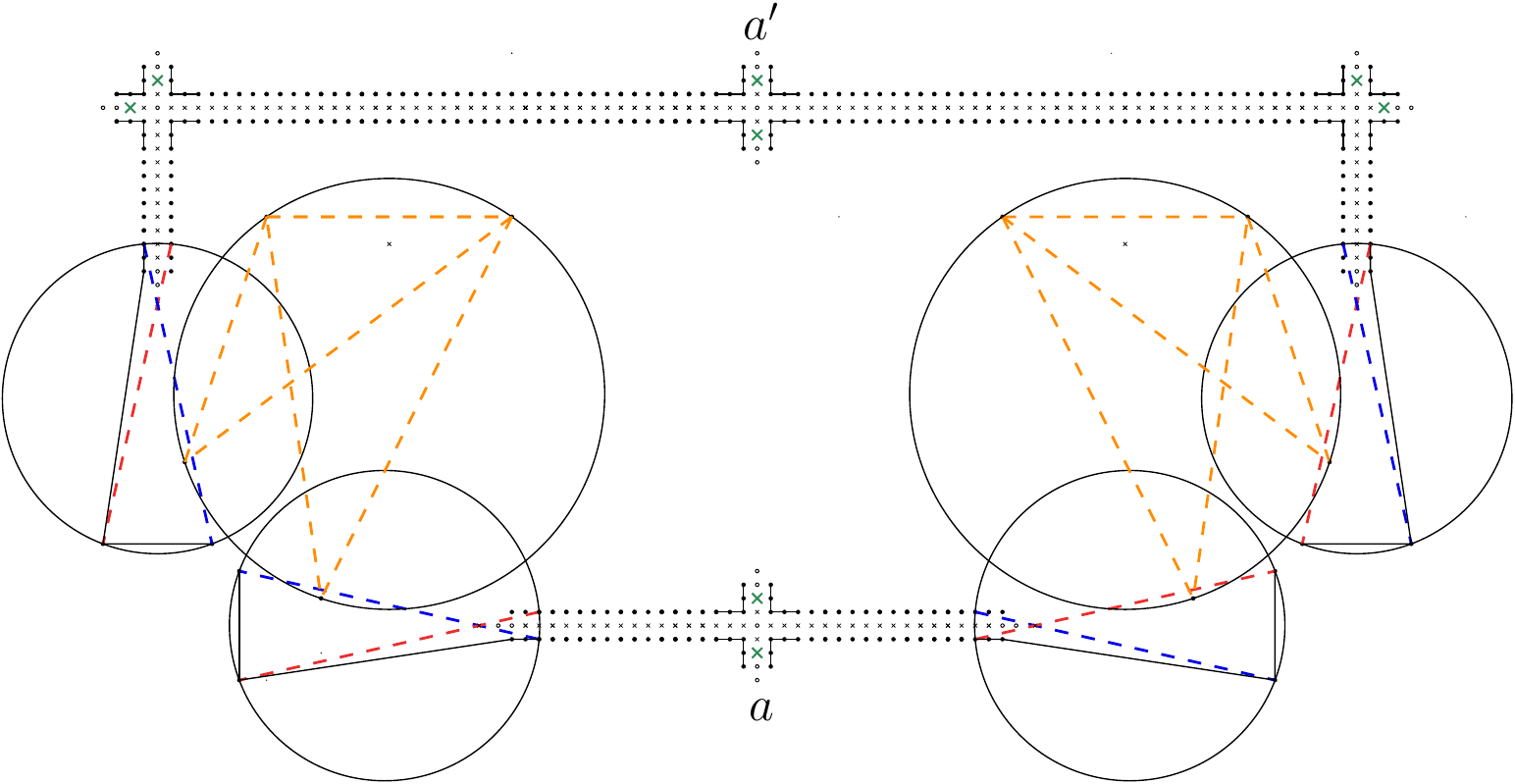}
	\caption{The negation gadget.
		If the signal at anchor point $a$ is negative it is negated in the left segment. If the signal at $a$ is positive it is negated in the right segment. Since the top wire carries a consistent signal, negation is ensured at $a'$.}
	\label{fig:negation}
\end{figure}
The last part of our construction is the \emph{negation gadget}. The core components of the negation gadget are the positive and negative negation segments. For the positive negation segment at a point $c$ we follow the construction of the clause gadget, except that we omit the triangle $T_3$ and the construction at $a_3$. The negative negation segment at $c$ is then a reflection of the positive negation segment at the vertical line through $c$, with switched definitions of positive and negative edges. For the negation gadget at a point $x\in \Z^2$ we place a multiplier segment at $x$, a positive negation gadget at $x+(27,17)$ and a negative negation gadget at $x+(-27,17)$. The anchor point $x+(2,0)$ of the multiplier segment is then connected via a wire gadget to the lower anchor point $x+(15,0)$ of the positive negation segment. The anchor point $x-(2,0)$ of the multiplier segment is connected via a wire gadget to the lower anchor point $x-(15,0)$ of the negative negation segment. Furthermore we place a multiplier segment at $x+(0,38)$ and connect the anchor point $x+(2,38)$ of this multiplier segment via a wire gadget to the upper anchor point $x+(44,29)$ of the positive negation segment. Finally we connect the anchor point $x+(-2,38)$ of this multiplier segment via a wire gadget to the upper anchor point $x+(-44,29)$ of the negative negation segment.
Figure \ref{fig:negation} visualizes the construction. 

We analyze the signal at the anchor points $a=x-(0,2)$ and $a'=x+(0,40)$:
\begin{lemma}
	\label{lemma:neg}    
	Suppose the instance contains a negation gadget at $x\in \Z^2$ and let $\RR$ be the reference points of this gadget. Let $\mS\subset \Z^2$ and assume $\mS$ does not contain forbidden points of the gadget. Any triangulation $D$ of $\mS$ with $\Err_D(\RR)=0$ is positive at $a$ iff it is negative at $a'$.   
\end{lemma}
\begin{proof}
	We consider the positive negation segment at point $x+(27,17)$, which equals the clause gadget without the construction at $a_3$ and $T_3$. We borrow the notation from the clause gadget. As in the proof of Lemma~\ref{lemma:clause} one can show that neither $T_1$ and $e_{r_1}^+$ nor $T_2$ and $e_{r_2}^+$ can simultaneously be in $D$. As one of $T_1,T_2$ is in $D$ this means that at least one of $a_1=x+(15,0),a_2=x+(44,29)$ has a positive signal.  
	Analogously at most one of the signals at the anchor points $a_1'=x-(15,0),a_2'=x+(-44,29)$ of the negative negation segment at $x+(-27,17)$ is negative.
	
	Suppose that the signal at $a$ is positive. Then by Lemma~\ref{lemma:wire} the signal at $a_1$ must also be positive. By the observation above the signal at $a_2$ must then be negative and so must be the signal at $a'$ by Lemma~\ref{lemma:wire}. If the signal at $a$ is negative the signal must be positive at $a_2'$ and $a'$ following the same arguments.    
\end{proof}

\subsection{Replacing mandatory edges}
\label{sec:edges}
Before we dedicate ourselves to the proof of Theorem \ref{thm}, it is left to drop the restriction that mandatory edges must be in any feasible triangulation, as this does not match the original definition of the zero-error triangulation problem. We slightly modify the previously constructed gadgets as follows. Let $e=\overline{st}$ be a mandatory edge in a gadget. We remove $e$ from the set of mandatory edges and instead add the reference point $r_e=\frac{1}{2}(s+t)$ to the gadget.

It is left to define the circle $C_{r_e}$ coupled to $r_e$. Notice that we would like to enforce the edge $e$ to be in \emph{every} zero-error triangulation of the gadget. Suppose that
\begin{enumerate}
	\item \label{man_edges_1} $\{s,t\}\subset C_{r_e}$ and
	\item \label{man_edges_2} $C_{r_e}\cup I_{C_{r_e}}$ does not contain further triangulation points.
\end{enumerate} 
Then any triangle with vertices $s,t$ represents $r_e$ with zero error by Lemma~\ref{lemma:circle_rep} and any triangulation which does not contain $e$ has positive error at $r_e$ by Lemma~\ref{lemma:lift_innout}. 

In most cases it is sufficient to define $C_{r_e}$ as the circle centered at $r_e$ with radius $\frac{\lVert s-t \rVert_2}{2}$. However this definition does not work for the three long edges of the clause gadget and the four long edges of the negation gadget, as for such an edge $e$ the set $I_{C_{r_e}}$ would contain triangulation points of the gadget. In this case let $Q_1$ and $Q_2$ be the two squares which contain $e$ as one of their edges. One of these squares contains triangulation points of the gadget other than $s,t$, while the other one does not. Let $Q_1$ be the square which does not contain any triangulation points other than $s,t$. We define $C_{r_e}$ as the circumcircle of $Q_1$. Then \ref{man_edges_1}. and \ref{man_edges_2}. are both satisfied for $C_{r_e}$. Finally we extend the set of forbidden points by $(C_{r_e}\cup I_{C_{r_e}})\backslash \{s,t\}$.
The following corollary is an immediate consequence of Lemma~\ref{lemma:circle_rep} and Lemma~\ref{lemma:lift_innout}. 
\begin{corollary}
	\label{corr:edges}
	Suppose the instance contains a gadget with reference points $\widetilde{\mR}$, including the reference points which replace the mandatory edges. If $\mS\subset \Z^2$ and $\mS$ does not contain forbidden points of the gadget, any triangulation $D$ of $\mS$ with $\Err_D(\widetilde{\mR})=0$ contains all mandatory edges of this gadget. 
\end{corollary}

\subsection{The reduction}
\label{sec:theorem}
\begin{figure}
	\centering
	\includegraphics[scale=0.45]{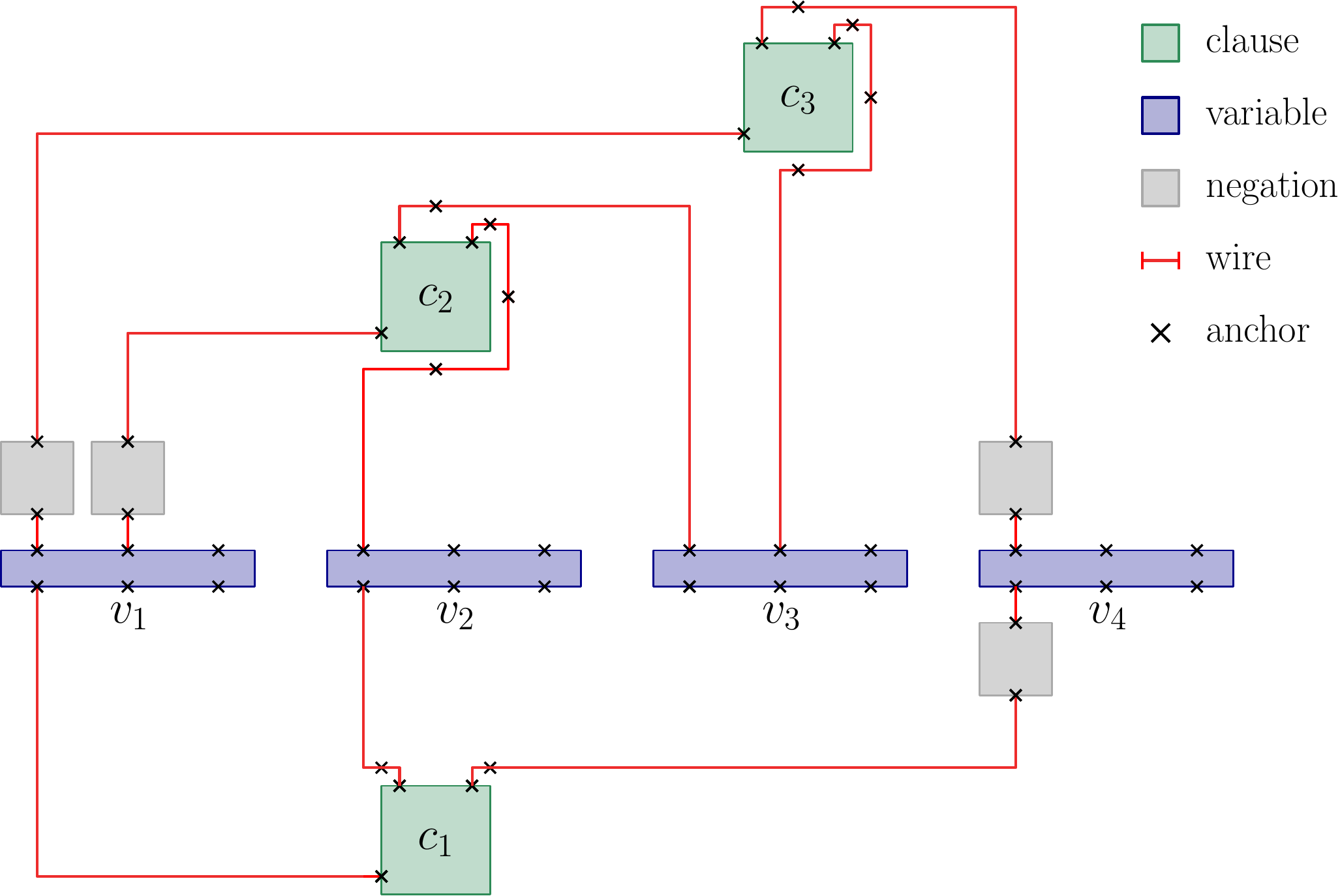}
	\caption{The triangulation instance corresponding to the 3SAT instance with clauses $c_1=\overline{v_1}\vee v_2\vee v_4, c_2=v_1\vee\overline{v_2}\vee v_3$ and  $c_3=v_1\vee\overline{v_3}\vee\overline{v_4}$. The anchor points at which we connect two gadgets are depicted as crosses. Notice that some of the anchor points at the variable gadgets may be left unused. }
	\label{fig:construction}
\end{figure}

Given an instance $\mathcal I$ of the planar 3SAT problem, with $V$ the set of variables and $K$ the set of clauses, we first explain how to construct the corresponding instance $\mathcal I_{\textup{err}}$ of the zero-error triangulation problem. Let $k=|K|+|V|$. We fix an integral rectilinear embedding of the 3SAT instance on the plane and scale it by a factor $\gamma\in O(k)$. Notice that the scaled embedding is still rectilinear. Let $G(v)$ denote the center of the box belonging to a variable $v\in V$ and $G(c)$ the vertex belonging to a clause $c\in K$ of the scaled embedding. Recall that $G(v)$ lies on the horizontal axis for all $v\in V$.

The zero-error triangulation instance is constructed as follows: We place a variable gadget at $G(v)$ for all $v\in V$ and a clause gadget at $G(c)$ for all $c\in K$.
For a clause $c\in K$ containing the variables $v_1,v_2,v_3$ we do the following: 
Notice that $G(v_1),G(v_2),G(v_3)$ lie on the horizontal axis and we assume that they appear on the axis from left to right in this order. 

If $G(c)$ lies above the horizontal axis, we connect the anchor point $a_i$ to an anchor of the variable gadget at $G(v_i)$ for all $i\in\{1,2,3\}$. If $G(c)$ lies below the horizontal axis we connect $a_1$ to  an anchor of the variable gadget at $G(v_1)$, $a_2$ to  an anchor of the variable gadget at $G(v_3)$ and $a_3$ to  an anchor of the variable gadget at $G(v_2)$. 
This is done by wire gadgets in such a way that the wire gadgets do not overlap each other (this is possible, because the embedding is planar and rectilinear).
However, if a variable that appears negated in the clause is connected to the $a_3$ anchor of the clause gadget, or if a variable that appears non-negated in a clause is connected to the $a_1$ or $a_2$ anchor of the clause gadget, then we do not connect the clause gadget directly to the variable gadget, but we insert a negation gadget: we use wire gadgets to connect the anchor of the clause gadget to the $a'$ anchor of the negation gadget, and a wire gadget to connect the $a$ anchor of the negation gadget to an anchor of the variable gadget (if $G(c)$ lies below the horizontal axis we first rotate the negation gadget by $\pi$). If we choose the distance $\alpha$ between multiplier segments in a variable gadget to be $\geq 200$ this construction can be done without the negation gadgets overlapping each other. Figure~\ref{fig:construction} shows the structure of the zero-error triangulation instance corresponding to our initial example.

Let $\mS$ be the set of triangulation points and $\mR$ the set of reference points of $\mathcal I_{\textup{err}}$. Notice that $\mS$ is contained in $\Z^2$ by construction. Furthermore we want to establish the property that $\mS$ does not contain any forbidden points. This is already true for each of the discussed gadgets. Remember that we scaled the rectilinear embedding of $\mathcal I$ by a factor $\gamma\in O(k)$ (the factor comes from the width of the variable gadget, which is in $O(k)$). If we pick $\gamma$ sufficiently large (e.g., $\gamma=1000k$) the gadgets do not overlap (excluding the overlap that occurs when two gadgets are connected at anchor points, which is explicitly allowed). Thus the instance does not contain forbidden triangulation points. 

We are now able to prove the hardness of the zero-error triangulation problem.
\mainTheorem*
\begin{proof}
	Let $\mathcal I$ be an instance of the planar 3SAT problem and let $\mathcal I_{\textup{err}}$ denote the corresponding instance of the zero-error triangulation problem. 
	
	Suppose that there exists an assignment of the variables under which the planar 3SAT formula is satisfied. For every reference point which replaces a mandatory edge $\overline{st}$ we add $\overline{st}$ to the triangulation $D$. By Corollary~\ref{corr:edges} the error of $D$ at such reference points is zero. For the other reference points we fix an assignment  under which the 3SAT formula is satisfied and define the triangulation $D$ of $\mS$ on the variable gadget at $G(v)$ to be positive if the value of $v\in V$ is one and negative if the value of $v$ is zero. Observe that by Lemma~\ref{lemma:circle_rep} a negation/wire gadget can be triangulated with zero error with a fixed signal at one of its anchor points. There exists a zero-error triangulation of the negation/wire gadget having the negated/same signal on the remaining anchor points. We extend $D$ on the negation/wire gadgets following the above observation. Now consider a clause gadget at $G(c)$ for some $c\in K$ and its three anchor points $a_1,a_2,a_3$ whose signals in $D$ are already determined by the wire gadgets connected to them. As clause $c$ is satisfied under the assignment, one of $a_1,a_2$ has a negative signal or $a_3$ has a positive signal. Thus at least one of the triangles triangulating $r=G(c)+(0,11)$ with zero error can be added to $D$. 
	
	For the other direction suppose there is a triangulation $D$ of $\mS$ with zero error. First observe that the mandatory edges must belong to $D$ by Corollary~\ref{corr:edges}.
	For $v\in V$ the triangulation must be either positive or negative on the variable gadget at $G(v)$ by Lemma~\ref{lemma:wire}. We assign to $v$ the value $1$ if $D$ is positive on the variable gadget at $G(v)$ and $0$ if it is negative. On all wire gadgets directly connected to a variable gadget, the triangulation must be either positive or negative by Lemma~\ref{lemma:wire}. If the triangulation is positive on a variable gadget, then it must be negative on all wire gadgets connected to it through a negation gadget and vice versa by Lemma~\ref{lemma:neg}. Lemma~\ref{lemma:clause} then guarantees that the 3SAT formula is satisfied under this assignment.
	
	It is left to show that the reduction works in polynomial time. The planar 3SAT formula can be embedded in polynomial time on an integral grid of size $O(k)\times O(k)$ \cite{Lichtenstein}. Scaling the embedding by $\gamma\in O(k)$ and constructing the set of triangulation points $\mS$ and the set of reference point $\mR$ can be done in polynomial time. The same holds for the computation of $f(p_1,p_2)=p_1^2+p_2^2$, as all triangulation points are integral. For the reference values we consider a reference point $r=(r_1,r_2)\in\mR$ and it coupled circle $C_r$ centered at a point $x=(x_1,x_2)$ with radius $\rho$. Recall that 
	\[h(r)=h_{C_r}(r)=2x_1r_1+2x_2r_2-x_1^2-x_2^2+\rho^2.\]
	We distinguish two cases.
	If $r$ was added to the instance to replace a mandatory edge $e=\overline{st}$ and $C_{r}$ is centered at $r$, we get that $h(r)=r_1^2+r_2^2+\frac{\lVert s-t\rVert_2^2}{4}$. Thus $h(r)$ can be computed in polynomial time.  
	Otherwise we know that $C_{r}$ contains at least three integral points $x,y,z$. It is a known fact that we can compute the squared radius and the center of such a circle in time polynomial in $x,y,z$. Hence we can compute $h(r)$ in polynomial time.   
	Thus the zero-error triangulation instance can be constructed in polynomial time from the planar 3SAT instance. 
	
	Every polynomial-time approximation algorithm to the minimum-error triangulation problem yields a polynomial-time algorithm to the zero-error triangulation problem. As the zero-error triangulation problem is NP-hard such a polynomial-time approximation algorithm does not exist unless P=NP.  
\end{proof}
\section{Fixed-edge graph experiments on random data}\label{sec:RandomExperiments}
\begin{figure}
\begin{subfigure}[t]{0.3\textwidth}
\centering
\fbox{\includegraphics[width =.9\textwidth ]{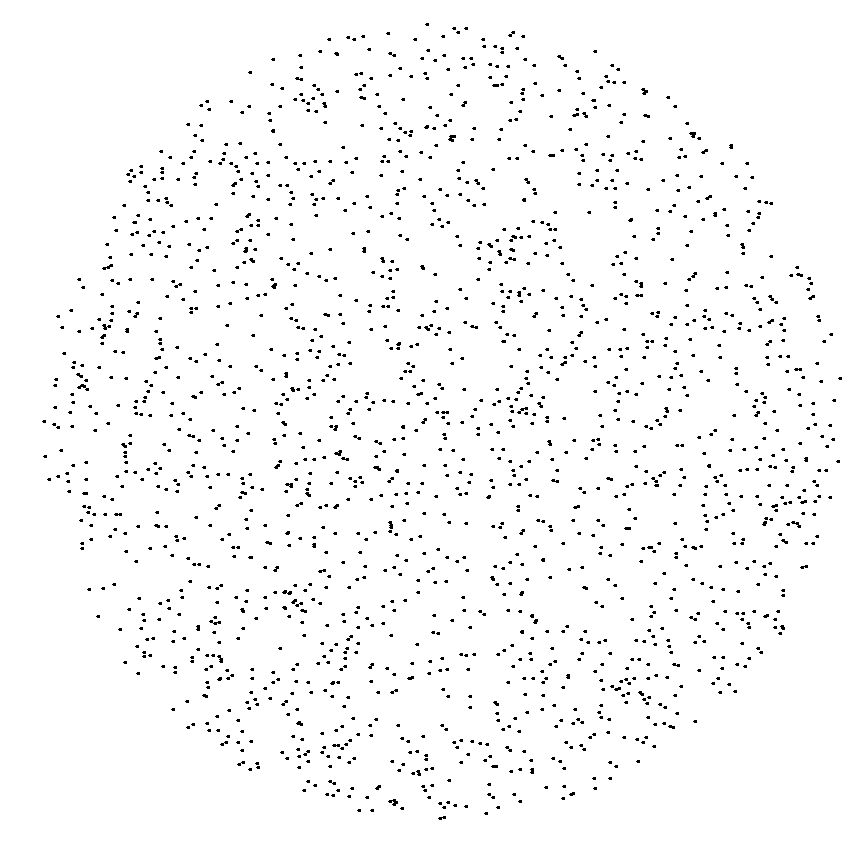}}
\caption{A Type 1 sample}
\end{subfigure}\hfill
\begin{subfigure}[t]{0.3\textwidth}
\centering
\fbox{\includegraphics[width =.9\textwidth ]{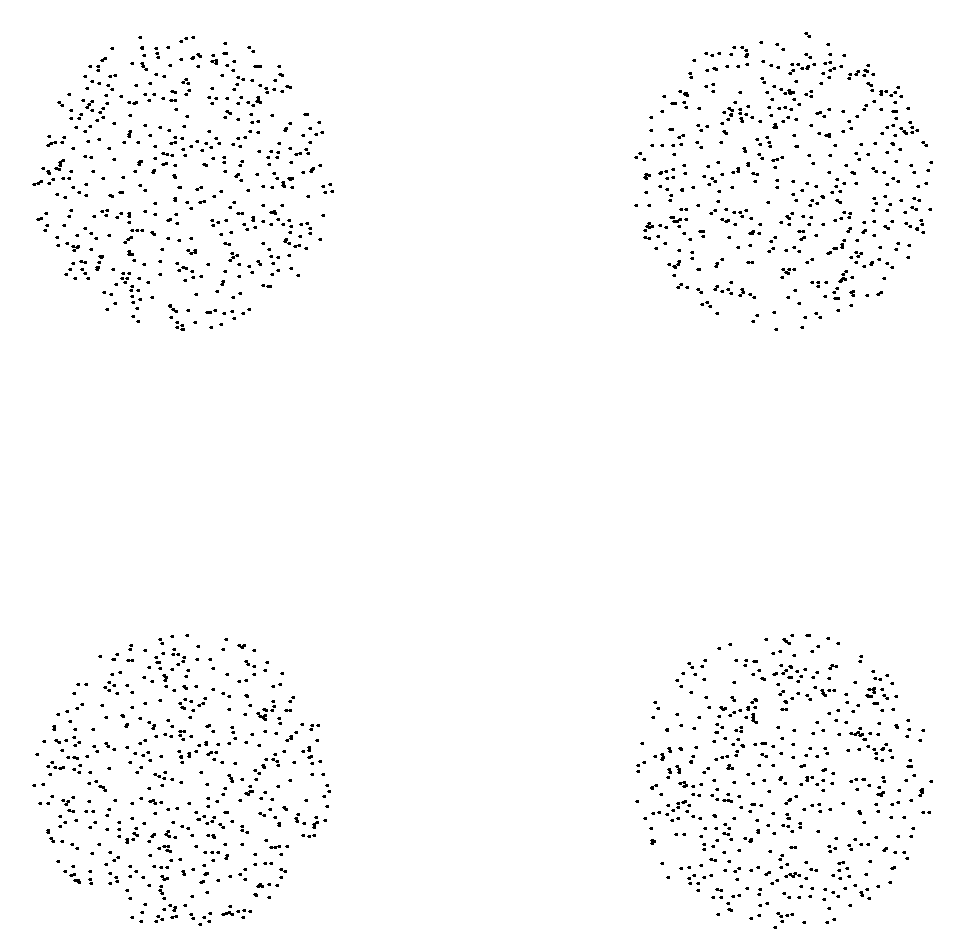}}
\caption{A Type 2 sample}
\end{subfigure}\hfill
\begin{subfigure}[t]{0.3\textwidth}
\centering
\fbox{\includegraphics[width =.878\textwidth ]{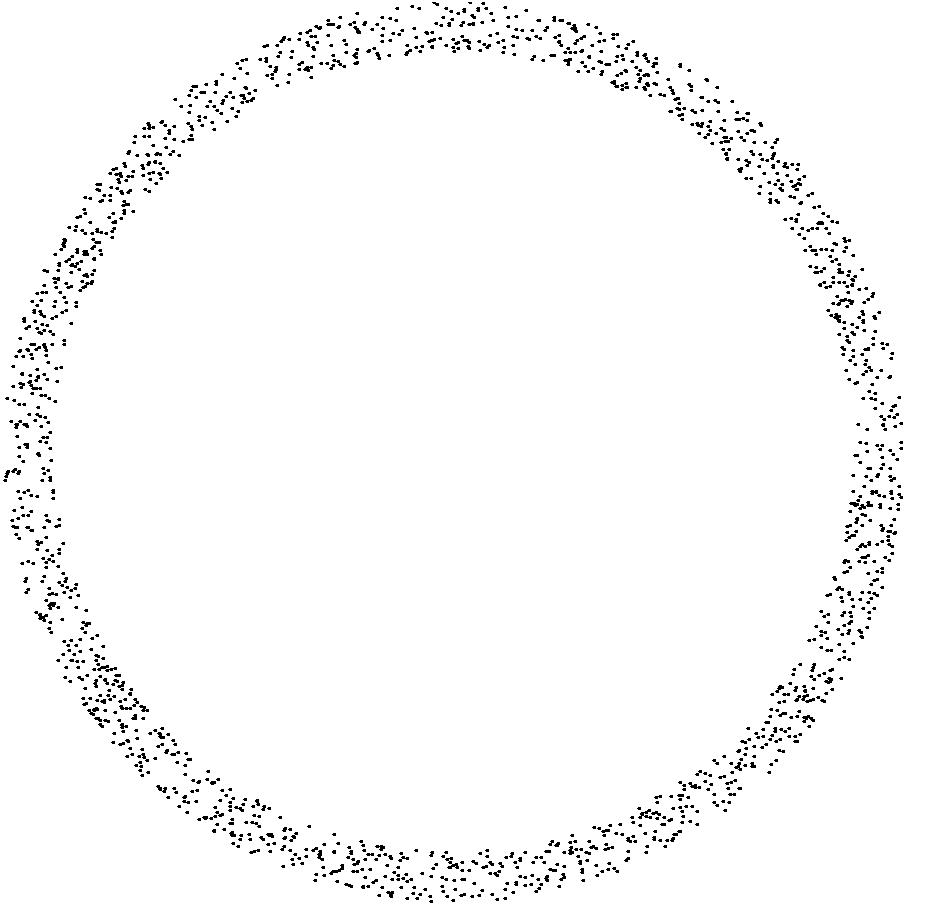}}
\caption{A Type 3 sample}
\end{subfigure}
\caption{The different sample types used in the experiments}\label{fig:4randomSets}
\end{figure}
In this section we perform some experiments with respect to the number of connected components contained inside the sub-polygons of the fixed-edge graphs. Note that our results for the uniform case match the experiments performed by Silveira et al. in \cite{SilveiraK2009}.

We try to investigate if the structure of the data has a big impact on the number $h$ of connected components, that are contained inside the sub-polygons  given by the fixed-edge graph for different Delaunay orders. For this task we look at three different types of randomly generated datasets:

\begin{description}
	\item[Type 1:] A Type 1 random dataset is given by points that are uniformly distributed in a circle of a given radius (Fig. \ref{fig:4randomSets} (a)).
	\item[Type 2:] A Type 2 random dataset is given by four Type 1 datasets such that the centers of the different circles are positioned on a square with side-length $2d$, where $d$ is the diameter of the circles (Fig. \ref{fig:4randomSets} (b)).
	\item[Type 3:] A Type 3 random dataset is given by a circular band, i.e., the points are uniformly distributed inside a circle with radius $r_1$, but a point is only accepted if its distance to the center of the circle is larger then $r_2$ (Fig. \ref{fig:4randomSets} (c)).
\end{description}

\noindent Type 1 random datasets have already been discussed in \cite{SilveiraK2009} and they are mainly used for comparison. We use Type 2 datasets to investigate, if clusters of points, that have a large distance between each other, form independent connected components and we use Type 3 datasets to investigate, if data that traces a geometric structure, e.g., a polygon or in our case a circle behaves differently with respect to the number of connected components. The tide gauge data can be seen as a combination of Type 2 and Type 3, since the tide gauge stations trace the coastlines and the different coastlines are far apart. Our random datasets are of course an extreme simplification.

In Tables \ref{tab:4r1}-\ref{tab:4r3} the results of our experiments are summarized. We calculated the average number of connected components inside sub-polygons and $c_\text{max}$ for a fixed sample. The depicted values are the averaged values over 200 samples for different numbers of data points $n$ and different Delaunay orders $k$.

For $k\leq 4$ it is unlikely that the fixed-edge graph of a dataset has a sub-polygon that contains a connected component.  Note that for $k=4$ none of the generated fixed-edge graphs contained a polygon with more than two connected components. Hence, our algorithm should be efficient for $k\leq 4$. Even for $k\leq 6$ the average maximum number of connected components is still small and the overall maximum number of connected components (not the average) was $10$, which is also promising with regards to the runtime. For $k=7$ the average maximal number of connected components is between $5$ and $7$. This suggests that we still have a lot of random instances we can solve optimally with reasonable runtime, but the worst fixed-edge graph for $k=7$ in our experiments had more then $20$ connected components inside a single sub-polygon. Hence, there can also be instances we are not able to solve efficiently for $k=7$. For $k\geq 8$ even the average maximal number of connected components is already bigger than 15. Thus, we cannot expect to triangulate pointsets optimally in reasonable time for $k\geq 8$.

We can compare the tables of the different Types. For $k\leq 6$ the values are all quite similar. This is reasonable, since our datasets are all locally uniformly distributed and Silveira et al.\ already mentioned in  \cite{SilveiraK2009} that the polygons for $k\leq 6$ most of the time only cover a small local area.

Some of the datasets of Type 2 samples still have polygons that are concentrated in individual disks for $k\geq 8$. This may be the reason for the slightly lower $c_\text{max}$ for Type 2 than Type 1 data for $k\geq 8$.

For Type 3 datasets even for  $k\geq 8$ the inner circle is an individual polygon most of the time. Thus, the maximal polygon must be inside the circular band and cannot cover the complete dataset. This may explain the significant lower $c_\text{max}$ for $k\geq 8$.\\

Overall no matter the generating process the $c_\text{max}$ values are promising for $k\leq 7$. For $k\geq 8$ the $c_\text{max}$ values are to large to be of practical use for all types of samples.
\begin{center}
\begin{table}[!ht]
\centering
\caption{Average/maximum number of connected components inside a sub-polygon for Type 1 random data averaged over 200 samples}\label{tab:4r1}
\begin{tabular}{ | c || c | c | c | c |}
\hline 
k   & $n=500$ & $n=1000$ & $n=1500$ & $n=20000$\\
\hline
3& 0.00	/	0.01	&	0.00	/	0.01	&	0.00	/	0.02	&	0.00	/	0.04	\\
4&0.00	/	0.33	&	0.00	/	0.46	&	0.00	/	0.64	&	0.00	/	0.68	\\
5&0.01	/	0.95	&	0.00	/	1.15	&	0.00	/	1.24	&	0.00	/	1.37	\\
6&0.05	/	1.94	&	0.04	/	2.29	&	0.04	/	2.39	&	0.04	/	2.48	\\
7&0.19	/	4.58	&	0.17	/	5.16	&	0.16	/	5.87	&	0.16	/	6.36	\\
8&0.58	/	13.67	&	0.54	/	17.72	&	0.51	/	20.88	&	0.51	/	23.94	\\
9&1.42	/	34.43	&	1.41	/	60.40	&	1.45	/	83.87	&	1.46	/	110.94	\\
10&2.78	/	59.87	&	3.09	/	116.78	&	3.12	/	167.14	&	3.32	/	227.93	\\
\hline
\end{tabular}
\end{table}\end{center}
\begin{center}
\begin{table}[!ht]
\centering
\caption{Average/maximum number of connected components inside a sub-polygon for Type 2 random data averaged over 200 samples}\label{tab:4r2}
\begin{tabular}{ | c || c | c | c | c |}
\hline 
k   & $n=500$ & $n=1000$ & $n=1500$ & $n=20000$\\
\hline
3&0.00	/	0.02	&	0.00	/	0.02	&	0.00	/	0.01	&	0.00	/	0.05	\\
4&0.00	/	0.50	&	0.00	/	0.64	&	0.00	/	0.68	&	0.00	/	0.78	\\
5&0.01	/	1.20	&	0.01	/	1.24	&	0.01	/	1.33	&	0.01	/	1.42	\\
6&0.08	/	2.33	&	0.06	/	2.58	&	0.05	/	2.68	&	0.05	/	2.96	\\
7&0.30	/	5.98	&	0.23	/	6.09	&	0.21	/	6.74	&	0.20	/	7.02	\\
8&0.85	/	13.94	&	0.71	/	16.95	&	0.67	/	20.49	&	0.63	/	23.20	\\
9&2.06	/	29.21	&	1.85	/	45.63	&	1.80	/	61.56	&	1.74	/	74.62	\\
10&4.03	/	54.13	&	3.93	/	92.88	&	3.87	/	126.68	&	3.97	/	171.79	\\
\hline
\end{tabular}

\end{table}\end{center}
\begin{center}\begin{table}[!ht]
\centering
\caption{Average/maximum number of connected components inside a sub-polygon for Type 3 random data averaged over 200 samples}\label{tab:4r3}
\begin{tabular}{ | c || c | c | c | c |}
\hline 
k   & $n=500$ & $n=1000$ & $n=1500$ & $n=20000$\\
\hline
3&0.00	/	0.04	&	0.00	/	0.04	&	0.00	/	0.04	&	0.00	/	0.07	\\
4&0.00	/	0.75	&	0.00	/	0.68	&	0.00	/	0.81	&	0.00	/	0.86	\\
5&0.03	/	1.86	&	0.01	/	1.50	&	0.01	/	1.55	&	0.01	/	1.59	\\
6&0.13	/	4.19	&	0.07	/	3.06	&	0.06	/	3.17	&	0.05	/	3.14	\\
7&0.35	/	10.29	&	0.21	/	7.13	&	0.18	/	7.19	&	0.18	/	7.40	\\
8&0.72	/	23.51	&	0.48	/	16.36	&	0.45	/	15.95	&	0.44	/	17.50	\\
9&1.19	/	50.55	&	0.89	/	34.57	&	0.86	/	37.10	&	0.89	/	42.17	\\
10&1.71	/	87.76	&	1.38	/	74.23	&	1.43	/	81.65	&	1.52	/	97.90	\\

\hline
\end{tabular}
\end{table}\end{center}
\newpage
\section{Missing proofs of Section \ref{sec:HODOptimization}}\label{sec:missingproofs}
In this section we provide the proof of Theorem \ref{theorem:fixedEdges}. Let $\mathcal{S}$ be a set of points in general position, i.e., no four points lie on a circle. We denote the circumcircle of three points $u,v,w\in \mathcal{S}$ by $C(u,v,w)$. In this section we say that two edges intersect only if they properly intersect, i.e, the intersection point is not an endpoint of an edge.

\medskip

Let $\overline{uv}$ be an edge. If not stated otherwise, we assume that the edge $\overline{uv}$ is oriented from $u$ to $v$, i.e., it corresponds to the oriented edge $\overrightarrow{uv}$.
We can find a point $s_l\in \mathcal{S}$ that is left of $\overline{uv}$ such that the circle $C(u,v,s_l)$ does not contain any other point from $\mathcal{S}$ that is left of $\overline{uv}$. We call $s_l$ the \emph{left defining point}, the circle $C(u,v,s_l)$ the \emph{left defining circle} and the empty triangle $T_{uvs_l}$ the \emph{left defining triangle} of $\overline{uv}$. In the same way we can find the \emph{right defining point} $s_r$ and the \emph{right defining circle} as well as the  \emph{right defining triangle}.

For a circle given by an edge $\overline{uv}$ and an additional point $x$ we define $\ca^{uv}_x$ to be $C(u,v,x)\cap H$, where $H$ is the half-plane defined by $\overline{uv}$ that does not contain $x$. Thus, the region $\ca^{uv}_x$ is the part of the circle $C(u,v,x)$ that is opposite of  $x$ with respect to $\overline{uv}$. See Figure~\ref{fig:usefulEdge} for an illustration of the defining circles and $\ca^{uv}_{s_l}$.

Note that for the left(right) defining circle all of their contained points are in $\ca^{uv}_{s_l}$($\ca^{uv}_{s_r}$). Additionally, $\ca^{uv}_{s_l}$($\ca^{uv}_{s_r}$) contains all of  $\ca^{uv}_{x}$ for every point $x\in\mathcal{S}$ that is left(right) of $\overline{uv}$.

\begin{figure}[!tb]
\begin{minipage}[t]{.37\textwidth}
\centering
\includegraphics[width=.75\textwidth]{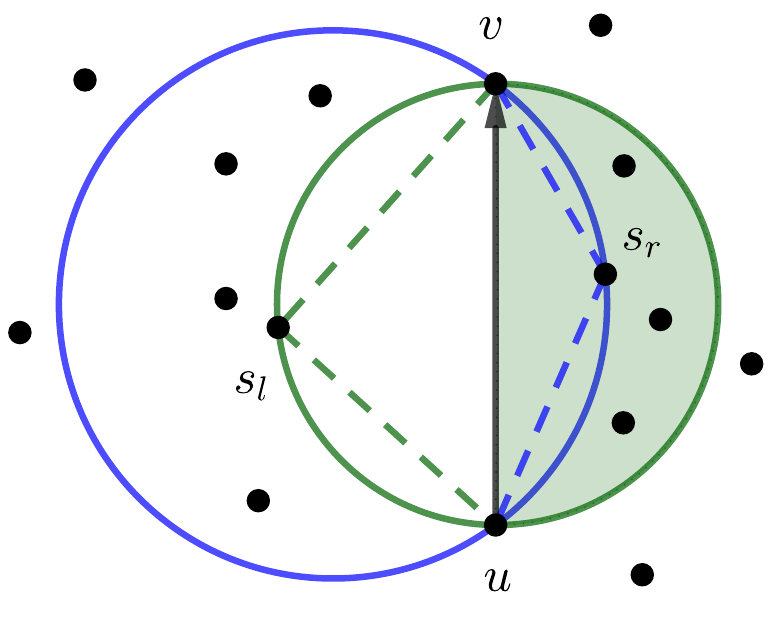}
\caption{The left defining circle of $\overline{uv}$ is given in green and the right defining circle is given in blue; the region  $\ca^{uv}_{s_l}$ is shaded in green}\label{fig:usefulEdge}
\end{minipage}
\hfill
\begin{minipage}[t]{.3\textwidth}
\centering
\includegraphics[width=0.75\textwidth]{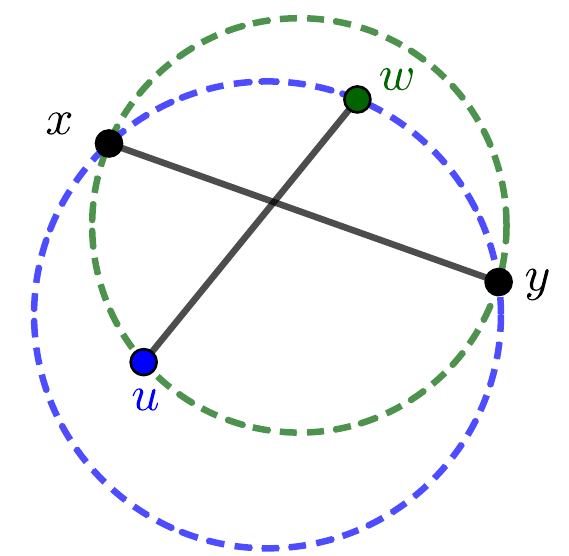}
\captionsetup{textformat=simple}
\caption{The circle $C(x,y,u)$ contains $w$ and $C(x,y,w)$ contains $u$; $(\overline{xy},\overline{uw})$ is a {Type-1} pair and $(\overline{uw},\overline{xy})$ is a Type-2 pair}\label{fig:intersectingEdges}
\end{minipage}
\hfill
\begin{minipage}[t]{.22\textwidth}
\centering
\includegraphics[width=0.9\textwidth]{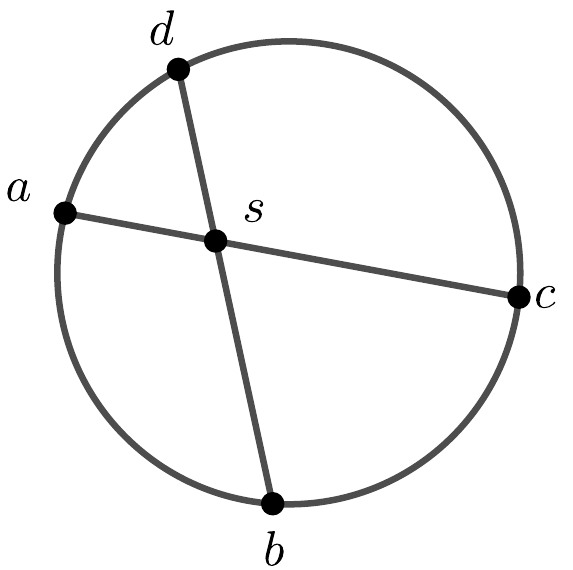}
\captionsetup{textformat=simple}
\caption{Two intersecting chords of a circle}\label{fig:chordtheorem}
\end{minipage}
\end{figure}
\begin{lemma}[from \cite{GudmundssonHK2002}]\label{lemma:useful}
Let $\overline{uv}$ be an edge. The edge $\overline{uv}$ is useful with order $k$ if and only if the  left and right defining triangles are $k$-OD triangles, i.e., the left and right defining circles each contain at most $k$ points.
\end{lemma}
We start by presenting the main properties of defining circles of intersecting edges that are needed for the proof. The following observation relates the defining circles of an edge $\overline{uv}$ to the endpoints of Delaunay edges that intersect $\overline{uv}$.
\begin{observation}[from \cite{GudmundssonHK2002}]\label{obs:delaunayEdge}
Let $\overline{uv}$ be an edge. The left(right) defining circle of $\overline{uv}$ contains all points that are right(left) of $\overline{uv}$ which are incident to Delaunay edges that intersect $\overline{uv}$.
\end{observation}
The following corollary summarizes the implications of the observation with respect to the number of points inside defining circles that are relevant for the proof.
\begin{corollary}\label{cor:delaunayEdge}
Let $\overline{uv}$ be an edge.
\begin{enumerate}
\item If $\overline{uv}$ intersects one Delaunay edge, both defining circles contain at least one point, i.e., each circle contains one of the endpoints of the Delaunay edge.
\item If $\overline{uv}$ intersects two Delaunay edges, at least one defining circle contains at least two points.
\item If $\overline{uv}$ intersects three Delaunay edges that all share one endpoint  left(right) of $\overline{uv}$, then the left(right) defining circle contains three points.
\end{enumerate}
\end{corollary}

Next, we can generalize the observation to arbitrary intersecting edges.
\begin{observation}\label{obs:intersectingEdges}
Let $\overline{uw}$ and $\overline{xy}$ be two edges that intersect. Then either $y\in C(u,w,x)$ and $x\in C(u,w,y)$ or $u\in C(y,x,w)$ and $w\in C(y,x,u)$; see Figure~\ref{fig:intersectingEdges}.
\end{observation}
This can again be reformulated in the context of defining circles.
\begin{corollary}\label{cor:intersectingEdges}
Let $\overline{uw}$ and $\overline{xy}$ be two edges that intersect. Either $y$ is in the left(right) defining circle of $\overline{uw}$  and $x$ is in the right(left) defining circle of $\overline{uw}$ or $u$ is in the left(right) defining circle of $\overline{xy}$ and $w$ is in the right(left) defining circle of $\overline{xy}$.
\end{corollary}
\begin{proof}
Without loss of generality we assume that $x$ is right of $\overline{uw}$ and $x\in C(u,w,y)$ and $y$ is left of $\overline{uw}$ and $y\in C(u,w,x)$. Let $s_l$ be the left defining point. We know that $\ca^{uw}_{s_l}$ contains $\ca^{uw}_{y}$ which implies that $x$ is in  $C(u,w,s_l)$. The same argument also holds for the right defining point $s_r$ and $y$.
\end{proof}
The last lemma we need is a statement that relates the lengths of intersecting chords of a circle.
\begin{lemma}[Intersecting Chord Theorem \cite{Glaister2007}]\label{lemma:chord}
Let $\overline{ac}$ and $\overline{bd}$ be two chords of a circle and let $s$ be the intersection point of the chords; see Figure \ref{fig:chordtheorem}. Then the following equality for the lengths of the chord segments holds:
\begin{align*}
|\overline{as}|\cdot |\overline{sc}|=|\overline{bs}|\cdot |\overline{sd}|.
\end{align*}
\end{lemma}

\noindent We now have established the main tools for the proof.

\fixedTheorem*

Before we start the proof, we state what being isolated for a vertex $v$ means in the context of fixed edge graphs. For a vertex $v$ we define its \emph{Delaunay neighbourhood} $N$ to be the union of all Delaunay triangles that have $v$ as one of their vertices. The Delaunay edges in $N$ that have $v$ as an endpoint are called \emph{connecting edges} and the Delaunay edges in $N$ that do not have $v$ as an endpoint are called \emph{boundary edges}. 
\begin{figure}[!tb]
\centering
\includegraphics[width=.33\textwidth]{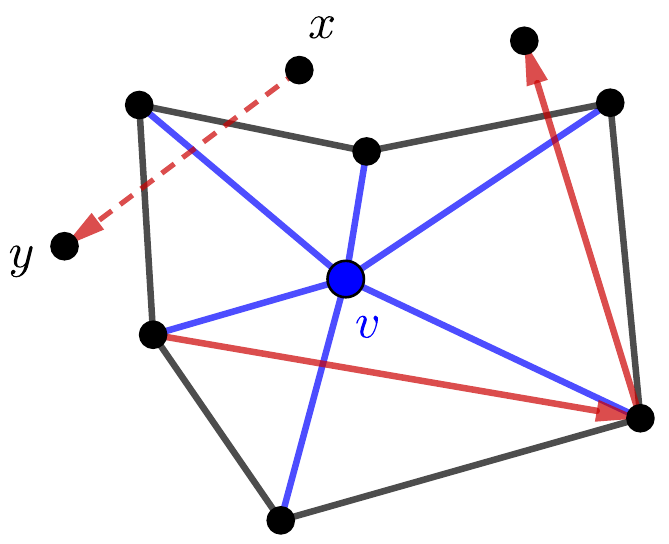}
\caption{The Delaunay neighbourhood of $v$ with connecting edges in blue and boundary edges in black; in red some (oriented) separation edges. Note that $\overline{xy}$ intersects a connecting edge, but is not a (useful) separation edge.}\label{fig:delNeighbourhood}
\end{figure}
A useful $2$-OD edge $s$ that intersects a connecting edge $e$ is called a \emph{separation edge}. In Figure~\ref{fig:delNeighbourhood} the Delaunay neighbourhood of a vertex $v$ and some separation edges are illustrated. If it is important which connecting edge is intersected, we call $s$ a separation edge \emph{for} the connecting edge $e$. Since a separation edge intersects a Delaunay edge that has $v$ as an endpoint, we know that one of its defining circles must contain $v$ by Observation \ref{obs:delaunayEdge}. We orient separation edges counter-clockwise with respect to $v$, i.e., the right defining circle is the circle that contains $v$.

Using these definitions a vertex $v$ is isolated if and only if all of its connecting edges are intersected by separation edges.

\begin{proof}
It is sufficient to show that for any fixed vertex $v$ there does not exist a set of separation edges such that every connecting edges of $v$ is intersected by an edge of the set. In particular there cannot be a \underline{minimal} set $E$ of separation edge such that every connecting edges of $v$ is intersected by an edge $e\in E$ (We say that a set $E$ is minimal, if there does not exist an edge $\hat{e}\in E$, such that every connecting edge is intersected by an edge $e\in E\setminus \{ \hat{e} \}$). 

We first investigate the position of the endpoints of separation edges with respect to the boundary edges of the Delaunay neighbourhood $N$.

\begin{restatable}{claim}{claimeins}
	\label{c1}
	Let $\overline{xy}$ be an edge that intersects a connecting edge. If neither $x$ nor $y$ is a boundary vertex of $N$, then $\overline{xy}$ cannot be a separation edge.
\end{restatable}\begin{restatable}{claim}{claimzwei}
	\label{c2}
	Let $\overline{uw}$ be a boundary edge. Let $\overline{xy}$ and $\overline{sr}$ be two edges that intersect $\overline{uw}$. If $\overline{xy}$ intersects the connecting edge $\overline{uv}$ and $\overline{sr}$ intersects the connecting edge $\overline{wv}$ (or vice versa), then one of the edges cannot be a separation edge. 
\end{restatable}
Since all edges in $E$ are separation edges, they cannot be positioned as discussed in Claim~\ref{c1} and Claim~\ref{c2}. The following claim shows that in this case $E$ must have a special structure. 
\begin{restatable}{claim}{claimdrei}
	\label{c3}
	The edges $e\in E$ form a cycle, i.e., there exists a cyclic ordering $(e_1,...,e_m)$ with $e_1=e_{m}$, such that for all $i$ the edge $e_i$ (properly) intersects $e_{i-1}$ and $e_{i+1}$. Such a cycle is illustrated in Figure~\ref{fig:cycleSeparation} (c).
\end{restatable}
Note that for every separation edge $e$ the Delaunay neighbourhood $N$ has at least two connecting edges that can not be intersected by $e$, i.e., the next connecting edge in counter-clockwise and clockwise order that has not been intersected by $e$ (they can not be identical, because of Claim \ref{c1}). Claim \ref{c1} and \ref{c2} also imply that both of these connecting edges must be intersected by different separation edges. Thus, $E$ contains at least three separation edges.

From now on we assume that $E$ is a cyclic set, i.e., we have an ordering such that $e_i$ intersects $e_{i-1}$ and $e_{i+1}$ for all $i$. We call a pair $(\overline{uw},\overline{u'w'})$ of separation edges a \emph{Type-1 pair}, if Observation \ref{obs:intersectingEdges} holds for $\overline{uw}$, i.e., $w'\in C(u,w,u')$ and $u'\in C(u,w,w')$. Otherwise we call it a \emph{Type-2 pair}; see Figure \ref{fig:intersectingEdges}. Note that, if $(\overline{uw},\overline{u'w'})$ is a Type-1 pair, then $(\overline{u'w'},\overline{uw})$ is a Type-2 pair.

It remains to show that a (minimal) cyclic set $E$ of separation edges cannot exist, i.e. at least one edge $e\in E$ cannot be useful. For this we investigate the usefulness of the edges in $E$, if we have specific types of edge pairs.

\begin{restatable}{claim}{claimvier}
	\label{c4}
	Let $e_i,e_{i+1},e_{i+2}\in E$ be consecutive edges, such that $(e_{i},e_{i+1})$ is a Type-2 pair and $(e_{i+1},e_{i+2})$ is a Type-1 pair. Then not all edges in $E$ can be useful.
\end{restatable}
\begin{restatable}{claim}{claimfunf}
	\label{c5}
	If the sequence $E$ only has successive edges $e_i,e_{i+1}$ that form Type-2 pairs $(e_i,e_{i+1})$, then one of the edges $e\in E$ cannot be useful. 
\end{restatable}

Claim \ref{c5} implies that, if the consecutive edges of $E$ only form Type-1 pairs (or symmetrically only Type-2 pairs), then not all edges in $E$ can be separation edges. It follows that we must have at least one Type-1 and one Type-2 pair, if all edges in $E$ are separation edges. Since the set is cyclic, we must have a Type-2 pair followed by a Type-1 pair, but Claim \ref{c4} shows that in this case not all edges in $E$ can be separation edges.

All in all, Claim \ref{c4} and Claim \ref{c5} imply that at least one of the edges in $E$ cannot be a separation edge. This is a contradiction to the definition of $E$.
\end{proof}

\paragraph*{Proofs of the claims}

\claimeins*
\begin{proof}
Let $\overline{xy}$ be an edge with both endpoints outside of $N$ that intersects a connecting edge. Then $\overline{xy}$ must intersect at least three Delaunay edges that are connected to three individual vertices on one side; see Figure~\ref{fig:delNeighbourhood}. Hence, $\overline{xy}$ cannot be useful by Corollary~\ref{cor:delaunayEdge} and Lemma~\ref{lemma:useful}. Consequently, it cannot be a separation edge.
\end{proof}

\claimzwei*
\begin{proof}
We first assume that the edges do not share an endpoint. Let $\overline{xy}$ and $\overline{sr}$ be edges as described in the claim. Since both edges start outside of $N$, they both intersect at least two Delaunay edges. Thus,  both of them have at least order two by Corollary \ref{cor:delaunayEdge}. This implies for both edges that at least one of the defining circles contains two points. Hence, it is sufficient to show that one of these circles must contain an additional point.

For the moment we assume that one of the points $x$ or $s$ is a vertex of a Delaunay triangle that has $\overline{uw}$ as an edge. Without loss of generality let $x$ be the vertex of the Delaunay triangle. This implies that $s$ must be outside of the triangle and hence, intersect one of the additional Delaunay edges $\overline{xu}$ or $\overline{xw}$. Now we have two cases: Either $\overline{sr}$ intersects the Delaunay edge $\overline{xw}$ (Fig.~\ref{fig:claim1a} (a)) or the Delaunay edge $\overline{xu}$ (Fig.~\ref{fig:claim1a} (b)):
\begin{figure}
\begin{subfigure}[t]{0.33\textwidth}
\centering
\includegraphics[width =1\textwidth ]{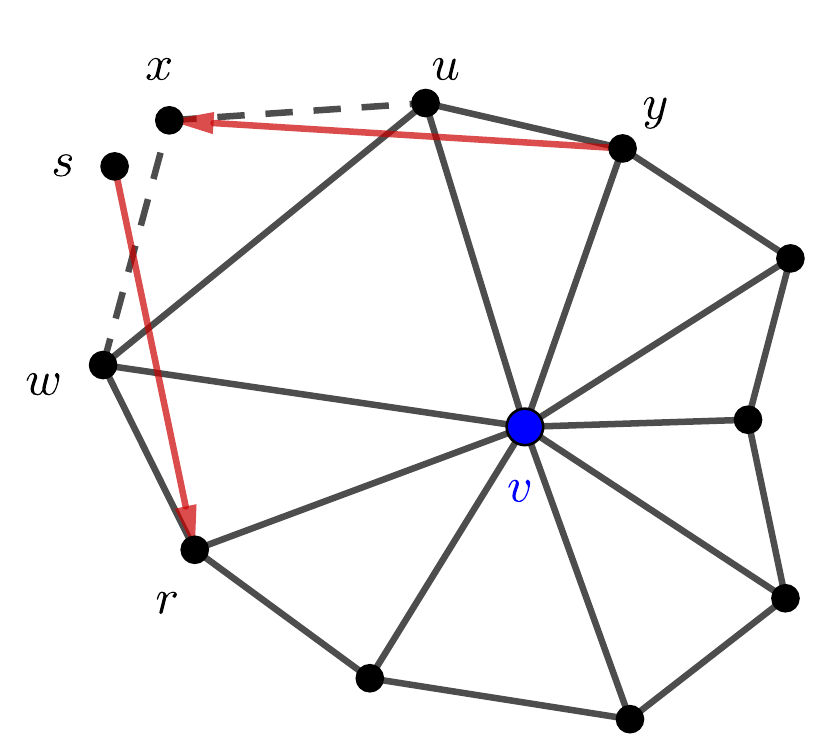}
\subcaption{$\overline{sr}$ and $\overline{xy}$ do not intersect}
\end{subfigure}\hfill
\begin{subfigure}[t]{0.33\textwidth}
\centering
\includegraphics[width =1\textwidth ]{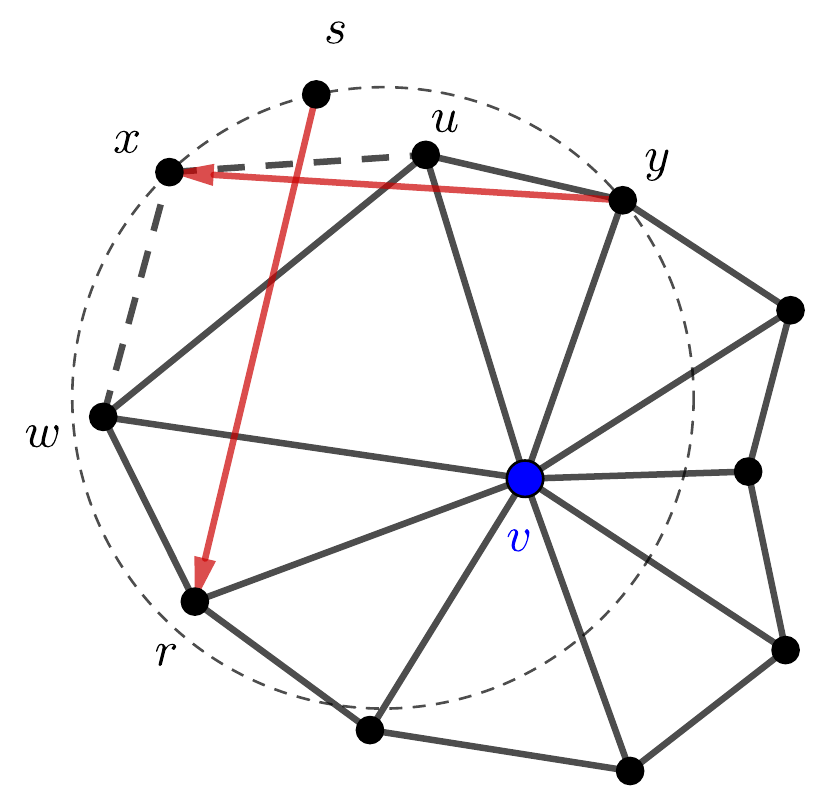}
\subcaption{$\overline{sr}$ and $\overline{xy}$ do intersect}
\end{subfigure}
\caption{The possible positions of the endpoints outside of the Delaunay neighbourhood}\label{fig:claim1a}
\end{figure}
\begin{description}
	\item[Case 1:] If $\overline{sr}$ intersects $\overline{xw}$, then three Delaunay edges are intersected such that they have three different endpoints on one side of  $ \overline{sr}$; the three points are $x,u$ and $v$, since they all are connected to $w$. Thus, the defining circle on the opposite side must contain these three points by Corollary \ref{obs:delaunayEdge} and therefore, the edge $\overline{sr}$ cannot be useful; see Figure~\ref{fig:claim1a}~(a).
	\item[Case 2:] If $\overline{sr}$ intersects $\overline{xu}$, then the edges $\overline{sr}$ and $\overline{xy}$ must intersect; see Figure \ref{fig:claim1a} (b). Without loss of generality Corollary \ref{cor:intersectingEdges} implies that the defining circles of $\overline{xy}$ each contain one of the points $s$, $r$. Thus, the right defining circle of $\overline{xy}$ must contain $v$, either $u$ or $w$ and either $s$ or $r$. None of these three points can be identical, which implies that $\overline{xy}$ cannot be useful by Lemma \ref{lemma:useful}.
\end{description}
Both cases imply that one of the  edges cannot be a separation edge. Consequently, the claim is true, if we assume that the edges do not share an endpoint.

Note that we assumed that either $x$ or $s$ is the endpoint of a Delaunay triangle that uses $\overline{uw}$. This may not be the case, i.e., there may be another vertex $z$ that is the endpoint of the triangle, but similar arguments as before can be used with respect to $T_{uwz}$ to show that in this case one separation edge cannot be useful.

\medskip

It remains to show that the claim also holds, if the edges share an endpoint outside of the Delaunay neighbourhood $N$. In Figure~\ref{fig:claim2b} (a) this situation is illustrated. Let the endpoints $x$ and $s$ be identical and outside of $N$. We assume that $\overline{xy}$ only intersects $\overline{uv}$ and $\overline{xr}$ only intersects $\overline{wv}$. It may happen that additional connecting edges, e.g., $\overline{yv}$ or $\overline{rv}$, may be intersected, but the following proof for one intersection on each side can easily be adapted for this case.

In Figure~\ref{fig:claim2b} (b) only the relevant vertices and edges are depicted. Note that all of the black edges are Delaunay edges and the black circle is the circle $C(u,w,v)$ of the Delaunay triangle $T_{uwv}$. Hence, the circle should be empty. We now argue that the circle $C(x,y,r)$ which is given in red must contain $u$ or $w$.
\begin{figure}
\begin{subfigure}[t]{0.3\textwidth}
\centering
\includegraphics[width =1\textwidth ]{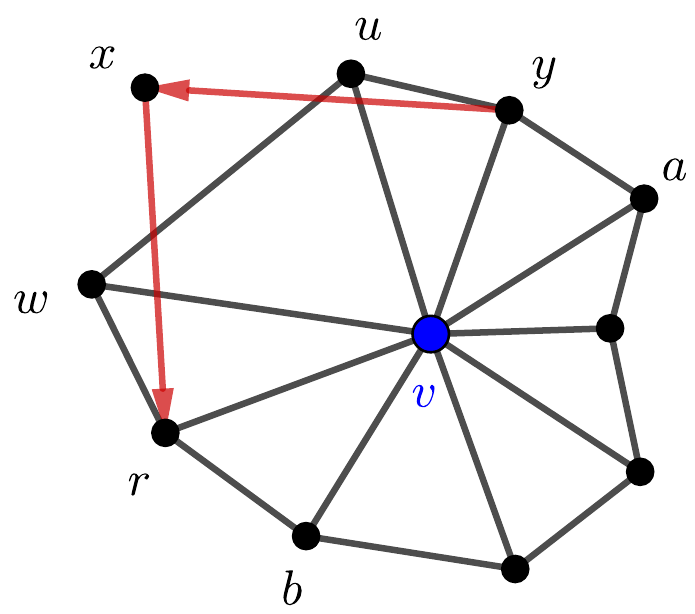}
\subcaption{the endpoint outside of $N$ is identical}
\end{subfigure}\hfill
\begin{subfigure}[t]{0.3\textwidth}
\centering
\includegraphics[width =1\textwidth ]{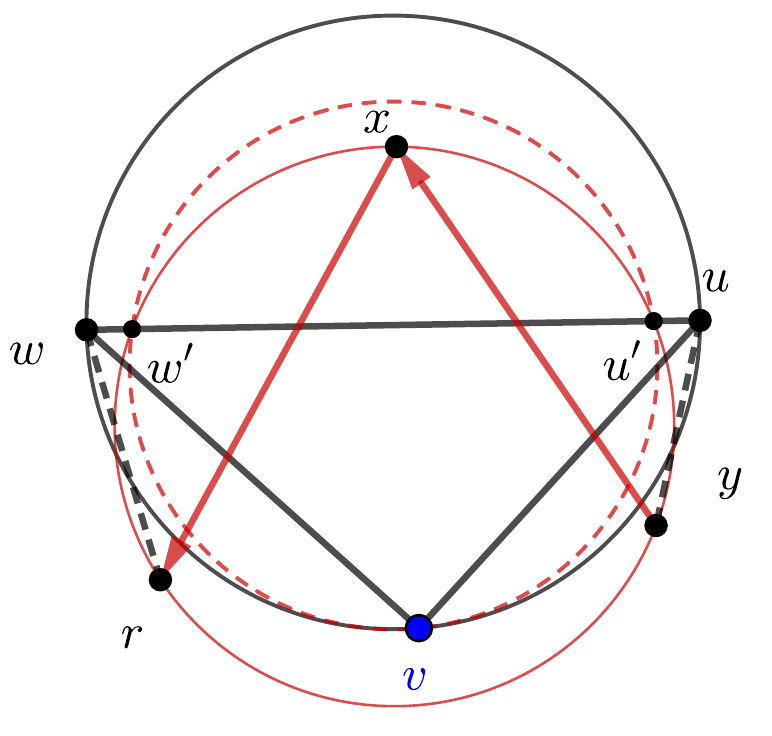}
\subcaption{$u$ and $w$ are outside of $C(x,y,r)$}
\end{subfigure}\hfill
\begin{subfigure}[t]{0.3\textwidth}
\centering
\includegraphics[width =1\textwidth ]{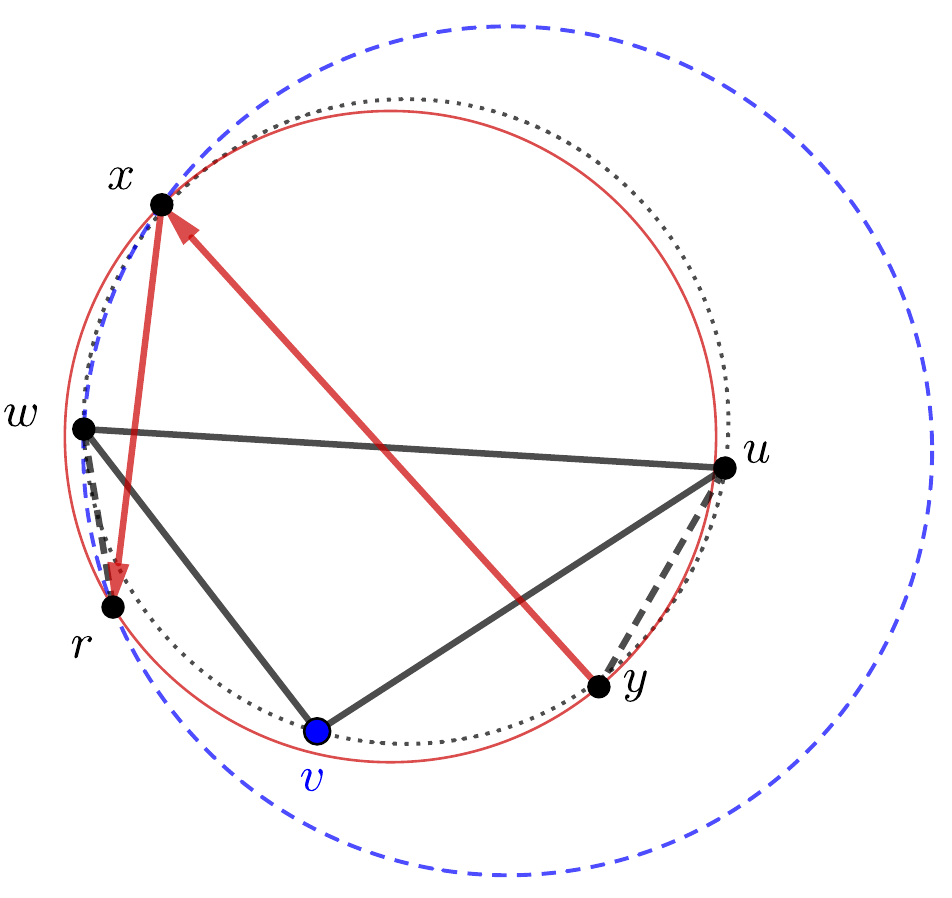}
\subcaption{$u$ is inside of $C(x,y,r)$}
\end{subfigure}
\caption{The possible position of the endpoint outside of the Delaunay neighbourhood}\label{fig:claim2b}
\end{figure}

We assume that neither $u$ nor $w$ is inside of $C(x,y,r)$. This implies that $\overline{uw}$ must intersect $C(x,y,r)$ twice in the points $u'$ and $w'$, since $\overline{xy}$ and $\overline{xr}$ intersect $\overline{uw}$; see Figure~\ref{fig:claim2b}(b). The vertices $y$ and $r$ must be outside of $C(u,w,v)$, since it is a circle of a Delaunay triangle. If we now move the red circle $C(x,y,r)$ until it touches $v$ while anchoring the circle on the points $u'$ and $w'$, we get the red dotted circle $C(v,u',w')$. Note that this circle must contain $x$, since it must contain all of $C(x,y,r)$ right of $\overline{uw}$. Since $C(u,w,v)$ and $C(u',w',v)$ share one defining point and $u'$ as well as $w'$ are inside of $C(u,w,v)$, we know that $C(u,w,v)$ must contain all of $C(v,u',w')$. In particular it must contain the point $x$, but we assumed that $C(u,w,v)$ is a Delaunay circle. This is a contradiction.

We now know that at least one of the points $u$ and $w$ must be inside the circle $C(x,y,r)$. Without loss of generality we can assume that $w$ is inside $C(x,y,r)$; see Figure~\ref{fig:claim2b} (c). We can now look at the defining circle of $\overline{xr}$ given by $C(w,x,r)$ depicted in blue. We know that $C(w,x,r)$ must contain $v$ and $u$ because of the Delaunay edges that are intersected by $\overline{xy}$ and Observation \ref{obs:delaunayEdge}. Moreover, we know that $\ca^{xr}_w$ contains all of $C(x,y,r)$ that is left of $\overrightarrow{xr}$, since $w$ is right of $\overrightarrow{xr}$ and in the interior of $C(x,y,r)$. In particular $\ca^{xr}_w$ contains $y$. Thus, $C(w,x,r)$ contains at least $v,u$ and $y$. Lemma \ref{lemma:useful} implies that $\overline{xy}$ cannot be useful. 

This proves that at least one of the edges cannot be a separation edge which completes the proof of the claim.
\end{proof}

\claimdrei*

\begin{proof}
We pick edges from $E$ in an iterative way to show that we get a cycle. We start with an arbitrary vertex $u_1$ on the boundary of the Delaunay neighbourhood $N$ of $v$. Then, we pick a separation edge $e_1\in E$ that intersects $\overline{u_1 v}$. Next, we consider the vertex $u_2$, which is the next vertex after $u_1$ on the boundary of $N$ in counter-clockwise order, such that the connecting edge $\overline{u_2v}$ is not already intersected. We have two possible cases how $e_1$ is positioned with respect to $u_2$. Either $e_1$ has $u_2$ as an endpoint or $e_1$ intersects the edge $\overline{u_2u_{pre}}$ where $u_{pre}$ is the predecessor of $u_2$ in the counter-clockwise order. Note that the vertex $u_1$ can also be $u_{pre}$.

Next, we pick a separation edge $e_2\in E$ that intersects the connecting edge $\overline{u_2v}$ and need to argue that in both cases $e_1$ is also intersected by $e_2$.
\begin{description}
	\item[Case 1:] Let $u_2$ be an endpoint of $e_1$. Every edge $e_2$ that intersects $\overline{u_2v}$ and has endpoints on the boundary of $N$ or outside of $N$ must also intersect the edge $e_1$; see Figure~\ref{fig:cycleSeparation}~(a). Otherwise $e_1$ would be redundant, since $e_2$ would intersect all connecting edges that $e_1$ intersects and additionally $\overline{vu_2}$.
\begin{figure}
\begin{subfigure}[t]{0.27\textwidth}
\centering
\includegraphics[width =.9\textwidth ]{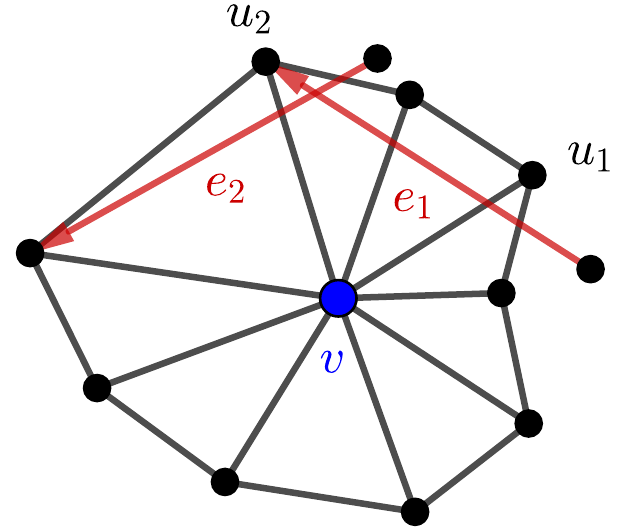}
\subcaption{$u_2$ is an endpoint of $e_1$}
\end{subfigure}\hfill
\begin{subfigure}[t]{0.27\textwidth}
\centering
\includegraphics[width =.9\textwidth ]{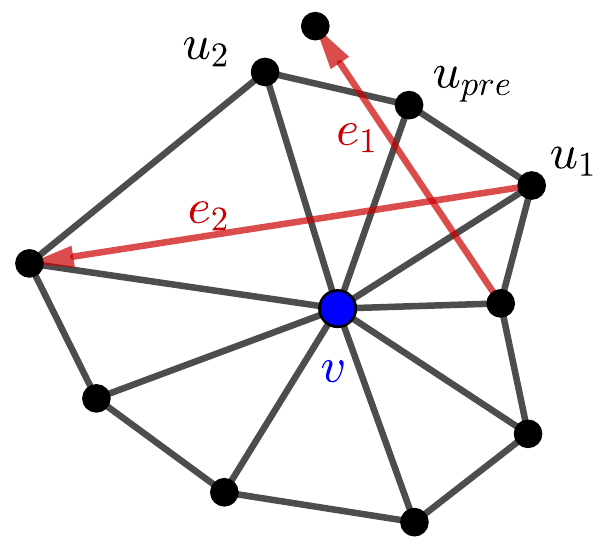}
\subcaption{$e_1$ intersects $\overline{u_2u_{pre}}$}
\end{subfigure}\hfill
\begin{subfigure}[t]{0.27\textwidth}
\centering
\includegraphics[width =.8\textwidth ]{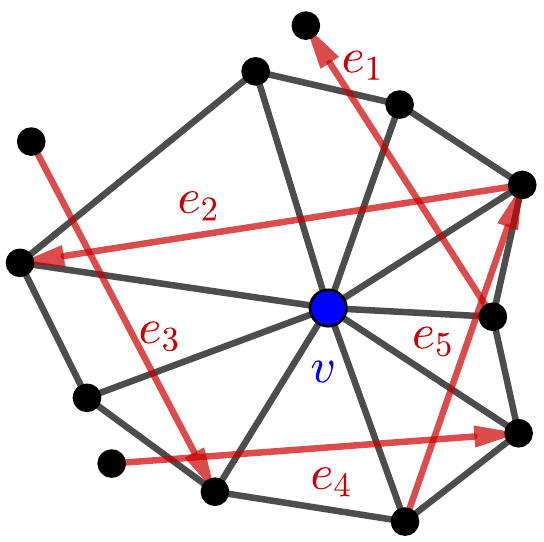}
\subcaption{a (minimal) cyclic set $E$ of separation edges}
\end{subfigure}
\caption{The separation edge cycle construction}\label{fig:cycleSeparation}
\end{figure}	
	\item[Case 2:] Let $e_1$ intersect $\overline{u_2u_{pre}}$. Claim \ref{c1} and Claim \ref{c2} imply that $e_2$ cannot intersect $\overline{u_2u_{pre}}$, too. Thus one endpoint must be at $u_{pre}$ or even earlier in the counter clockwise order (possibly it may also be outside of $N$, but then it must intersect a boundary edge that comes before $u_{pre}$ in the counter-clockwise order). Hence, $e_2$ must intersect $e_1$; see Figure~\ref{fig:cycleSeparation}~(b).	
\end{description}
In this way we can add all of the edges iteratively which always results in one of the two cases and the last edge must intersect the second to last separation edge, but also the first separation edge we picked. Consequently, we get a cycle.
\end{proof}

\claimvier*
\begin{proof}
Without loss of generality $i=1$, i.e., the edge pair $(e_1,e_2)$ is a Type-2 pair and $(e_2,e_3)$ is a Type-1 pair. We have three possible cases with respect to the endpoints of edges:
(1) all edges have distinct endpoints, (2) $e_1$ and $e_3$ share an endpoint that is right of $e_2$ and (3) $e_1$ and $e_3$ share an endpoint left of $e_2$. The cases are depicted in Figure~\ref{fig:type21}; Case 1 corresponds to (a), Case 2 corresponds to (b) and Case 3 corresponds to (c) and (d). Note that $e_2$ cannot share an endpoint with either $e_1$ or $e_3$. It remains to show that in all of the cases there exists a separation edge in $E$ that is not useful.
\begin{figure}[!tb]
\begin{subfigure}[t]{0.20\textwidth}
\centering
\includegraphics[width =1\textwidth ]{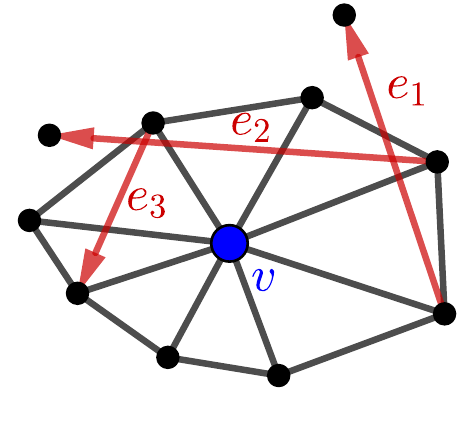}
\subcaption{$e_1$ and $e_3$ do not share an endpoint}
\end{subfigure}\hfill
\begin{subfigure}[t]{0.20\textwidth}
\centering
\includegraphics[width =1\textwidth ]{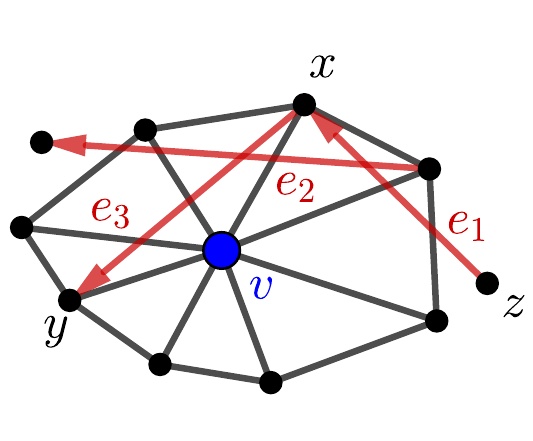}
\captionsetup{textformat=simple}
\subcaption{$e_1$ and $e_3$ share an endpoint right of $e_2$}
\end{subfigure}\hfill
\begin{subfigure}[t]{0.20\textwidth}
\centering
\includegraphics[width =1\textwidth ]{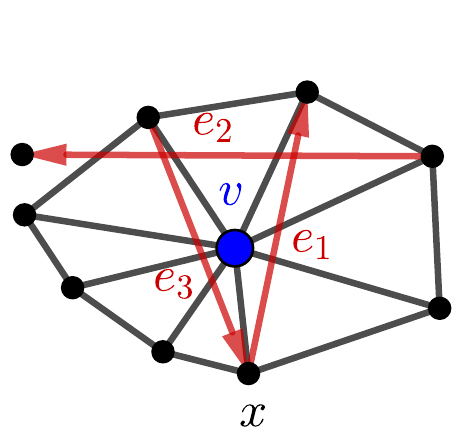}
\captionsetup{textformat=simple}
\subcaption{$e_1$ and $e_3$ share an endpoint left of $e_2$, such that $v$ is in the interior of the intersections}
\end{subfigure}\hfill
\begin{subfigure}[t]{0.23\textwidth}
\centering
\includegraphics[width =1\textwidth ]{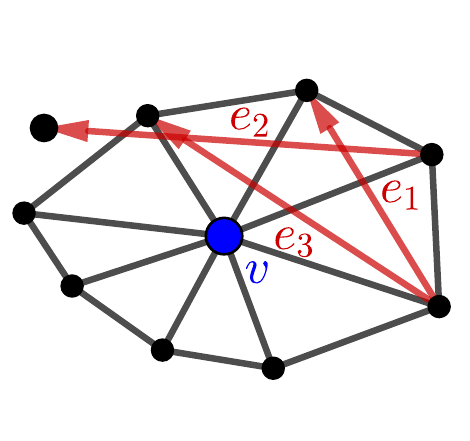}
\captionsetup{textformat=simple}
\subcaption{$e_1$ and $e_3$ share an endpoint left of $e_2$, such that $v$ is not in the interior of the intersections}
\end{subfigure}
\caption{The different cases of three consecutive edges $e_1,e_2,e_3\in E$. Note that case (d) cannot happen, since the redundant edge $e_1$ cannot be in $E$.}\label{fig:type21}
\end{figure}
\begin{description}
	\item[Case 1:] Let all edges not share any endpoints. Since $(e_1,e_2)$ is a Type-2 pair, the right defining circle of  $e_2$ must contain an endpoint of $e_1$ and since
 $(e_2,e_3)$ is a Type-1 pair the right defining circle of $e_2$, must contain an endpoint of $e_3$. Additionally, the right defining circle must contain the point $v$, since $e_2$ is a separation edge. Thus, the right defining circle contains three points and $e_2$ cannot be useful be Lemma \ref{lemma:useful}.
 
 	\item[Case 2:] Let the edges $e_1=\overline{zx}$ and $e_3=\overline{xy}$ share the endpoint $x$ right of $e_2$. Then the same argument as in Case 1 can be applied and the right defining circle of $e_2$ must contain $z,y$ and $v$ and cannot be useful.
 	
 	\item[Case 3:] Let $e_1$ and $e_3$ share the endpoint $x$ left of $e_2$. We have two sub-cases: The vertex $v$ can be outside of the intersection of the edges $e_1,e_2,e_3$ (Fig.~\ref{fig:type21} (d)) or it can be inside of the intersection (Fig.~\ref{fig:type21} (c)).
 	If $v$ is outside of the intersection, then one of the edges must be redundant (In Figure~\ref{fig:type21} (d) this would be $e_1$). Thus, this case cannot happen, if $E$ is a minimal set.

 	Consequently, we only further investigate the case where $v$ is in the interior of the intersection. Figure~\ref{fig:type21_special} depicts this case with only the involved edges and points. The edges of the two pairs are $e_1=\overline{xb}$, $e_2=\overline{zy}$ and $e_3=\overline{ax}$. Since the connecting edge $\overline{vx}$ cannot be intersected by any of the separation edges $e_1,e_2,e_3$, exactly one additional separation edge $e_4=\overline{uw}$ must exist in $E$.

 	Next, we need to discuss how $e_4$ may be positioned with respect to $e_3$ and $e_1$. Let $u_1,...,u_k$ be the endpoints of the connecting edges that $e_3$ intersects in counter-clockwise order. If we assume that $a$ is right of $e_4$, it follows that $e_4$ must also be a separation edge for $u_1,...,u_k$ and at least the additional connecting edge $\overline{xv}$. Thus, the edge $e_3$ would be redundant, if $a$ is right of $e_4$. In the same way it follows that $b$ cannot be right of $e_4$. Note that by the same arguments  $e_4$ cannot have $a$ or $b$ as an endpoint. Thus $a$, $b$ and also $v$ must be left of $e_4$. 
 \begin{figure}[!tb]
\centering
\includegraphics[width=.33\textwidth]{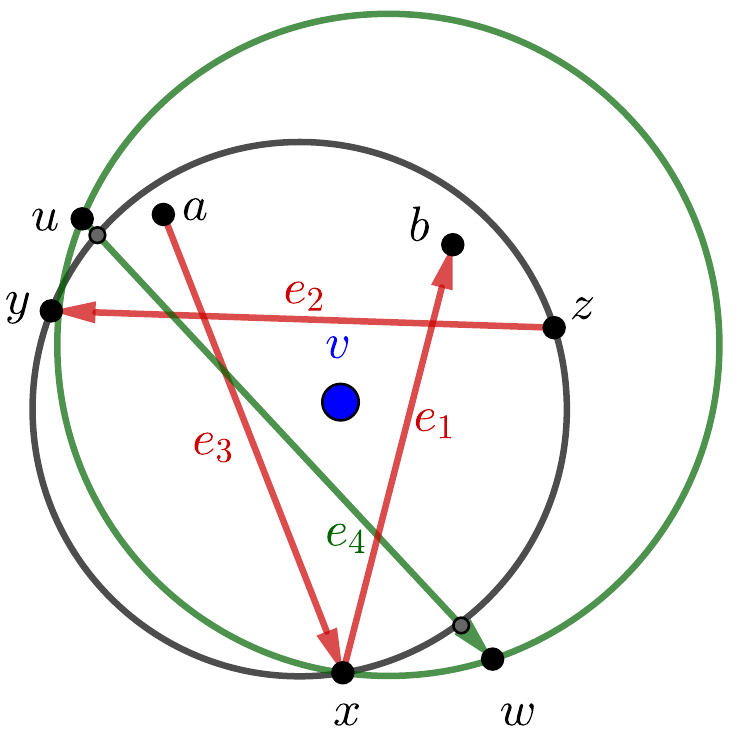}
\caption{A Type-2 pair followed by a Type-1 pair that share an endpoint left of $e_1$}\label{fig:type21_special}
\end{figure} 

We now show that either $e_2$ or $e_4$ cannot be useful. Let $s_l$ and $s_r$ denote the defining points of $e_2$. We know that $\ca^{zy}_{s_r}$ contains $x$ because of the Type-2 pair $(e_1,e_2)$. Additionally, we know that $v$ is in $\ca^{zy}_{s_r}$, since $\overline{zy}$ is a separation edge. Thus, $\ca^{zy}_{s_r}$ contains two points. Since $(e_1,e_2)$ is a Type-2 pair, we know that $\ca^{zy}_{s_l}$ contains $b$ and since $(e_2,e_3)$ is a Type-1 pair, we know that $\ca^{zy}_{s_l}$ contains $a$. Consequently, both defining circles of $e_2$ contain at least two points. This implies that, if either $u$ or $w$ is inside of $C(x,y,z)$, then $e_2$ cannot be useful, since $\ca^{zy}_{x}\subset \ca^{zy}_{s_l}$ and  $\ca^{zy}_{a}\subset\ca^{zy}_{s_r}$ (and the part of the circle $C(x,y,z)$ that is left of $e_2$ is contained in $\ca^{zy}_{a}$, since $a$ is inside the circle $C(x,y,z)$).

	Now we assume that $e_2$ is useful. This implies that $u$ as well as $w$ are outside of $C(x,y,z)$. Consequently, the circle $C(x,y,z)$ must be intersected by $e_4$ twice and $x$ must be right of $e_4$, since $e_4$ is a separation edge for $\overline{xv}$. Thus, the circle $C(x,u,w)$ must contain all of $C(x,y,z)$ that is left of $e_4$; see Figure~\ref{fig:type21_special}. We previously argued that $a,b$ and $v$ are left of $e_4$. Additionally, we know that $C(x,y,z)$ contains $a,b$ and $v$. Thus, the right defining circle of $e_4$ contains the three points and $e_4$ cannot be useful.
\end{description}
\end{proof}

\claimfunf*
\begin{proof}
Let $e_1=\overline{uw}, e_2=\overline{xy}$ and $e_3=\overline{zs}$ be three consecutive separation edges in $E$ such that $(e_1,e_2)$ as well as $(e_2,e_3)$ are Type-2 pairs. The separation edges $e_1$ and $e_3$ must have one endpoint left and one right of $e_2$, because of Claim \ref{c3}. Let $z$ and $w$ be the two endpoints that are right of $e_2$. We  now discuss all possible positions of edges and endpoints with respect to each other and show that in all cases one of the edges cannot be a useful separation edge. 

We have four major distinctions:
\begin{description}
	\item [Case 1] No edges share endpoints and the endpoint $s$ of $\overline{zs}$ is left of $\overline{uw}$, i.e., $e_1$ and $e_3$ do not intersect left of $e_2$ (Fig.~\ref{fig:C5Case1}).
	\item [Case 2] No edges share endpoints and the endpoint $s$ of $\overline{zs}$ is right of $\overline{uw}$, i.e.,  $e_1$ and $e_3$ intersect left of $e_2$ (Fig.~\ref{fig:C5Case2}).
	\item [Case 3] The edges $\overline{uw}$ and $\overline{zs}$ share one endpoint right of $e_2$ (Fig.~\ref{fig:C5Case3}).
	\item [Case 4] The edges $\overline{uw}$ and $\overline{zs}$ share one endpoint left of $e_2$ (Fig.~\ref{fig:C5Case4}).
	
\end{description}

We say $z$ is disk-closer to $e_2$ than $w$, if $\ca_{w}^{xy}\subset\ca_{z}^{xy}$. Without loss of generality we can assume that the right defining circle of $e_2$ is given by the disk-closer of the two endpoints that are right of $e_2$. If this was not the case, there would be an additional point $s_r$ that is disk-closer to $e_2$. The circle given by $C(x,y,s_r)$ must contain all of $\ca_{w}^{xy}$($\ca_{z}^{xy}$), because it is defining. Thus, it must contain all points that are left of $e_2$ and in $\ca_{w}^{xy}$($\ca_{z}^{xy}$). Consequently, if there are three points in the right defining circle while ignoring $s_r$, there must also be three or more points in the right defining circle, if we also consider $s_r$.

For all cases we must consider four sub-cases: (a) $w$ is disk-closer to $e_2$ than $z$ and $s$ is in the right defining circle of $e_2$, (b) $w$ is disk-closer to $e_2$ than $z$ and $s$ is not in the right defining circle of $e_2$, (c) $z$ is disk-closer to $e_2$ than $w$ and $s$ is not in the right defining circle of $e_2$ and (d) $z$ is disk-closer to $e_2$ than $w$ and $s$ is in the right defining circle of $e_2$.

The sub-case (d) is only mentioned for the sake of completeness. This case cannot happen, since the right defining circle is $C(x,y,z)$ and $(e_2,e_3)$ is a Type-2 pair which implies that $s$ cannot be in $C(x,y,z)$.

Next, we show that in all of the different cases one of the separation edges cannot be useful.
\begin{description}
	\item [Case 1] \textit{No edges share endpoints and the endpoint $s$ of $\overline{zs}$ is left of $\overline{uw}$.} 
	\begin{description}
		
\begin{figure}[!t]
\begin{subfigure}[t]{0.3\textwidth}
\centering
\includegraphics[width =1\textwidth ]{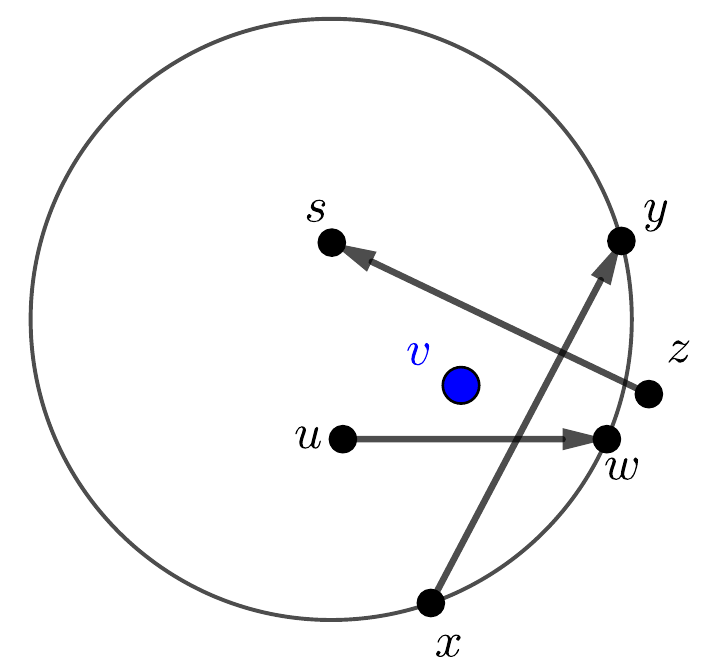}
\subcaption{$w$ is disk-closer to $e_2=\overline{xy}$ and $s$ is in the right defining circle of $e_2$}
\end{subfigure}\hfill
\begin{subfigure}[t]{0.3\textwidth}
\centering
\includegraphics[width =1\textwidth ]{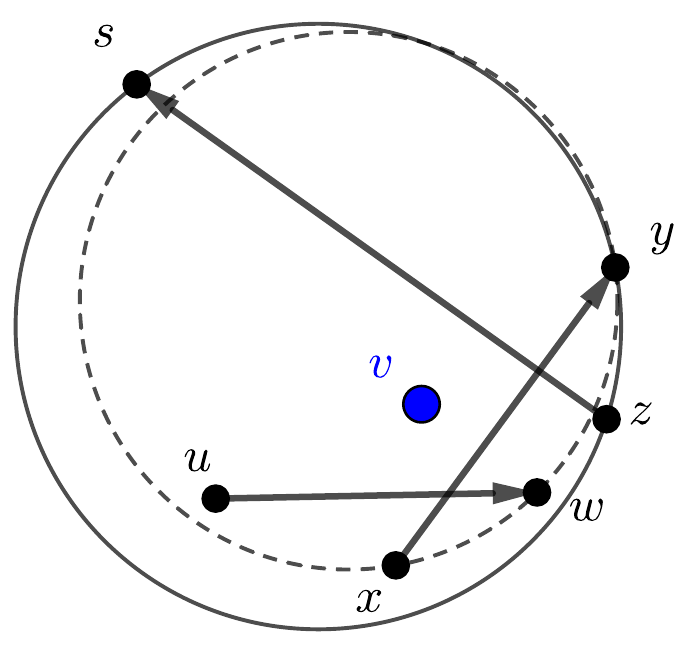}
\subcaption{$w$ is disk-closer to $e_2=\overline{xy}$ and $s$ is not in the right defining circle of $e_2$}
\end{subfigure}\hfill
\begin{subfigure}[t]{0.3\textwidth}
\centering
\includegraphics[width =1\textwidth ]{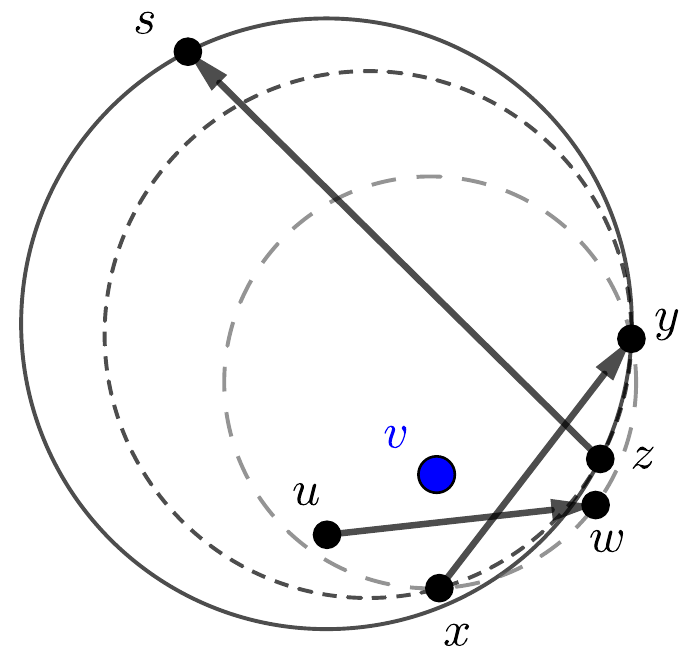}
\subcaption{$z$ is disk-closer to $e_2=\overline{xy}$ and $s$ is not in the right defining circle of $e_2$}
\end{subfigure}
\caption{Two consecutive Type-2 pairs. No edges share endpoints and the endpoint $s$ of $e_3=\overline{zs}$ is left of $e_1=\overline{uw}$.}\label{fig:C5Case1}
\end{figure}
	
		\item[Case 1.a]  \textit{$w$ is disk-closer to $e_2$ and $s$ is in the right defining circle of $e_2$ (Fig. \ref{fig:C5Case1}~(a)).}\\
		We know that $\ca_w^{xy}$ contains $u$, since $(e_2,e_1)$ is a Type-1 pair. We know that $\ca_w^{xy}$ contains $v$, since $e_2=\overline{xy}$ is a separation edge and $w$ is right 			of $e_2$. Lastly we know by assumption that $s$ is inside $\ca_w^{xy}$. Thus, the right defining circle contains three points and $e_2$ cannot be useful.
		
		\item[Case 1.b] \textit{$w$ is disk-closer to $e_2$ and $s$ is not in the right defining circle of $e_2$ (Fig. \ref{fig:C5Case1}~(b)).}\\
		We know that $\ca_w^{xy}$ contains $u$, since $(e_2,e_1)$ is a Type-1 pair. We know that $\ca_y^{zs}$ contains $x$, since $(e_3,e_2)$ is a Type-1 pair. Since $w$ is 			disk-closer to $e_2=\overline{xy}$ and $s$ is outside of the right defining circle, $e_3=\overline{zs}$ must intersect $C(x,y,w)$ twice. Thus, $C(s,z,y)$ must contain all of $C(x,w,y)$ that is left 			of $e_3$. By assumption $u$ is left of $e_3$ thus we overall get that the right defining circle of $e_3$ must contain $u$, $x$ and additionally $v$, since 				$e_3$ is a separation edge. Consequently, $e_3$ cannot	be useful.  
		
		\item[Case 1.c] \textit{$z$ is disk-closer to $e_2$ and $s$ is not in the right defining circle of $e_2$ (Fig. \ref{fig:C5Case1}~(c)).}\\
		We know that $\ca_w^{xy}$ contains $u$, since $(e_2,e_1)$ is a Type-1 pair. Since $z$ is disk-closer to $e_2=\overline{xy}$ than $w$ we know that $\ca_z^{xy}$ must also contain $u$.
		We again know that $\ca_y^{zs}$ contains $x$, since $(e_3,e_2)$ is a Type-2 pair. $C(s,y,z)$ shares two defining points with $C(x,y,z)$ and we know that the endpoint 			$s$ of $e_2$ is outside of $C(x,y,z)$. Hence, by the same arguments as in Case 1.b  we get that $x,u$ and $v$ must be in the right defining circle of $e_3=\overline{zs}$ and $e_3$ 			cannot be useful.
	\end{description}

\begin{figure}[!t]
\begin{subfigure}[t]{0.3\textwidth}
\centering
\includegraphics[width =1\textwidth ]{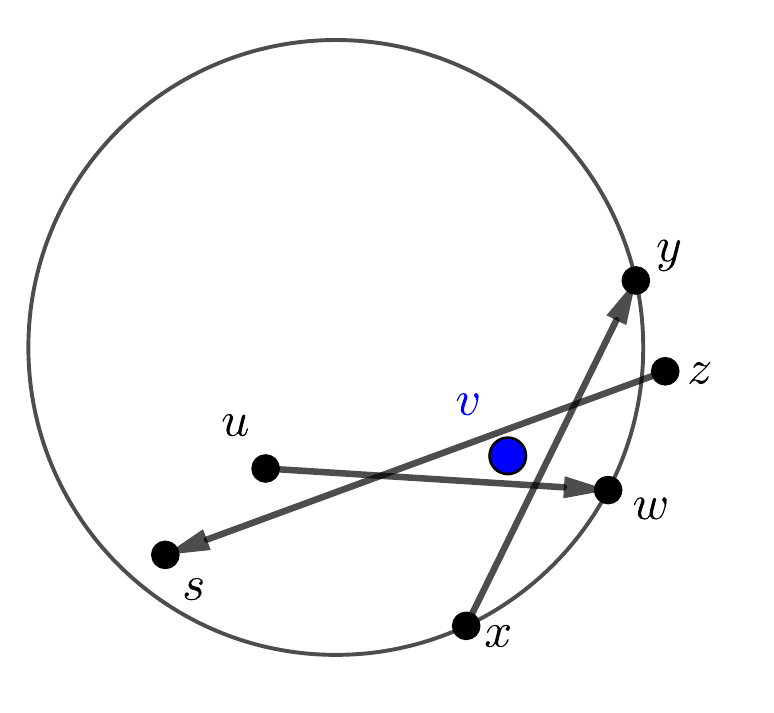}
\subcaption{$w$ is disk-closer to $e_2=\overline{xy}$ and $s$ is in the right defining circle of $e_2$}
\end{subfigure}\hfill
\begin{subfigure}[t]{0.3\textwidth}
\centering
\includegraphics[width =1\textwidth ]{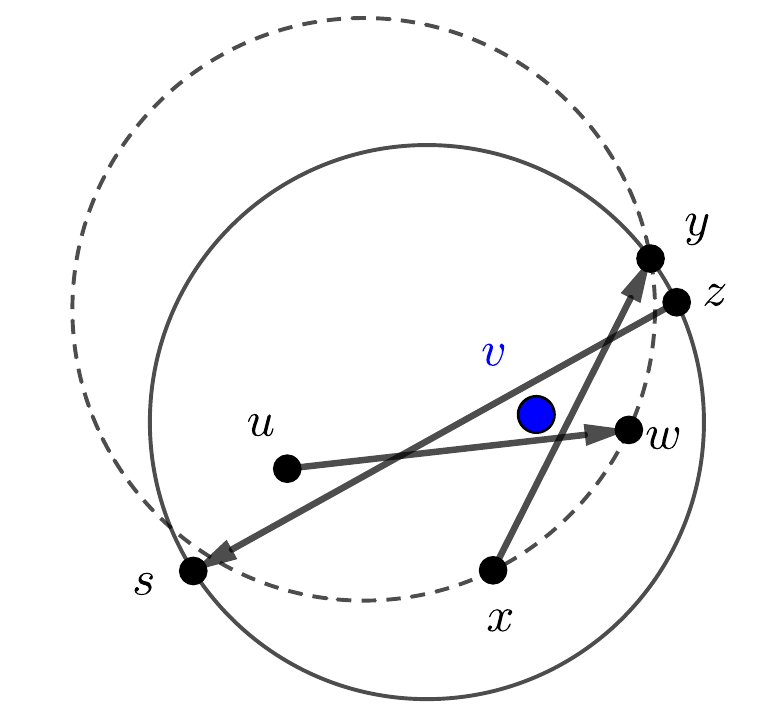}
\subcaption{$w$ is disk-closer to $e_2=\overline{xy}$ and $s$ is not in the right defining circle of $e_2$}
\end{subfigure}\hfill
\begin{subfigure}[t]{0.3\textwidth}
\centering
\includegraphics[width =1\textwidth ]{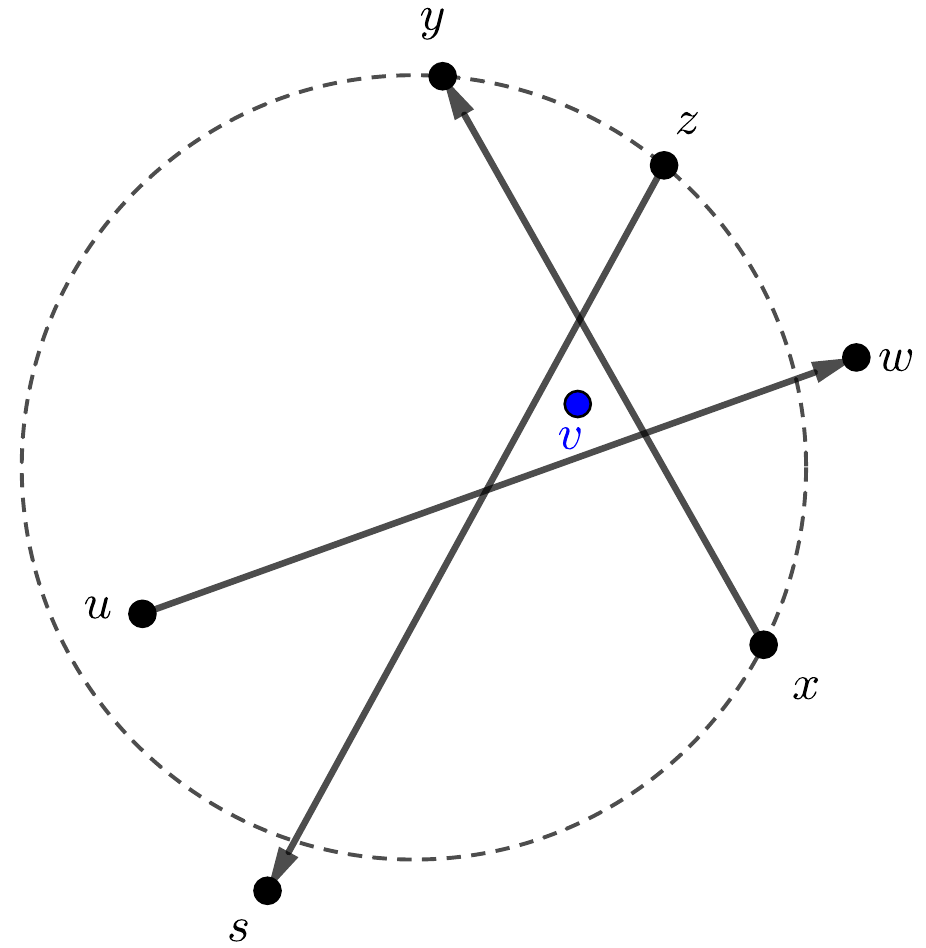}
\subcaption{$z$ is disk-closer to $e_2=\overline{xy}$ and $s$ is not in the right defining circle of $e_2$}
\end{subfigure}
\caption{Two consecutive Type-2 pairs. No edges share endpoints and the endpoint $s$ of $e_3=\overline{zs}$ is right of $e_1=\overline{uw}$.}\label{fig:C5Case2}
\end{figure}	

	\item [Case 2] \textit{No edges share endpoints and the endpoint $s$ of $\overline{zs}$ is right of $\overline{uw}$.} 
	\begin{description}
		\item[Case 2.a]  \textit{$w$ is disk-closer to $e_2$ and $s$ is in the right defining circle of $e_2$ (Fig. \ref{fig:C5Case2}~(a)).}\\
		This case can be handled exactly like Case 1.a. Consequently, $e_2=\overline{xy}$ cannot be useful.
		
		\item[Case 2.b] \textit{$w$ is disk-closer to $e_2$ and $s$ is not in the right defining circle of $e_2$ (Fig.~\ref{fig:C5Case2}~(b)).}\\
		The circle $C(x,w,y)$ is intersected by $\overline{zs}$ twice, since $z$ and $s$ are outside of it by assumption. Thus, extending it while fixing $y$ right of $e_3=\overline{zs}$ only makes the region left of $e_3$ bigger. Thus, $C(s,z,y)$ must contain both $x$ and  $w$ and both of them must be left of $e_3$ by assumption. Additionally, $v$ must 		be in the right defining circle of $e_3$, since $e_3$ is a separation edge. Thus, the right defining circle of $e_3$ must contain three points and $e_3$ cannot be 			useful.
		
		\item[Case 2.c] \textit{$z$ is disk-closer to $e_2$ and $s$ is not in the right defining circle of $e_2$ (Fig. \ref{fig:C5Case2}~(c)).}\\
		We show that this case cannot happen, i.e., there does not exist a set of edges such that the edge pairs have the correct 			types.
		
		We know that $E$ only consists of the three edges $e_1,e_2$ and $e_3$. This implies that $(e_3,e_1)$ must also be a Type-2 pair by assumption. We now show that this cannot be the case. We use a different argument to the previous cases. Let $e_2=\overline{xy}$, $e_3=\overline{zs}$ and also the point $w$ right of $e_2$ be fixed (note that $z$ must be disk-closer to $e_2$ than $w$). Let $m$ be the intersection point of $e_2$ and $e_3$. This is depicted in Figure \ref{fig:SpecialCase} (a). Since $v$ is left of $e_2$ and $e_3$, it must be in the cone given by $\overrightarrow{mx}$ and $\overrightarrow{ms}$. Thus, $e_1=\overline{uw}$ must intersect $e_2$ on the segment $\overline{mx}$ and $e_3$ on the segment $\overline{ms}$.\\
		We now show that this is not possible, if we assume that all of the separation edges are useful and all pairs of consecutive edges are of type 2. First, we argue where the second endpoint $u$ of $e_1$ may be positioned such that $(e_1,e_2)$ and $(e_3,e_1)$ can be Type-2 pairs.
		
	\noindent Since $(e_1,e_2)$ is a Type-2 pair, $(e_2,e_1)$ must be a Type-1 pair and $u$ must be inside of the circle $C(x,y,w)$. Additionally, $u$ must be outside the circle $C(s,z,w)$, since $(e_3,e_1)$ is a Type-2 pair. Moreover, the circles must intersect, since $z$ is inside of $C(x,y,w)$ and $s$ is outside of $C(x,y,w)$.  Lastly, $u$ must  be left of $e_2$, since $w$ is right of $e_2$. Thus, $u$ must be in the region left of $e_2$, inside of $C(x,y,w)$ and outside of $C(s,z,w)$. This region may be empty, but then we are done, since there does not even exist a candidate, that results in a valid edge $e_1$, i.e., an edge $e_1$ such that $(e_1,e_2)$ and $(e_3,e_1)$ are Type-2 pairs. The region is depicted in gray in Figure \ref{fig:SpecialCase} (a). \\
	\begin{figure}[!tb]
\begin{subfigure}[t]{0.45\textwidth}
\centering
\includegraphics[width =.72\textwidth ]{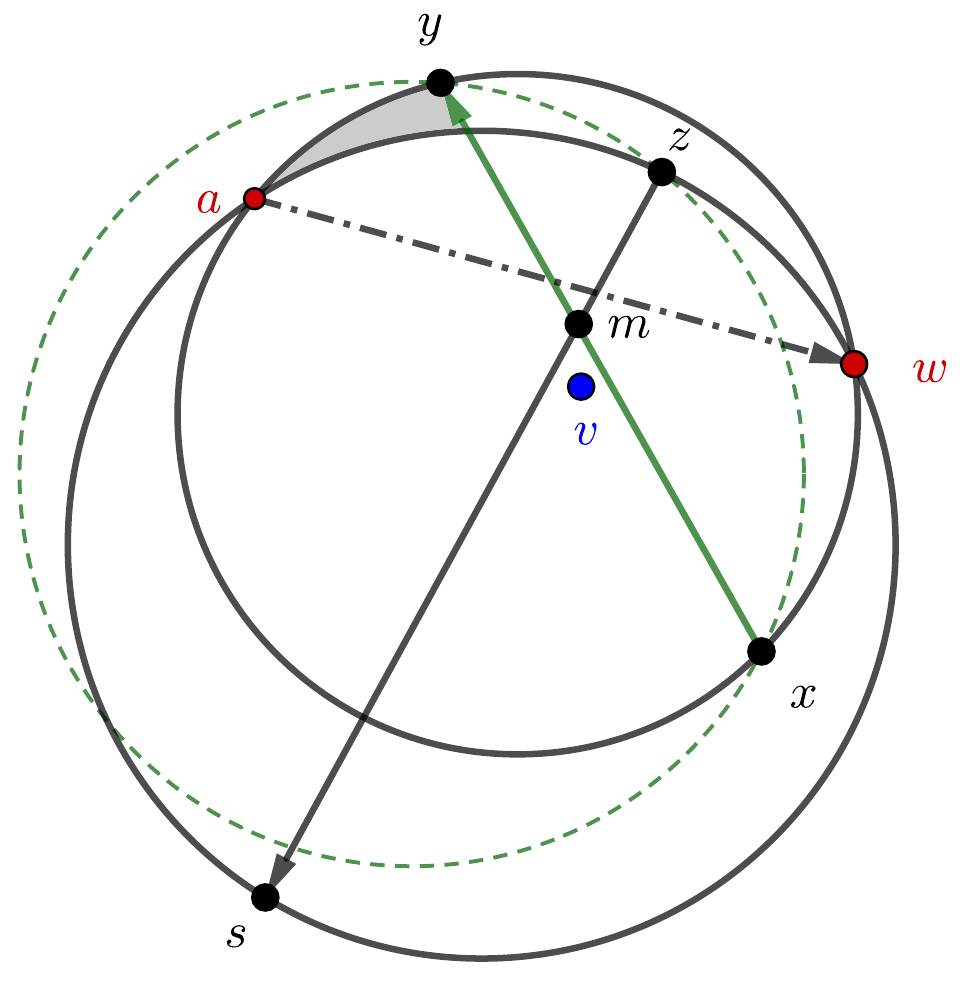}
\subcaption{in gray the valid region for the endpoint $u$ of $e_1=\overline{uw}$ and the best candidate $a$}
\end{subfigure}\hfill
\begin{subfigure}[t]{0.45\textwidth}
\centering
\includegraphics[width =.72\textwidth ]{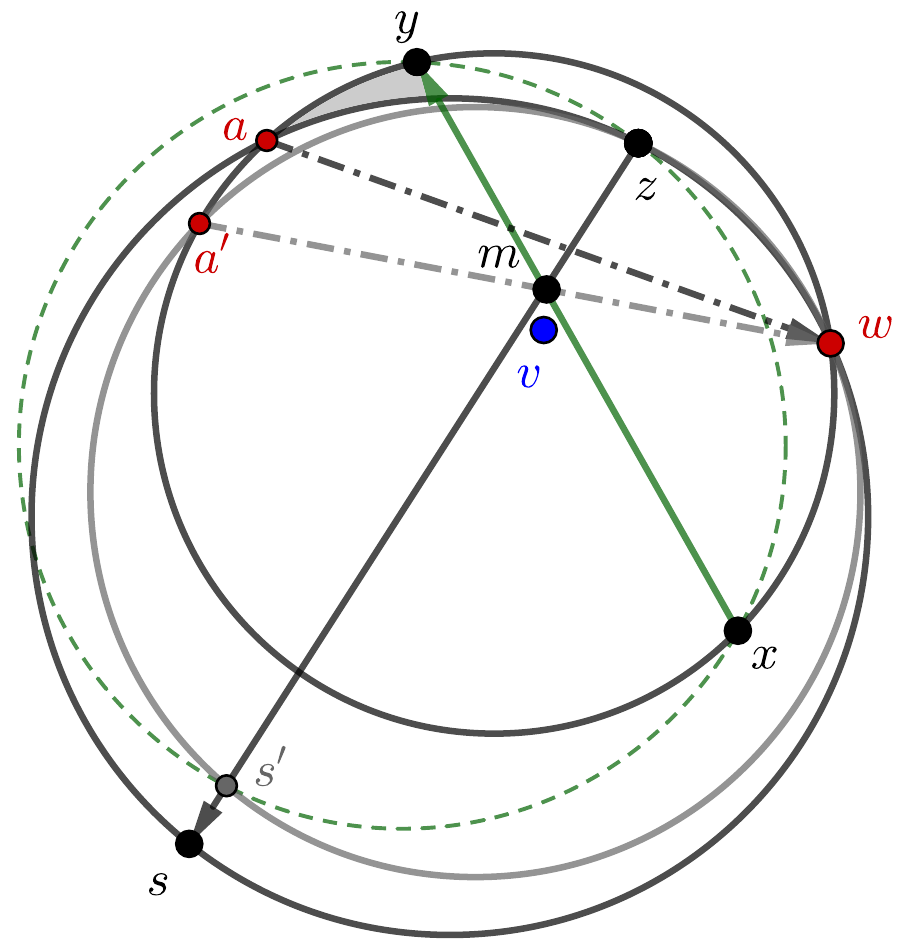}
\subcaption{the additional gray circle results in one intersection point $m$ for all edges}
\end{subfigure}
\caption{Construction of the region in which the endpoint $u$ of the edge $e_1$ may be positioned}\label{fig:SpecialCase}
\end{figure}
		The most promising candidate with respect to intersecting the correct segments of $e_2$ and $e_3$, is the point in the region with the furthest (angular) distance to $y$.  This point is exactly the intersection point $a$ of the two circles that define the region. Note that $a$ is actually an illegal endpoint, since it violates the general position assumption (also note that $\overline{aw}$ would not yield any Types of pairs with the other edges, since it is essentially the configuration where we would switch from Type-1 to Type-2 pairs and vice versa). Nevertheless, it is sufficient to show that $\overline{aw}$ cannot intersect $\overline{ms}$ and $\overline{mx}$, since all other points yield edges that are even worse with respect to  the intersection of the segments.\\
		It remains to show that $\overline{aw}$ cannot intersect $\overline{ms}$ and $\overline{mx}$. To prove this we consider an additional circle. Let $s'$ be the intersection point of $\overline{zs}$ and $C(x,y,z)$. This point exists, since $(e_2,e_3)$ is a Type-2 pair. Then we can define the circle $C(s',w,z)$. This circle must also intersect $C(x,y,w)$ in a point $a'$ by the same argument as before. Note that the circles $C(s',w,z)$, $C(x,y,w)$ and $C(x,y,z)$ all pairwise intersect. The intersection points of $C(s',w,z)$ with $C(x,y,w)$ correspond to $\overline{a'w}$, the intersection points of $C(s',w,z)$ with $C(x,y,z)$ correspond to $\overline{zs'}$ and the intersection points of $C(x,y,w)$ with $C(x,y,z)$ correspond to $\overline{xy}$. Lemma \ref{lemma:chord} implies that the three edges all must intersect in a single point $m$. Since $\overline{zs'}$ is contained in $\overline{zs}$, this intersection property also holds with respect to $\overline{zs}$. This is depicted in Figure \ref{fig:SpecialCase} (b).\\
		Obviously $m$ must be the same intersection point that was already given by the intersection of $e_2$ and $e_3$. Thus, the only thing that remains to be shown, is that $a$ is left of $\overline{a'w}$, since this implies that $a$ cannot intersect $\overline{ms}$ and $\overline{mx}$. We know that $C(s',w,z)$ and $C(s,w,z)$ share two endpoints and $C(s',w,z)$ can be transformed to $C(s,w,z)$ by loosening $s'$ and extending the circle until it hits $s$. Thus, the circle gets larger. It follows that $C(s,w,z)$ intersects $C(y,x,w)$ closer to $y$  than $C(s',w,z)$ which implies that $a$ must be left of $\overline{a'w}$. Thus all possible candidates for $u$ do not yield an edge $\overline{uw}$ that is correctly positioned, i.e., positioned such that it intersects $\overline{ms}$ and $\overline{mx}$. Hence, one of the three pairs $(e_1,e_2),(e_2,e_3)$ and $(e_3,e_1)$ cannot be of type 2. Thus, this case cannot happen.			
	\end{description}

\begin{figure}[t]
\begin{subfigure}[t]{0.45\textwidth}
\centering
\includegraphics[width =.6\textwidth ]{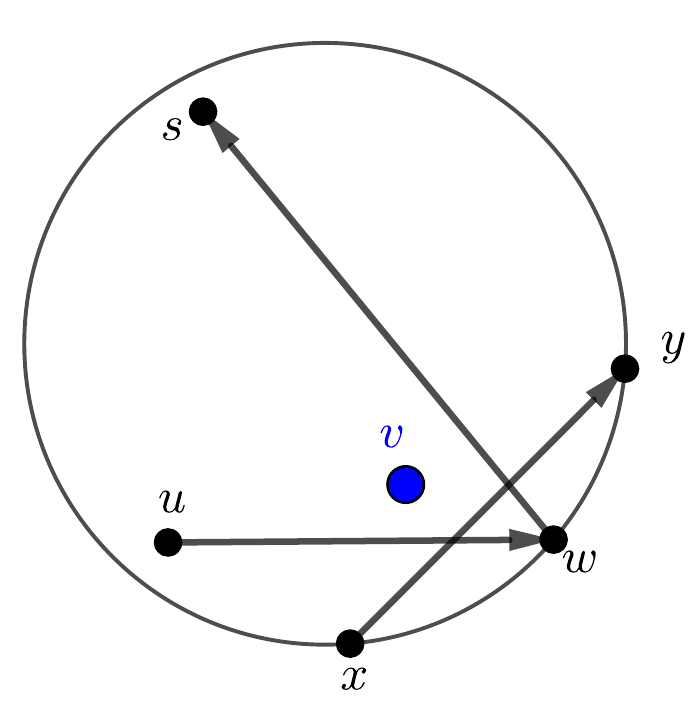}
\subcaption{$s$ is in the right defining circle of $e_2=\overline{xy}$; note that $(e_2,e_3)$ cannot be a Type-2 pair.}
\end{subfigure}\hfill
\begin{subfigure}[t]{0.45\textwidth}
\centering
\includegraphics[width =.6\textwidth ]{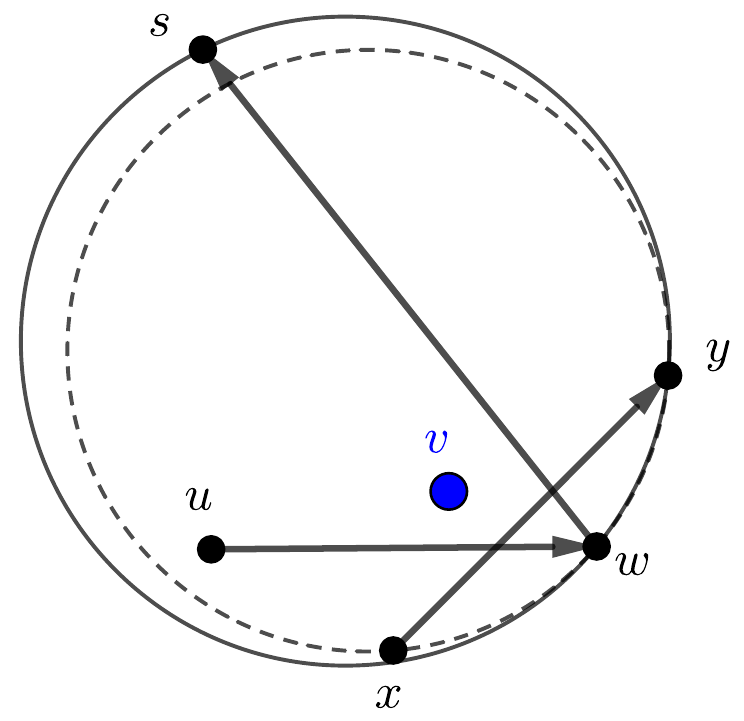}
\subcaption{$s$ is not in the right defining circle of $e_2=\overline{xy}$}
\end{subfigure}
\caption{Two consecutive Type-2 pairs. The separation edges $e_1=\overline{uw}$ and $e_3=\overline{ws}$ share one endpoint right of $e_2$}\label{fig:C5Case3}
\end{figure}		
	
	\item [Case 3] \textit{$\overline{uw}$ and $\overline{zs}$ share one endpoint right of $e_2$.} \\
	Note that in this case $w$ and $z$ ``have the same distance to $e_2$'', since they are identical.
	\begin{description}
		\item[Case 3.a]  \textit{$s$ is in the right defining circle of $e_2=\overline{xy}$ (Fig. \ref{fig:C5Case3}~(a)).}\\
		This case cannot happen, since the right defining circle is given by $C(w,x,y)$ and $(e_2,e_3)$ is a Type-2 pair which implies that $s$ is outside of $C(w,x,y)$.
		\item[Case 3.b] \textit{$s$ is not in the right defining circle of $e_2$ (Fig. \ref{fig:C5Case3}~(b)).}\\
		We can handle this case exactly like like Case 1.c. Thus, $e_3=\overline{ws}$ cannot be useful.	
	\end{description}

\begin{figure}[t]
\begin{subfigure}[t]{0.45\textwidth}
\centering
\includegraphics[width =.6\textwidth ]{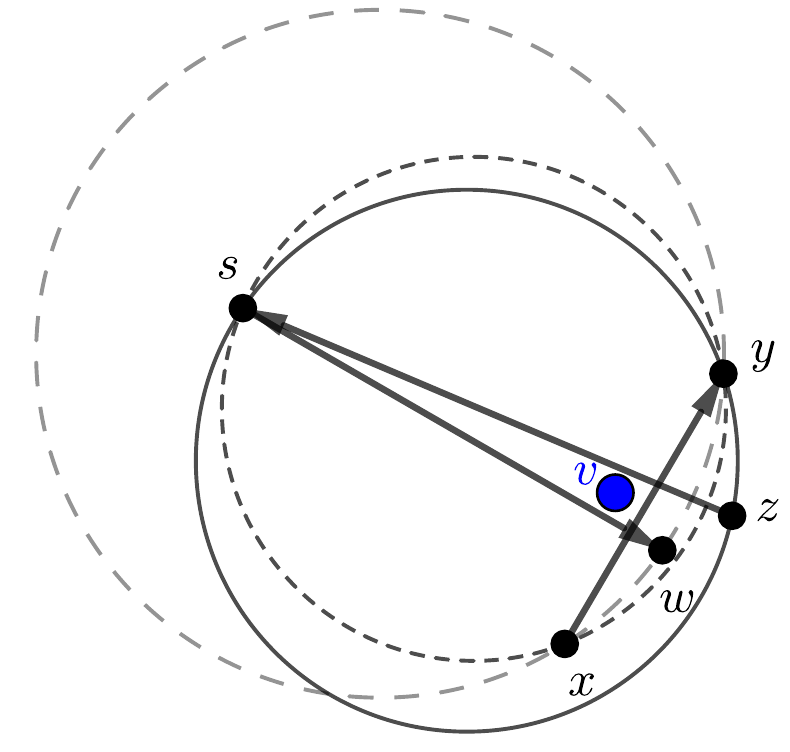}
\subcaption{$w$ is disk-closer to $e_2=\overline{xy}$ and $s$ is in the right defining circle of $e_2$}
\end{subfigure}\hfill
\begin{subfigure}[t]{0.45\textwidth}
\centering
\includegraphics[width =.6\textwidth ]{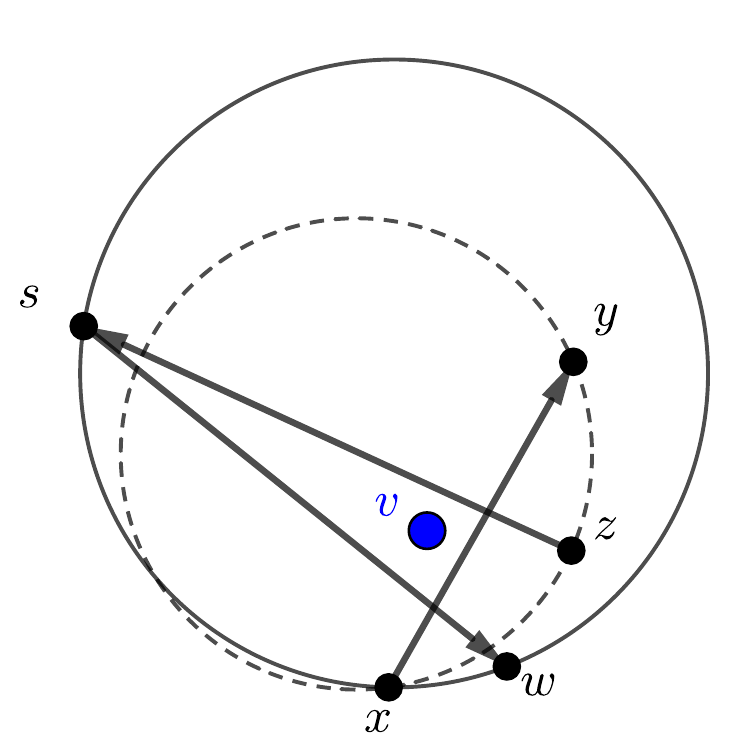}
\subcaption{$z$ is disk-closer to $e_2=\overline{xy}$ and $s$ is not in the right defining circle of $e_2$; note that $(e_1,e_2)$ cannot be a Type-2 pair.}
\end{subfigure}
\caption{Two consecutive Type-2 pairs.The separation edges $e_1=\overline{sw}$ and $e_3=\overline{zs}$ share one endpoint left of $e_2$, i.e., $u=s$}\label{fig:C5Case4}
\end{figure}	
	
	\item [Case 4] \textit{$\overline{uw}$ and $\overline{zs}$ share one endpoint left of $e_2$, i.e., $u=s$.} 
	\begin{description}
		\item[Case 4.a]  \textit{$w$ is disk-closer to $e_2$ and $s$ is in the right defining circle of $e_2$ (Fig. \ref{fig:C5Case4}~(a)).}\\
		We know that $C(s,x,y)$ contains $w$ because $(e_2,e_1)$ is a Type-1 pair. Additionally, $z$ is outside of $C(s,x,y)$ because $(e_2,e_3)$ is a Type-2 pair. Thus, we can extend the circle $C(s,x,y)$ while fixing $s$ and $y$ until it hits $z$ and becomes $C(s,x,z)$. This circle must contain all of $C(s,x,y)$ that is left of 			$e_3=\overline{zs}$ and also the point $x$. Consequently, it contains $w$ and $x$ and also $v$, since $e_3$ is a separation edge. Hence, $e_3$ cannot be useful.
		\item[Case 4.b] \textit{$w$ is disk-closer to $e_2$ and $s$ is not in the right defining circle of $e_2$ .}\\
		In this case $(e_2,e_1)$ cannot be a Type-1 pair, since $C(x,y,w)$ does not contain $s$. This implies that $(e_1,e_2)$ cannot be a Type-2 pair. Consequently, this case cannot happen.
				
		\item[Case 4.c] \textit{$z$ is disk-closer to $e_2$ and $s$ is not in the right defining circle of $e_2$ (Fig. \ref{fig:C5Case4}~(b)).}\\
		$C(x,y,z)$ contains neither $s$, since $(e_2,e_3)$ is Type-2 pair nor $w$ (by assumption). Thus, $e_1=\overline{sw}$ intersects $C(x,y,z)$ twice and can be extended to be $C(x,s,w)$. Consequently, $C(x,s,w)$ must contain $y$ which is left of $e_1$. Thus, $(e_1,e_2)$ could not have been a Type-2 pair. Hence, this case is also not possible.
	\end{description}	
\end{description}
All in all we have shown that for two consecutive Type-2 pairs not all separation edges can be useful.
\end{proof}
\end{document}